\documentclass[12pt,oneside,peerreview,draftcls,onecolumn]{IEEEtran}

\usepackage{amsmath,amsfonts,amssymb,amsbsy, amsthm,ae,aecompl}
\usepackage{algorithm,algpseudocode}
\usepackage[utf8]{inputenc}
\usepackage[english]{babel}
\usepackage{bm}
\usepackage{color}
\usepackage[noadjust]{cite}
\usepackage{epsfig}
\usepackage{enumerate}
\usepackage{float} 
\usepackage{fancyhdr}
\usepackage[T1]{fontenc}
\usepackage[acronym,toc,shortcuts]{glossaries}
\usepackage{graphicx, caption, subcaption}
\usepackage{hyperref}
\usepackage{lastpage}
\usepackage{listings}
\usepackage{lipsum}
\usepackage{dsfont}
\usepackage{multirow,tabularx}
\usepackage[normalem]{ulem}
\usepackage{tikz,pgfplots}
\usepackage{times}
\usepackage{verbatim}
\usepackage[all]{xy}
\usepackage{mathtools}
\usepackage{tikz}
\usepackage{tikz-qtree,tikz-qtree-compat}
\usepackage{pgfplots}

\usepackage{lipsum}
\usetikzlibrary{calc}
\usetikzlibrary{patterns}

\hypersetup{
    colorlinks,%
    citecolor=black,%
    filecolor=black,%
    linkcolor=black,%
    urlcolor=black
}

\pgfplotsset{
   grid style = {
   dash pattern = on 0.025mm off 0.95mm on 0.025mm off 0mm, 
   line cap = round,
   black,
   line width = 0.5pt
  },
  tick label style={font=\small},
  label style={font=\small},
  legend style={font=\footnotesize},
}

\pgfplotscreateplotcyclelist{laneasIACaching1}{
cyan!60!black,		solid, 	every mark/.append style={fill=cyan!60!black},mark=x\\%
cyan!60!black,		solid, 	every mark/.append style={fill=cyan!60!black},mark=+\\%
cyan!60!black,		solid, 	every mark/.append style={fill=cyan!60!black},mark=o\\%
red!80!black,       dashed\\%
cyan!60!black, 		densely dotted,	every mark/.append style={fill=cyan!80!black},mark=diamond*\\%
}

\pgfplotscreateplotcyclelist{laneasIACaching2}{
cyan!60!black,		solid, 	every mark/.append style={fill=cyan!60!black},mark=x\\%
lime!60!black,		solid, 	every mark/.append style={fill=lime!60!black},mark=+\\%
red!60!black,		solid, 	every mark/.append style={fill=cyan!60!black},mark=o\\%
}

\pgfplotscreateplotcyclelist{laneasIACaching3}{
cyan!60!black,		solid, \\%
cyan!60!black,		dashed\\%
lime!60!black,		solid, 	every mark/.append style={fill=cyan!60!black},mark=x\\%
lime!60!black,		dashed, 	every mark/.append style={fill=cyan!60!black},mark=x\\%
}

\pgfplotscreateplotcyclelist{thechairLivingOnTheEdge}{
cyan!60!black,		densely dotted, 	every mark/.append style={fill=cyan!80!black},mark=otimes*\\%
red,                 solid, every mark/.append style={fill=red!80!black},mark=otimes*\\%
cyan!60!black, 		densely dotted,	every mark/.append style={fill=cyan!80!black},mark=diamond*\\%
red,           		solid, every mark/.append style={solid,fill=red!80!black},mark=diamond*\\%
teal,				every mark/.append style={fill=teal!80!black},mark=*\\%
orange,				every mark/.append style={fill=orange!80!black},mark=square*\\%
red!70!white,		mark=star\\%
yellow!60!black,		densely dashed, every mark/.append style={solid,fill=yellow!80!black},mark=square*\\%
black,				every mark/.append style={solid,fill=gray},mark=otimes*\\%
blue,				densely dashed,mark=star,every mark/.append style=solid\\%
red,					densely dashed,every mark/.append style={solid,fill=red!80!black},mark=diamond*\\%
}

\graphicspath{{figures/}}

\makeglossaries
\newacronym{BS}{BS}{base station}
\newacronym{CDN}{CDN}{content delivery network}
\newacronym{CDF}{CDF}{cumulative distribution function}
\newacronym{CF}{CF}{collaborative filtering}
\newacronym{CN}{CN}{core network}
\newacronym{CRP}{CRP}{{C}hinese restaurant process}
\newacronym{CS}{CS}{central scheduler}
\newacronym{CSI}{CSI}{channel state information}
\newacronym{D2D}{D2D}{device-to-device}
\newacronym{DoF}{DoF}{degree-of-freedom}
\newacronym{HetNet}{HetNet}{heterogeneous network}
\newacronym{FDD}{FDD}{frequency-division duplex}
\newacronym{ICIC}{ICIC}{inter-cell interference coordination}
\newacronym{ICN}{ICN}{information-centric network}
\newacronym{IA}{IA}{interference alignment}
\newacronym{ISI}{ISI}{inter-stream interference}
\newacronym{IUI}{IUI}{inter-user interference}
\newacronym{LTE}{LTE}{long term evolution}
\newacronym{MIMO}{MIMO}{multiple-input multiple-output}
\newacronym{PPP}{PPP}{{P}oisson point process}
\newacronym{PHY}{PHY}{physical layer}
\newacronym{SBS}{SBS}{small base station}
\newacronym{SINR}{SINR}{signal-to-interference-plus-noise ratio}
\newacronym{SNR}{SNR}{signal-to-noise ratio}
\newacronym{SCN}{SCN}{small cell network}
\newacronym{SVD}{SVD}{singular value decomposition}
\newacronym{TL}{TL}{transfer learning}
\newacronym{TDD}{TDD}{time-division duplex}
\newacronym{UT}{UT}{user terminal}
\newacronym{QoS}{QoS}{quality-of-service}
\newacronym{QoE}{QoE}{quality-of-experience}
\newacronym{RAN}{RAN}{radio access network}
\newacronym{TDMA}{TDMA}{time division multiple access}
\newacronym{PDF}{PDF}{probability density function}
\newacronym{MGF}{MGF}{moment-generating function}
\newacronym{CC}{CC}{computational complexity} 
\newacronym{ZF}{ZF}{zero forcing} 

\newtheorem{defi}{Definition}
\newtheorem{theorem}{Theorem}
\newtheorem{lemma}{Lemma}
\newtheorem{proposition}{Proposition}

\begin{document}

\title{Queueing Stability and CSI Probing of a TDD Wireless Network with Interference Alignment }
\author{
		\IEEEauthorblockN{Matha Deghel, \emph{Student}, \emph{IEEE}, Mohamad Assaad, \emph{Senior}, \emph{IEEE}, \\ Mérouane Debbah, \emph{Fellow}, \emph{IEEE}, and Anthony Ephremides, \emph{Life Fellow, IEEE} }\\

		}

\maketitle
\let\thefootnote\relax\footnote{M. Deghel and M. Assaad are with Laboratoire de Signaux et Systèmes (L2S, CNRS, UMR8506) CentraleSupélec, 3 rue Joliot-Curie, 91192, Gif-sur-Yvette, cedex. France. (matha.deghel@centralesupelec.fr, mohamad.assaad@centralesupelec.fr).

M. Debbah is with the Large Systems and Networks Group (LANEAS), CentraleSupélec, Gif-sur-Yvette, France \\ (merouane.debbah@centralesupelec.fr) and also with the Mathematical and Algorithmic Sciences Lab, Huawei Technologies Co. Ltd., France (merouane.debbah@huawei.com).

A. Ephremides is with the Department of Electrical and Computer Engineering and Institute for Systems Research University of Maryland, College Park, MD 20742. (etony@umd.edu).

Parts of this paper have been presented at 
the IEEE International Symposium on Information Theory (ISIT), Hong Kong, 2015 \cite{Deghel2015StabilityIA}. }
\begin{abstract}
This paper characterizes the performance of interference alignment (IA) technique taking into account  the dynamic traffic pattern and the probing/feedback cost. We consider a time-division duplex (TDD) system where transmitters acquire their channel state information (CSI) by decoding the pilot sequences sent by the receivers. Since global CSI knowledge is required for IA, the transmitters have also to exchange their estimated CSIs over a backhaul of limited capacity (i.e. imperfect case). Under this setting, we characterize in this paper the stability region of the system under both the imperfect and perfect (i.e. unlimited backhaul) cases, then we examine the gap between these two resulting regions. Further, under each case, we provide a centralized probing algorithm (policy) that achieves the max stability region. These stability regions and scheduling policies are given for the symmetric system where all the path loss coefficients are equal to each other, as well as for the general system. For the symmetric system, we compare the stability region of IA with the one achieved by a time division multiple access (TDMA) system where each transmitter applies a simple singular value decomposition technique (SVD). We then propose a scheduling policy that consists in switching between these two techniques, leading the system, under some conditions, to achieve a bigger stability region. Under the general system, the adopted scheduling policy is of a high computational complexity for moderate number of pairs, consequently we propose an approximate policy that has a reduced complexity but that achieves only a fraction of the system stability region. A characterization of this fraction is provided. 
\end{abstract}
\section{Introduction}

One of the key issues in wireless communication systems is the interference that is caused by a large number of users communicating on the same channel, resulting
into severe performance degradations unless treated properly. In this regard, \ac{IA} was
introduced in \cite{Cadambe2008Interference} as an efficient interference management technique and is shown to result in higher throughputs compared to conventional
interference-agnostic methods. Indeed, \ac{IA} is a linear precoding technique that attempts to align interfering signals in time, frequency, or space. In \ac{MIMO} networks, \ac{IA} utilizes the spatial dimension offered by multiple antennas for alignment. By aligning interference at all receivers (users), \ac{IA} reduces the dimension of interference, allowing users to suppress interference via linear techniques and decode their desired signals interference free. 
However, the implementation of \ac{IA} in existing systems faces some challenges. A major disadvantage of the above \ac{IA} scheme lies in the fact that the global \ac{CSI} must be available at each transmitter, which weakens its application in practical systems, because \ac{CSI}, especially interference \ac{CSI}, is difficult to obtain at the transmitters.

In scenarios where the receivers quantize and send the
\ac{CSI} back to the transmitters, the \ac{IA} scheme is explored over
frequency selective channels for single-antenna users in \cite{Thukral2009LimIA} and for multiple-antenna users in \cite{Krishnamachari2010LimIA}. Both references provide \ac{DoF}-achieving quantization schemes and establish the required scaling of the number of feedback bits. For alignment using spatial dimensions, \cite{Rezaee2012LimIA} provides the scaling of feedback bits to achieve \ac{IA} in \ac{MIMO} interference channel (IC). For the broadcast channel, the scaling of the feedback bits was characterized in \cite{Jindal2006MIMO}. In \cite{Santipach2009Capacity}, quantization of the precoding matrix using random vector quantization (RVQ) codebooks is investigated, which provides insights on the asymptotic optimality of RVQ. To overcome the problem of scaling codebook size, and relax the reliance on frequency selectivity for quantization, \cite{Ayach2012Interference} proposed an analog feedback strategy for constant \ac{MIMO} interference
channels. From another point of view, \cite{Tresch2009Cellular} provides an analysis of the effect of imperfect CSI on the mutual information of the interference alignment scheme. 
On the other side, for \ac{TDD} systems, every transmitter can estimate its downlink channels from the uplink transmission phase thanks to reciprocity. However, for the \ac{IA} scheme, this local knowledge is not sufficient, and the transmitters need to share their channel estimates that can be carried out through backhaul links between transmitters. These links generally have a limited capacity, which should be exploited efficiently. For instance, in \cite{Park2013PerCloud} a compression scheme for the cloud radio access networks is proposed. In \cite{Rezaee2013CSIT}, the Grassmannian Manifold quantization technique was adopted to reduce the information exchange over the backhaul.
The above works on \ac{IA} and limited feedback do not take into account the dynamic traffic processes of the users, meaning that they assume users with infinite back-logged data. 

It is of great interest to investigate the impact of \ac{MIMO} in the higher layers \cite{Boche2006TheInterplay}, more specifically in the media access control (MAC) layer. The cross-layer design goal here is the achievement of the entire stability region of the system.
In broad terms, the stability region of a network is the set of arrival rate vectors such that the entire network load can be served by some service policy without an infinite blow up of any queue. The special scheduling policy achieving the entire stability region, called the stability-optimal policy (or simply optimal policy), is hereby of particular interest.
The concept of stability-optimal operation comes originally from the control and automation theory \cite{McKeown1999Achieving100,Kumar1996Duality,Leonardi2001Bounds,Szpankowski1993StabilityConditions}. It was applied to the wireless communication systems first in \cite{Tassiulas1992Stability}, and the view was extended by some bounds in \cite{Neely2003Power}. Since then, this concept has been investigated in the wireless framework under various traffic and network scenarios. For instance, in \cite{Boche2007Optimization}, the authors have presented a precoding strategy that achieves the system stability region, under the assumptions of perfect \ac{CSI} and use of Gaussian codebooks. This strategy is based on Lyapunov drift minimization given the queue lengths and channel states every timeslot.
Authors in \cite{Swannack2004Lowcomplexity} have considered the broadcast channel (BC) and proposed a technique based on \ac{ZF} precoding, with a heuristic user scheduling scheme that selects users whose channel states are nearly orthogonal vectors and illustrate the stability region this policy achieves via simulations.
In \cite{Kobayashi2006AnIterativeWater}, it has been noticed that the policy resulting from the minimization of the drift of a quadratic Lyapunov function is to solve a weighted sum rate maximization problem (with weights being the queue lengths) each timeslot and they propose an iterative water-filling algorithm for this purpose.
In addition, authors in \cite{Cheng2006OptimalDown} propose to use the delays of the packets in the head of each queue along with the queue lengths as weights.
All these works assume accurate CSI available at the transmitter. In the case of delayed channel state information and channels having a correlation in time, authors in \cite{Kobayashi2007TransmitDivers} compare the stability and delay performance of opportunistic beamforming and space time coding, while in \cite{Kobayashi2007JointBeamforming} they propose a user scheduling and precoding algorithm. Further, in \cite{Lau2012StabilityDelay}, the authors studied the impact of channel state quantization on the stability of a system using \ac{ZF} precoding under a centralized scheme where the transmitter selects the users to be scheduled based only on the queue lengths. However, in these works, the fact that radio resources i.e. time and/or spectrum are needed to acquire channel state information is not accounted for. For the case where the \ac{CSI} acquisition process consumes a fraction of the timeslot, the authors in \cite{Chaporkar2009ScedwithLimited} have explored the resulting trade-off between acquiring \ac{CSI} and exploiting channel diversity to the various receiver. In addition, taking into account the probing cost, the authors in \cite{Destounis2015Traffic-Aware} have examined three different scheduling policies (centralized, decentralized and mixed policies) for MISO wireless downlink systems under ZF precoding technique. It is worth noting that all the aforementioned works consider networks with a relatively simple physical layer (e.g. on-off channel, ZF, ...).

In this paper, we have a system with a more complicated physical layer. Specifically, we consider a Multipoint-to-Multipoint network where multiple transmitter-receiver pairs operate in \ac{TDD} mode and apply the IA technique under backhaul links of limited capacity. Each transmitter acquires its local \ac{CSI} from its corresponding user by exploiting the channel reciprocity. 
Indeed, there are two ways to perform this acquisition (probing): (i) users estimate their channels and then feed the \ac{CSI} back to their corresponding transmitters in a \ac{TDMA} manner, and (ii) users send training sequences in the uplink so that the transmitters can estimate the channels. The latter scheme, which we adopt in our system,  uses (pre-assigned) orthogonal sequences among the users, so the length of each one of these sequences should be proportional to the number of active users in the system; orthogonal sequences are produced e.g. by Walsh-Hadamard on pseudonoise sequences. It means that after acquiring the \ac{CSI} of, for example, $L$ users, the throughput is multiplied by $1-L\theta$, where $\theta$ is the fraction of time that takes the CSI acquisition of one user \cite{Chaporkar2009ScedwithLimited}.
Thus, it can be seen that the more the number of active pairs $L$ is large, the more the acquisition process consumes a larger fraction of time and hence leaves a smaller fraction for transmission. Thus, it is important to focus on the tradeoff between having a large number of active transmitter-receiver pairs (so having a high probing cost but many pairs can communicate simultaneously) and having much time of the slot dedicated to data transmission (which means getting a low probing cost but few pairs can communicate simultaneously)  \cite{Destounis2015Traffic-Aware}. Therefore, under this scheme, it can happen that only a subset of transmitter-receiver pairs is active (scheduled) at each timeslot. 

In order to choose the subset of active pairs at each timeslot, three approaches can be used \cite{Destounis2015Traffic-Aware}: (i) the centralized scheme (policy), where the decision of which pairs will be scheduled is made at the transmitters side and based only on the statistics of the channels of the users and the state of their queue lengths at each slot \cite{Lau2012StabilityDelay}, (ii) the decentralized scheme, meaning that the users decide which subset of them should actually train, and consequently this subset with its corresponding subset of transmitters will be active for transmission, and (iii) the mixed policy, which corresponds to combine the centralized and decentralized policies. Note that the  centralized approach is used in current standards (e.g. Long Term Evolution (LTE) \cite{Stefania2009LTEtheUMTS}), where the base station explicitly requests some users for their \ac{CSI}.

In this paper, we adopt the first approach, that is the centralized policy. Specifically, for the \ac{MIMO} system model described earlier, in which we use \ac{IA} as an interference management technique, we consider that there is a \ac{CS} that has a full knowledge of the queue lengths at each timeslot and the statistics of the channels. Based on this information, this \ac{CS} schedules the subset of pairs at each timeslot. In broad terms, using the centralized policy, we examine in this work the stability performances  of a \ac{MIMO} system under \ac{TDD} mode with limited backhaul capacity, where we apply \ac{IA} as an effective way to reduce the interference and where the \ac{CSI} probing cost is accounted for.

It is known that with \ac{IA} technique the backhaul is flooded due to the \ac{CSI} exchange process among the active transmitters. In some scenarios, it may be beneficial not to occupy the backhaul with this huge amount of signaling but instead exploited it more efficiently. For instance, if the backhaul is wireless, the \ac{CSI} exchange process consumes a part of the total reserved bandwidth, which can be instead used in the transmission process. Hence, it is of high interest to study the system under an interference management technique for which no \ac{CSI} exchange over the backhaul is required. For this purpose, we investigate the system performance under \ac{TDMA} as a channel access method, meaning that there is only one active pair at a given timeslot and thus no backhaul usage occurs, and using \ac{SVD} as a precoding technique. The choice of \ac{SVD} can be justified by the fact that it provides the best performances for point-to-point \ac{MIMO} systems \cite{Emre1999Capacityof}. One may wonder which one between \ac{TDMA}-\ac{SVD} and \ac{IA} outperforms the other in terms of stability. We will provide an answer to this question by comparing the system stability performances under these two techniques.

The rest of this paper is organized as follows. Section \ref{sec:systemmodel} presents the system model and the interaction between physical layer and queueing performance. The average rate expressions under the adopted system are derived in Section \ref{sec:calcul_avgrate}. In Section \ref{sec:PerAnalysis_sym}, we present a deep stability analysis for the symmetric system   
where all the path loss coefficients are equal to each other. Specifically, for this system, we provide a precise characterization of the stability region and we propose an optimal scheduling decision to achieve this region in both the perfect and imperfect cases. Further, we examine the gap between these two resulting stability regions, namely the region under the imperfect case and the one under the perfect case. Furthermore, for this same system, we compare the stability region of \ac{IA} with the one achieved by \ac{TDMA}-\ac{SVD}, then, using this comparison, we provide a way to select one of these two techniques. In addition, we characterize the resulting stability region when the considered system switches between these two techniques. At the end of this section, we investigate the impact of the number of bits and the number of pairs on the system stability region.
In Section \ref{sec:generalcase}, we investigate the stability performances for the general case, namely where the path loss coefficients are not necessarily equal to each other, by characterizing the corresponding stability region and providing an optimal scheduling policy, under both the imperfect and perfect cases. Then, since the scheduling policy for this system is of high computational complexity, we propose an approximate policy that has a reduced complexity but that achieves only a fraction of the system stability region. After that, a characterization of the achievable fraction is provided. At the end of this section, we investigate the gap between the stability region under the imperfect case and the one under the perfect case. Section \ref{sec:numerical} is dedicated to numerical results. Finally, Section \ref{sec:conclusions} concludes the paper.

\emph{Notation}: Boldface uppercase symbols (i.e., ${\bf A}$) represent matrices and lowercases (i.e., ${\bf a}$) are used for vectors, unless stated otherwise. The symbol ${\bf I}_N$ denotes the identity matrix of size $N$. The operator $\otimes$ is the Kronecker product. The notation $|\cdot|$ is used to indicate the absolute value for scalars and the cardinality for sets (or subsets). In addition, $||\cdot||_1$  and $||\cdot||$ are used for the norms of first and second degree, respectively. The notation $\mathbf{1}$ is used for the all-ones vector. Finally, superscripts $T$ and $H$ over a matrix or vector denote its transpose and conjugate transpose,
respectively.

\section{System Model}
\label{sec:systemmodel}

\begin{figure}[ht!]
\centering
\includegraphics[width=0.87\linewidth, height=12cm]{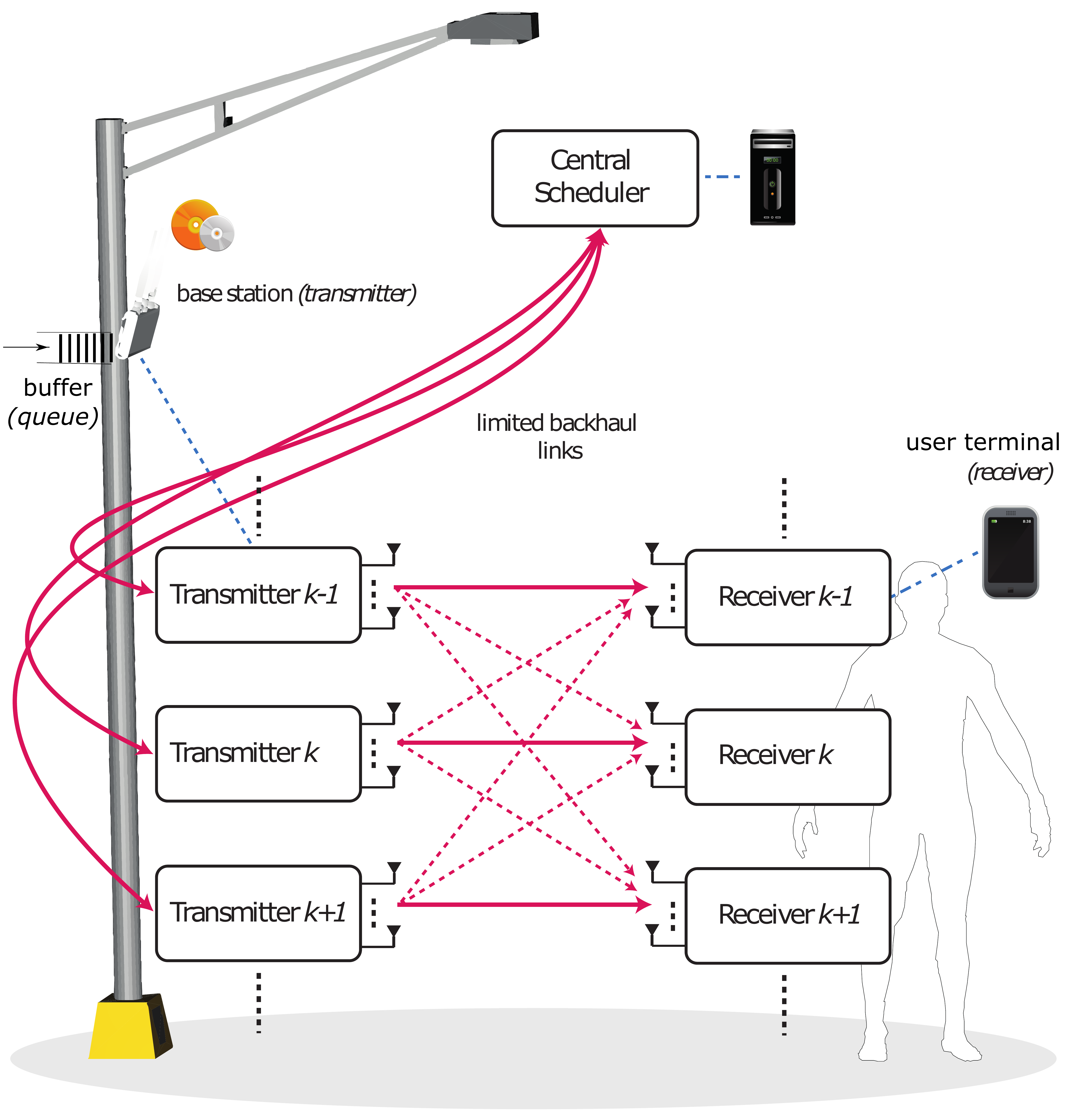}
\captionsetup{font=small}
\caption{\small{A sketch of $N$-user \ac{MIMO} interference network with limited backhaul.}}
\label{fig:scenario}
\end{figure} 
We consider the \ac{MIMO} interference channel with $N$ transmitter-receiver pairs shown in Fig. \ref{fig:scenario}. For simplicity of exposition, we consider a homogeneous network where all transmitters are equipped with $N_{\text{t}}$ antennas and all receivers (users) with $N_{\text{r}}$ antennas. We assume that time is slotted. As we will see later on, only a subset $\mathcal {L}(t)$, of cardinality $L(t)$, of pairs is active at each timeslot, with $L(t) \le N$.  While each transmitter communicates with its intended receiver, it also creates interference to other $L(t)-1$ unintended receivers. Transmitter $k$ has $d_k \le \min \left(  N_{\text{t}},N_{\text{r}}\right)$ independent data streams to  transmit to its intended user $k$. 

Given this channel model, the received signal at active user $k$ 
($\in \mathcal{L}(t)$) can be expressed as 
\begin{align}
\mathbf{y}_k=\sum\limits_{i \in \mathcal {L}(t)} \sqrt{ \frac{ \zeta_{ki}P }{d_i} } \mathbf{H}_{ki} \sum\limits_{j=1}^{d_i}  \mathbf{v}_i^{(j)} x_i^{(j)} + \mathbf{z}_k,
\label{eq:yk}
\end{align}
where $\mathbf{y}_k$ is the $N_r \times 1$ received signal vector, $\mathbf{z}_k$ is the additive white complex Gaussian noise with zero mean and covariance matrix $\sigma^2\mathbf{I}_{N_{\text{r}}}$, $\mathbf{H}_{ki}$ is the $N_{\text{r}} \times N_{\text{t}}$ channel matrix between transmitter $i$ and receiver $k$ with independent and identically distributed (i.i.d.) zero mean and unit variance complex Gaussian entries, $\zeta_{ki}$ represents
the path loss of channel $\mathbf{H}_{ki}$, $P$ is the total power at each transmitting node, which is equally allocated among its data streams, $x_i^{(j)}$ represents the $j$-th data stream from transmitter $i$, and $\mathbf{v}_i^{(j)}$ is the corresponding  $N_{\text{t}} \times 1$ precoding vector of unit norm.
For the rest of the paper, we denote by $\alpha_{ki}$ the fraction $\frac{ \zeta_{ki}P }{d_i}$.

\subsection{Interference Alignment Technique}

\ac{IA} is an efficient linear precoding technique that often achieves the full \ac{DoF} supported by \ac{MIMO} interference channels. In cases where the full \ac{DoF} cannot be guaranteed, \ac{IA} has been shown to provide significant gains in high \ac{SNR} sum-rate. To investigate \ac{IA} in our model, we start by examining the effective channels created after precoding and combining. For tractability, we restrict ourselves to a per-stream zero-forcing receiver. Recall that in the high (but finite) \ac{SNR} regime, in which \ac{IA} is most useful, gains from more involved receiver designs are limited \cite{ElAyach2012OverheadofIA}. In such a system, receiver $k$ uses the $N_{\text{r}} \times 1$ \emph{combiner} vector $\mathbf{u}_k^m$ of unit norm to detect its $m$-th stream, such as
\begin{align}
\hat{x}_k^{(m)}  \nonumber &= \left(\mathbf{u}_k^{(m)} \right)^{\! H} \mathbf{y}_k &\\
\nonumber &=   \overbrace{  \sqrt{\alpha}_{kk}  \left(\mathbf{u}_k^{(m)} \right)^{\! H} \mathbf{H}_{kk} \mathbf{v}_k^{(m)} x_k^{(m)}   }^{\text{desired signal}}     +  \overbrace{  \sqrt{\alpha}_{kk} \sum_{j=1,j\ne m}^{d_k} \left(\mathbf{u}_k^{(m)} \right)^{\! H} \mathbf{H}_{kk} \mathbf{v}_k^{(j)} x_k^{(j)} }^{\text{inter-stream interference (ISI)}}   &\\  &\qquad {} +   \overbrace{   \sum\limits_{\substack{ i \in \mathcal {L}(t), i\ne k}}  \sqrt{\alpha}_{ki}  \sum\limits_{j=1}^{d_i} \left(\mathbf{u}_k^{(m)} \right)^{\! H} \mathbf{H}_{ki} \mathbf{v}_i^{(j)} x_i^{(j)}  }^{\text{inter-user interference (IUI)}} + \overbrace{  \left(\mathbf{u}_k^{(m)} \right)^{\! H} \mathbf{z}_k}^{\text{noise}},  \label{eq:xk} 
\end{align} 
where the first term at the right-hand-side of this expression is the desired signal, the second one is the \ac{ISI} caused by the
same transmitter, and the third one is the \ac{IUI}
resulting from the other transmitters. In order to mitigate these
interferences and improve the system performances, \ac{IA} is performed
accordingly, that is designing the set of combiner and precoder vectors such that
\begin{align}
\left(\mathbf{u}_k^{(m)} \right)^{\! H} \mathbf{H}_{ki} \mathbf{v}_i^{(j)} = 0,  &&  \forall (i,j)\ne (k,m), \text{with} \, i,k \in \mathcal {L}(t). \label{eq:IAeq} 
\end{align} 
Note that the above conditions are those of a \emph{perfect interference alignment}. In other words, suppose that all the transmitting nodes have perfect global \ac{CSI} and each receiver obtains a perfect version of its corresponding combiner vector, \ac{ISI} and \ac{IUI} can be suppressed completely. However, obtaining the perfect global \ac{CSI} at the transmitters is not always practical due to the fact that backhaul links, which connect transmitters to each other, are of limited capacity. The \ac{CSI} sharing mechanism is detailed in the next subsection.

Finally, some assumptions  and remarks are in order. First, in our study, we assume that each active receiver obtains a perfect version of its corresponding combiner vector. The cost of this latter process is not considered in our analysis. In addition, it is worth noting that due to the limitation of spatial degree of freedom, the values of $d_k$ must fulfill the feasibility conditions of \ac{IA} \cite{Yetis2010Feasibility}. In what follows, we suppose \ac{IA} is feasible by selecting the data steams numbers $d_k$ carefully. Further, we recall that the total transmit power is split equally among the transmitters, and then each of which equally allocates its power among its data streams; it means that we do not perform power control for our system. This is done to further simplify the transmission scheme that relies on \ac{IA} technique, which does not lack complexity. 

\subsection{CSIT Sharing Over Limited Capacity Backhaul Links}
The process of \ac{CSI} sharing is restricted to the scheduled pairs (represented by subset $\mathcal{L}(t)$). Thus, here, even if we did not mention it, when we write “transmitter” (resp., “user”) we mean “active transmitter” (resp., “active user”). Three different scenarios regarding the \ac{CSI} sharing problem can be considered:
\begin{enumerate}[(a)]
\item Each transmitter receives all the required \ac{CSI} and independently
computes the \ac{IA} vectors,
\item The \ac{IA} processing node is a separate central node that
computes and distributes the \ac{IA} vectors to other transmitters,
\item One transmitter acts as the \ac{IA} processing node. 
\end{enumerate}
For the last two scenarios, one node performs the computations and then distributes the \ac{IA} vectors among transmitters. So, since the backhaul is limited in capacity, in addition to the quantization required for the \ac{CSI} sharing process, another quantization is needed to distribute the \ac{IA} vectors over the backhaul. This is not the case for the first scenario where only the first quantization process is needed. Thus, for simplicity of exposition and calculation, we focus on the first scenario, which we detail in the following. 

As alluded earlier, global \ac{CSI} is required at each transmitting node in order to design the \ac{IA} vectors that satisfy \eqref{eq:IAeq}. As shown in Fig. \ref{fig:scenario}, we suppose that all the transmitters are connected to a \ac{CS} via their limited backhaul links, meaning that this \ac{CS} serves as a way for connecting the transmitters to each other; as we will see later on, this scheduler decides which pairs to schedule at each timeslot. We assume a \ac{TDD} transmission strategy, which enables the transmitters to estimate their channels toward different users by exploiting the reciprocity of the wireless channel. We consider throughout this paper that there are no errors in the channel estimation. Under the adopted strategy, the users send their training sequences in the uplink phase, allowing each transmitter to estimate (perfectly) its \emph{local} \ac{CSI}, meaning that the $i$-th transmitter estimates perfectly the channels  $\mathbf{H}_{ki}$, $k,i \in \mathcal {L}(t)$. However, the local \ac{CSI}, excluding the direct links (since they do not enter in computing the \ac{IA} vectors), of other transmitters are obtained via backhaul links of limited capacity.
In such limited backhaul conditions, a codebook-based quantization technique needs to be adopted to reduce the huge amount of information exchange used for \ac{CSI} sharing, which we detail as follows. Let $\mathbf{h}_{ki} $ denote the vectorization of the channel matrix $\mathbf{H}_{ki}$. Then, for all $i \ne k$, transmitter $i$ selects the index $n_o$ that corresponds to the optimal codeword in a predetermined codebook $\mathcal{CB}=\left[\mathbf{\hat {h}}_{ki}^{(1)},...,\mathbf{\hat {h}}_{ki}^{(2^B)} \right] $ according to
\begin{equation}
 n_o=  \operatorname*{arg\,max}_{1\le n \le 2^B} \left\{  \left| \left( \mathbf{\tilde{h}}_{ki} \right)^{\! H} \, \mathbf{\hat {h}}_{ki}^{(n)}  \right|^2 \right\}, \label{eq:n0}
\end{equation}
in which $B$ is the number of bits used to quantize $\mathbf{H}_{ki}$ and  $\mathbf{\tilde{h}}_{ki} =  \frac{\mathbf{h}_{ki}}{\left\| \mathbf{h}_{ki}  \right\|}$ is the channel direction vector.
After quantizing all the matrices of its local \ac{CSI}, we assume that transmitter $i$ sends the corresponding optimal indexes to all other  active transmitters, which share the same codebook, allowing these transmitters to reconstruct the quantized local knowledge of transmitter $i$. Let us now define the quantization error as $e_{ki}= 1-\frac{ \left| \mathbf{\hat{h}}_{ki}^H \mathbf{h}_{ki} \right|^2 }  { \left\| \mathbf{h}_{ki} \right\|^2 } $ and adopt the same model in \cite{Jindal2006MIMO,Lau2012StabilityDelay} that relies on the theory of quantization cell approximation. The \ac{CDF} of $e_{ki}$ is then given by
 \begin{align}
\mathbb{P} \left\{ e_{ki} \le \varepsilon \right\} =
  \begin{dcases}
   2^B \varepsilon^Q, &  \qquad{} 0\le \varepsilon \le 2^{- \frac{B}{Q}} \\
   1, & \qquad{} \, \varepsilon > 2^{- \frac{B}{Q}}
  \end{dcases}
  \end{align}
 where $Q=N_{\mathrm{t}}N_{\mathrm{r}}-1$.

\subsection{Rate Model and Impact of Training}
 
Before proceeding with the description, we define the \emph{perfect case} as the case where the backhaul has an infinite capacity, which leads to a perfect global \ac{CSI} knowledge at the transmitters; so no quantization is needed. Further, we call \emph{imperfect case} the model described previously, where a quantization is performed over the backhaul of limited capacity.
 
For the perfect case, the \ac{IA} constraints null the \ac{ISI} and the \ac{IUI}, and no  residual interference exists. For the imperfect case, as explained in the previous subsection, each transmitter designs its \ac{IA} vectors based on a perfect version of its local \ac{CSI} and an imperfect (quantized) version of the local \ac{CSI} of other transmitters. For this reason, in this case, the IA technique is able to completely cancel the \ac{ISI} but not the \ac{IUI}. 
Thus, under such observations, the SINR/SNR for stream $m$ at active receiver $k$ can be written as
\begin{align}
 \gamma_k^{(m)} =
 \begin{dcases}
 \, \, \frac{ \alpha_{kk} \left|  \left(\mathbf{\hat{u}}_k^{(m)} \right)^{\! H} \mathbf{H}_{kk} \mathbf{\hat{v}}_k^{(m)}  \right|^2  }{  \sigma^2 +  \sum\limits_{\substack{i \in \mathcal {L}(t), i \ne k}}   \alpha_{ki}  \sum\limits_{j=1}^{d_i} \left| \left(\mathbf{\hat{u}}_k^{(m)} \right)^{\! H}  \mathbf{H}_{ki} \mathbf{\hat{v}}_i^{(j)} \right|^2     },   & \qquad{} \text{imperfect case} \\
 \, \, \frac{ \alpha_{kk} \left|  \left(\mathbf{u}_k^{(m)} \right)^{\! H} \mathbf{H}_{kk} \mathbf{v}_k^{(m)}  \right|^2  }{\sigma^2},   &  \qquad{} \text{perfect case}    \label{eq:SINR} 
  \end{dcases}
\end{align}
where $\mathbf{\hat{v}}_k^{(m)}$ and $\mathbf{\hat{u}}_k^{(m)}$ are designed under the limited backhaul case, that is to say using an imperfect global \ac{CSI} due to the quantization process over the backhaul, whereas $\mathbf{v}_k^{(m)}$ and $\mathbf{u}_k^{(m)}$ are designed under the unlimited backhaul case, i.e. using a perfect global \ac{CSI}.
As alluded earlier, only a subset $\mathcal{L}(t)$ (we recall that $ \left| \mathcal{L}(t) \right| =L(t)$) of pairs is scheduled at a time. For notational convenience, we will use \ac{SINR} as a general notation to denote \ac{SNR} for the perfect case and \ac{SINR} for the imperfect case, unless stated otherwise.

We now explain some useful points that are adopted in the rate model.
At a given timeslot, a rate of $R$ bits is assigned to stream $m$ of user $k$ if $\gamma_k^{(m)}$, i.e. the corresponding \ac{SINR}, is higher than or equal to a given threshold, which we denote by $\tau$; otherwise, the assigned rate is $0$. Let us denote by $\tilde{R}_k(t)$ the assigned rate (in units of bits/slot) for user $k$ at timeslot $t$, thus $\tilde{R}_k(t)$ is the sum of the assigned rates for all the streams of user $k$ at $t$.
For this model, channel acquisition cost is not negligible and should be considered. As mentioned earlier, we consider a system under \ac{TDD} mode where users send training sequences in the uplink so that the transmitters can estimate their channels; this is a promising approach, especially for systems with large antenna arrays at the transmitters, due to the fact that the feedback overhead does not scale with the number of antennas. This scheme uses orthogonal sequences among the users, so their lengths are proportional to the number of active users in the system. We assume that acquiring the \ac{CSI} of one user takes fraction $\theta$ of the slot. Thus, since we have $L(t)$ active users, the \emph{actual} rate for transmission to active user $k$ at timeslot $t$ is $(1-L(t)\theta) \tilde{R}_k(t)$. Let us define $B_k(t)=(1-L(t)\theta) \tilde{R}_k(t)$. Note that $B_k(t)$ is equal to $0$ if pair $k$ is not active at time $t$.

Under this setting, the average rate for active user $k$ can be written in function of the transmission success probability conditioned on the subset of active pairs as 
\begin{align}
\mathbb{E}\left\{ B_k(t) \mid \mathcal{L}(t) \right\} =(1-L(t) \theta) \sum_{m=1}^{d_k}   R \, \mathbb{P} \left\{\gamma_k^{(m)} \ge \tau \mid \mathcal{L}(t) \right\}  \label{eq:r_M_M1}.
\end{align}
It can be noticed that the feedback overhead ($1-L(t) \theta$) scales with the number of active pairs, meaning that when $L(t)$ is large there will be little time left to transmit in the timeslot before the channels change again. Here, it is clear that the fraction $L(t) \theta$ should be lower than $1$. Since the maximum number of pairs $N$ is such that $N \ge L(t)$, we should also have $N \theta <1$. In practice, the fraction of the timeslot dedicated for \ac{CSI} acquisition is less than $\frac{1}{2}$, i.e.  at least half of the timeslot is reserved for data transmission.

\subsection{Queue Dynamics, Stability and Scheduling Policy}  

For each user, we assume that the incoming data is stored in a respective queue (buffer) until transmission, and we denote by $\mathbf{q}(t)=\left[ q_1(t),...,q_N(t)  \right]$ the queue length vector. We designate by $\mathbf{A}(t)=\left[ A_1(t),...,A_N(t)  \right]$ the vector of number of bits arriving in the buffers in timeslot $t$, which is an i.i.d. in time process, independent across users and with $A_k(t)<A_{\text{max}}$. The mean arrival rate (in units of bits/slot) for user $k$ is denoted by $a_k= \mathbb{E} [A_k(t)]$. We recall that a user will get $B_k(t)$ served bits per slot if it gets scheduled and zero otherwise. Note that $B_k(t)$ is finite because $R$ is finite, so we can define a finite positive constant $B_{\text{max}}$ such that $B_k(t)<B_{\text{max}}$, for $k=1,\ldots,N$.

At each timeslot, the \ac{CS} selects the pairs to schedule based on the queue lengths and average rates (per user) in the system. To this end, we suppose that (i) this scheduler  has a full knowledge of average rate values under different combinations of choosing active pairs, which can be provided offline since an average rate is time-independent, (ii) at each timeslot, each transmitter sends its queue length to the \ac{CS} so that it can obtain all the queue dynamics of the system, and (iii) the cost of providing such knowledge to the scheduler will not be taken into account in our analysis. After selecting the set of pairs to be scheduled (represented by $\mathcal{L}(t)$), the \ac{CS} broadcasts this information so that the selected transmitter-user pairs activate themselves, and then the active users send their pilots in the uplink so that the (active) transmitters can estimate the \ac{CSI}. It is worth noting that, as alluded previously, if we select a large number of pairs ($L(t)$) for transmission, many pairs can communicate (i.e. this will leave a small fraction of time for transmission) but a high CSI acquisition cost is needed. On the other hand, a small $L(t)$ requires a low acquisition cost, but, at the same time, it allows a few number of simultaneous transmissions.
The decision of selecting active pairs is referred simply as the \emph{scheduling policy}. 
At the $t$-th slot, this policy can be represented by an indicator vector $\mathbf{s}(t) \in \mathcal{S} \coloneqq \left\{  0,1 \right\}^N$, where the $k$-th component of $\mathbf{s}(t)$, denoted by $s_k(t)$, is equal to $1$ if the $k$-th queue (pair) is scheduled or otherwise equal to $0$. It can be seen that the cardinality of set  $\mathcal{S}$ is equal to $\left| \mathcal{S}\right|=2^N$. Remark that, in terms of notation, $\mathbf{s}(t)$ and $\mathcal{L}(t)$ are used to represent the same thing, that is the scheduled pairs at timeslot $t$, but they illustrate it differently. Specifically, using $\mathbf{s}(t)$ the active pairs correspond to the non-zero coordinates (equal to $1$), whereas $\mathcal{L}(t)$ contains the indexes (positions) of these pairs. Let $\bm{\mathcal{L}}$ be the set of all possible subsets $\mathcal{L}(t)$.

Now, using the definition of $B_k(t)$, which was provided earlier,
the queueing dynamics (i.e. how the queue lengths evolve over time) can be described by the following
\begin{align}
q_k(t+1)=\max \left\{q_k(t)- B_k(t), 0\right\} + A_k(t),  \qquad \, \,  \forall k \in \{1,\ldots,N \}, \forall t \in \{0,1,\ldots \}, \label{eq:q_evo}
\end{align}
where we note that $B_k(t)$ depends on the scheduling policy.

In this work, the focus will be mainly on the stability of the system. Formally, its definition is as follows.
\begin{defi}[Strong Stability]
The condition for strong stability of the system can be expressed  as the following
\begin{align}
 \underset{T \rightarrow \infty}{\lim \sup}  \frac{1}{T} \sum\limits_{t=0}^{T-1}  \mathbb{E} \left\{ q_k(t) \right\} < \infty, \forall k \in \{ 1,...,N \}.
 \end{align}
\end{defi}
From this definition, stability implies that the mean queue length of every queue in the system is finite, further implying finite delays in single hop systems. Note that in the remainder of the manuscript “stable” will imply “strongly stable” unless stated otherwise. This definition leads us to the concept of stability region.
\begin{defi}[Stability Region]
The stability region can be defined as the set of mean arrival rate vectors for which all the queues are strongly stable. Furthermore, a scheduling policy (algorithm) that achieves this region is called throughput optimal.
\end{defi}
For the rest of the paper, when describing and characterizing stability regions, we implicitly mean that the system is stable in the \emph{interior} of the characterized region. Normally, for the boundary points, the system has at least a weaker form of stability called “mean rate stability”.

Now, we discuss stability optimal policies for our setting. To this end, we denote the stability region by $\Lambda$ and we define $\mathcal{V}$ as the set that contains the corner points (vertices) of this region. If we have a system where the arrival rates are known, the stability can be achieved by a predefined time-sharing strategy. Indeed, an arrival rate vector $\mathbf{a} \in \Lambda$ can be expressed as a convex combination of the points in $\mathcal{V}$. More in detail, we have $\mathbf{a}=\sum_{n=1}^{|\mathcal{V}|} p_n \mathbf{r}_n$, where $\mathbf{r}_n$ represents the $n$-th element of $\mathcal{V}$, $p_n \ge 0$ and $\sum_{n=1}^{|\mathcal{V}|} p_n =1$. We can find at least one point $\mathbf{a}^\prime$ on the boundary of $\Lambda$ such that $\mathbf{a} \preceq \mathbf{a}^\prime$. Since $\mathbf{a}^\prime \in \Lambda$, we can  write $\mathbf{a}^\prime = \sum_{n=1}^{|\mathcal{V}|} p^\prime_n \mathbf{r}_n$, with $p^\prime_n \ge 0$ and $\sum_{n=1}^{|\mathcal{V}|} p^\prime_n =1$. Recall that a point $\mathbf{r}_n$ represents a specific scheduling decision. Then, to achieve queues stability, each point (decision) $\mathbf{r}_n$ should be selected with probability $p^\prime_n$.  
In our system, as well as in most practical systems, a-priori knowledge of the arrival rates is not available, which is needed to calculate the set of probabilities $p^\prime_n$. We recall that at the beginning of each timeslot the \ac{CS} makes the scheduling decision, i.e. selects the set of active pairs, and knows only the queue lengths and the channel statistics. 
Then, in order to stabilize the queues in our system, we can consider the knowledge of average rates and queue lengths rather than arrival rates \cite{Lau2012StabilityDelay,Georgiadis06resourceallocation}, using the policy described as the following 
\begin{align}
 \Delta^{\text{*}} :  \mathbf{s}(t) = \operatorname*{arg\,max}_{ \mathbf{s} \in   \mathcal{S}  }  \left\{ \mathbf{r}(  \mathbf{s} ) \cdot \mathbf{q}(t)  \right\}, \label{eq:MW}
\end{align}
where “$\cdot$” is the scalar (dot) product, and $\mathbf{r}(  \mathbf{s} )$ is constructed by replacing the non-zero coordinates of $\mathbf{s}$, which represent the selected pairs, with their corresponding average rate values. More in detail, recalling that $\mathcal{L}$ represents the positions (indexes) of the non-zero coordinates of $\mathbf{s}$, vector $\mathbf{r}(  \mathbf{s} )$ contains $ \mathbb{E}\{ B_k \mid \mathcal{L} \}$ at position $k$ if the $k$-th coordinate of $\mathbf{s}$ is '$1$' and $0$ if this coordinate is '$0$'. The proposed algorithm is nothing but a weighted sum maximization, and in general it is called the \emph{Max-Weight} rule.  
Remark that, due to our centralized setting in which the \ac{CS} should select the set of active pairs at the beginning of each timeslot, the scheduling policy in our system depends on the average transmission rate and not on the instantaneous one. For this policy, the following statement holds.
\begin{lemma}
\label{le:OptSchedPolicy}
Under the adopted system, the scheduling policy $\Delta^{\text{*}}$ is throughput optimal, meaning that it can stabilize the system for every mean arrival rate vector in $\Lambda$.
\end{lemma}
\begin{proof}
We show that policy $\Delta^{\text{*}}$ stabilizes the system for all $\mathbf{a} \in \Lambda$ by proving that the  Markov chain of the corresponding system is positive recurrent. For this purpose, we use Foster's theorem. Such proof is standard in the literature and is thus omitted for sake of brevity.
\end{proof}
\subsubsection*{Computational Complexity of $\Delta^{\text{*}}$}
For such optimal policy, an important factor to investigate is the computational complexity (CC), which we derive next. Because what we are looking for is the maximum over $2^N$ possible values, due to $2^N$ combinations, thus it takes $O(2^{N})$  after computing all values $\mathbf{r}(  \mathbf{s} )  \cdot \mathbf{q}(t)$ to find the maximum value (resp., the corresponding argument). Note that for two fixed vectors we can compute this product in time $O(N)$. Thus we would have $O(N 2^{N})$ ignoring computing $\mathbf{r}( \mathbf{s} )$, which can be done offline. We can notice that this computational complexity increases considerably with the maximum number of pairs $N$.
\begin{table}
\centering
\captionsetup{labelsep=newline,size=small,font=sc}
\caption{list of the parameters used in the model}   
\small
\begin{tabular}{|c||l|}
\hline
\textbf{Parameter}	& \textbf{Description}                              \\
\hline \hline
$N$		  &	Maximum number of  pairs  					                \\
\hline
$N_\text{t}$ 	  &	Number of antennas at each transmitter 			    \\
\hline
$N_\text{r}$	  & Number of antennas at each receiver  				 \\
\hline
$d_k$       &	Number of data streams 	for pair $k$	                 \\
\hline
$P$        &  Total power at each transmitter                            \\
\hline
$\theta$  &  Fraction of slot duration to probe one user 		          \\
\hline
$B$       &  Number of quantization bits		          \\
\hline                                       
$\tau$      &  SINR threshold    
\\       
\hline
$R$      &  Assigned rate corresponding to $\tau$                                   \\
\hline
$\mathcal{L}(t)$ & Subset of scheduled (active) pairs at timeslot $t$     \\
\hline
$L(t)$    & Cardinality of subset $\mathcal{L}(t)$                         \\
\hline
$\mathbf{s}(t)$     & Scheduling decision vector at timeslot $t$                         \\       
\hline
$\mathbf{a}=\left( a_1,\ldots,a_N \right)$   &  Vector of mean arrival rates (in bits per timeslot)                                                             \\
\hline
\end{tabular}
\end{table}

\section{Derivation of Success Probabilities and Average Rates}
\label{sec:calcul_avgrate}

In this section, we give the expression for the success probability of the \ac{SINR} and, subsequently, the expression for the average transmission rate under the imperfect case as well as under the perfect case. 

For the calculation of the average rate,
we recall that we adopt the model that relies on the success probability. Note that other average rate model exists and could be adopted (see \cite{Ayach2012Interference}), which consists in averaging the $\log_2 (1+ \text{SINR})$. We next provide a proposition in which we calculate the success probabilities under the considered setting.
\begin{proposition}
The probability that the received \ac{SINR} corresponding to stream $m$ of active user $k$ exceeds a threshold $\tau$ given that $\mathcal{L}(t)$ is the set of scheduled pairs (including pair $k$) can be given by
\begin{align}
\mathbb{P} \left\{ \gamma_k^{(m)} \ge \tau \mid \mathcal{L}(t) \right\}=
\begin{dcases}
\, \, e^{ -\frac{\sigma^2 \tau}{\alpha_{kk}}} \, \mathit{MGF}_{\! \! \mathit{RI}_k^{(m)}} \! \left(-\frac{\tau}{\alpha_{kk}} \right), &  \qquad{} \text{imperfect case} \\
\, \, e^{ -\frac{\sigma^2 \tau}{\alpha_{kk}}},  &  \qquad{} \text{perfect case}
\end{dcases}
\end{align}
where 
\begin{align}
\mathit{RI}_{\! k}^{(m)}= \sum\limits_{\substack{i \in \mathcal {L}(t), i \ne k}}   \alpha_{ki}  \sum\limits_{j=1}^{d_i} \left| \left(\mathbf{\hat{u}}_k^{(m)} \right)^{\! H}  \mathbf{H}_{ki} \mathbf{\hat{v}}_i^{(j)} \right|^2
\end{align}
is the residual interference, which appears in the denominator of $\gamma_k^{(m)}$ in the imperfect case, and 
$\mathit{MGF}_{ \! \! \mathit{RI}_{\! k}^{(m)} } (\cdot)$ stands for the \ac{MGF} of $\mathit{RI}_{\! k}^{(m)}$. \label{pr:SuccProb}
\end{proposition}
\begin{proof}
It was shown in \cite[Appendix A] {Ayach2012Interference} that  both $\left|  \left(\mathbf{\hat{u}}_k^{(m)} \right)^{\! H } \mathbf{H}_{kk} \mathbf{\hat{v}}_k^{(m)}  \right|^2$ and  $\left|  \left(\mathbf{u}_k^{(m)} \right)^{\! H } \mathbf{H}_{kk} \mathbf{v}_k^{(m)}  \right|^2$ have an exponential distribution with parameter 1, thus the proof for the prefect case follows directly. However, for the imperfect case, the proof is not straightforward and needs some investigations. By defining $G= \left|  \left(\mathbf{\hat{u}}_k^{(m)} \right)^{\! H } \mathbf{H}_{kk} \mathbf{\hat{v}}_k^{(m)}  \right|^2$, we can write
\begin{align}
\mathbb{P} \left\{ \gamma_k^{(m)} \ge \tau \mid \mathcal{L}(t) \right\} \nonumber &= \mathbb{P} \left\{ \frac{ G }{   \mathit{RI}_{\! k}^{(m)} + \sigma^2 }\ge \frac{\tau}{\alpha_{kk}} \right\} \nonumber\\ &=\mathbb{P} \left\{ G \ge \frac{ \mathit{RI}_{\! k}^{(m)} \tau}{\alpha_{kk}}   + \frac{ \sigma^2 \tau}{\alpha_{kk}}  \right\} \nonumber \\  &= \int_0^\infty
 \mathit{CCDF}_{\!  G} \left( \frac{\mathit{RI}_{\! k}^{(m)} \tau}{\alpha_{kk}}   + \frac{ \sigma^2 \tau}{\alpha_{kk}} \right) \mathit{PDF} \! \left( \mathit{RI}_{\! k}^{(m)}  \right) \,  d \mathit{RI}_{\! k}^{(m)} \nonumber \\  &= \int_0^\infty e^{-\frac{ \mathit{RI}_{\! k}^{(m)} \tau}{\alpha_{kk}}  - \frac{ \sigma^2 \tau}{\alpha_{kk}}} \, \mathit{PDF} \! \left( \mathit{RI}_{\! k}^{(m)} \right) \, d\mathit{RI}_{\! k}^{(m)},
\end{align}
where the last equality holds since $G$ is exponentially distributed with parameter $1$ and thus its complementary cumulative distribution function can be given by $\mathit{CCDF}_{\! G}(x)= e^{-x} $. Note that $\mathit{PDF} \! \left( \mathit{RI}_{\! k}^{(m)} \right)$ is the \ac{PDF} of $\mathit{RI}_{\! k}^{(m)}$. Thus, we get
\begin{align}
\mathbb{P} \left\{ \gamma_k^{(m)} \ge \tau \mid \mathcal{L}(t) \right\} \nonumber&= \int_0^\infty e^{- \frac{ \sigma^2 \tau}{\alpha_{kk}}  } \, e^{-\frac{ \mathit{RI}_{\! k}^{(m)} \tau}{\alpha_{kk}} }   \, \mathit{PDF} \! \left( \mathit{RI}_{\! k}^{(m)} \right) \, d\mathit{RI}_{\! k}^{(m)}  \\ &=    e^{- \frac{ \sigma^2 \tau}{\alpha_{kk} } } \mathit{MGF}_{ \! \! \mathit{RI}_{\! k}^{(m)} } \! \left(-\frac{\tau}{\alpha_{kk}} \right),
\end{align}
in which $\mathit{MGF}_{ \! \! \mathit{RI}_{\! k}^{(m)} } (\cdot)$ is  the \ac{MGF} of $\mathit{RI}_{\! k}^{(m)}$. This concludes the proof.
\end{proof} 
In the above result, the success probability expression in the imperfect case is given in function of the \ac{MGF} of the leakage interference $\mathit{RI}_{\! k}^{(m)}$. 
It is noteworthy to mention that the explicit expression of this \ac{MGF} will be given  afterwards during the average rate calculations. But first, 
let us focus on the expression $\mathit{RI}_{\! k}^{(m)}$. Indeed, we have 
\begin{align}
\mathit{RI}_{\! k}^{(m)} \nonumber &= \sum\limits_{\substack{i \in \mathcal {L}(t), i \ne k}}   \alpha_{ki}  \sum\limits_{j=1}^{d_i} \left| \left(\mathbf{\hat{u}}_k^{(m)} \right)^{\! H}  \mathbf{H}_{ki} \mathbf{\hat{v}}_i^{(j)} \right|^2  \\ \nonumber &= \sum\limits_{\substack{i \in \mathcal {L}(t), i \ne k}}    \alpha_{ki} \sum\limits_{j=1 }^{d_i} \left| \mathbf{h}_{ki}^{H} \, \mathbf{T}_{k,i}^{(m,j)} \right|^2 \\   &=  \sum\limits_{\substack{i \in \mathcal {L}(t), i \ne k}}    \alpha_{ki} \left\| \mathbf{h}_{ki} \right\|^2  \sum\limits_{j=1 }^{d_i} \left| \mathbf{\tilde{h}}_{ki}^H \, \mathbf{T}_{k,i}^{(m,j)} \right|^2,
\end{align}
in which $\mathbf{T}_{k,i}^{(m,j)} = \mathbf{\hat{v}}_i^{(j)} \otimes ((\mathbf{\hat{u}}_k^{(m)})^H)^T $ (where $\otimes$ is the Kronecker product) and $\mathbf{\tilde{h}}_{ki}$ is the normalized vector of channel $\mathbf{h}_{ki}$, i.e. $\mathbf{\tilde{h}}_{ki}= \frac{\mathbf{h}_{ki}}{\left\| \mathbf{h}_{ki} \right\|}$. Note that $((\mathbf{\hat{u}}_k^{(m)})^H)^T$ is nothing but the conjugate of $\mathbf{\hat{u}}_k^{(m)}$.
Following the model used in \cite{Xiaoming2014PerformanceAnaly}, the channel direction $\mathbf{\tilde{h}}_{ki}$ can be written as follows
\begin{align}
\mathbf{\tilde{h}}_{ki}=  \sqrt{1-e_{ki}} \, \mathbf{\hat{h}}_{ki} + \sqrt{e_{ki}} \, \mathbf{w}_{ki},
\end{align}
where $\mathbf{\hat{h}}_{ki}$ is the channel quantization vector of $\mathbf{h}_{ki}$ and $\mathbf{w}_{ki}$ is a unit norm vector isotropically distributed in the null space of $\mathbf{\hat{h}}_{ki}$, with $\mathbf{w}_{ki}$ independent of $e_{ki}$.
Since \ac{IA} is performed based on the quantized \ac{CSI} $\mathbf{\hat{h}}_{ki}$, we get 
\begin{align}
\left| \mathbf{\tilde{h}}_{ki}^H \, \mathbf{T}_{k,i}^{(m,j)} \right|^2 = \left|  \sqrt{1-e_{ki}}  \, \mathbf{\hat{h}}_{ki}^H \, \mathbf{T}_{k,i}^{(m,j)} + \sqrt{e_{ki}} \, \mathbf{w}_{ki}^H \, \mathbf{T}_{k,i}^{(m,j)} \right|^2 = e_{ki} \left|  \mathbf{w}_{ki}^H \, \mathbf{T}_{k,i}^{(m,j)} \right|^2.
\end{align}
Therefore, $\mathit{RI}_{\! k}^{(m)}$ can be rewritten as
\begin{align}
\mathit{RI}_{\! k}^{(m)} =  \sum\limits_{\substack{i \in \mathcal {L}(t), i \ne k}}   \alpha_{ki} \left\| \mathbf{h}_{ki} \right\|^2  e_{ki}  \sum\limits_{j=1}^{d_i}  \left|  \mathbf{w}_{ki}^H \, \mathbf{T}_{k,i}^{(m,j)} \right|^2. \label{eq:LI}
\end{align}

Based on the above results, we now have all the required materials to derive the average rate expressions for both the perfect and imperfect cases. We recall that if $\mathcal{L}(t)$ is the subset of scheduled pairs, the general formula of the average rate of active user $k$ can be given as 
\begin{align}
\mathbb{E}\left\{ B_k(t) \mid \mathcal{L}(t) \right\} =(1-L(t) \theta) \sum_{m=1}^{d_k}   R \, \mathbb{P} \left\{\gamma_k^{(m)} \ge \tau \mid \mathcal{L}(t) \right\}  \label{eq:r_M_M1}.
\end{align}
The explicit rate expressions we are looking for are provided in the following theorem.
\begin{theorem}
\label{th:AvgRate}
Given a subset of scheduled pairs, $\mathcal{L}(t)$, the average rate of user $k$ ($\in \mathcal{L}(t)$) is:
\begin{itemize}
\item For the imperfect case, this rate can be expressed as   
\end{itemize}
\begin{align}
(1-L(t) \theta) d_k   R e^{- \frac{ \sigma^2 \tau}{\alpha_{kk} } } \mathit{MGF}_{ \! \! \mathit{RI}_{\! k}^{(m)} } \! \left(-\frac{\tau}{\alpha_{kk}} \right),
\end{align}
in which the \ac{MGF} can be written as 
\begin{align}
\mathit{MGF}_{ \! \! \mathit{RI}_{\! k}^{(m)}} \! \left(-\frac{\tau}{\alpha_{kk}} \right)= \prod\limits_{i \in \mathcal{L}(t),i\ne k}  \left( \frac{\alpha_{ki}\tau d_i  }{ \alpha_{kk}   2^{\frac{B}{Q}}   }+1 \right)^{\! -Q} \,_2F_1(\breve{b}_i,Q;\breve{a}_i+\breve{b}_i;\frac{1}{\frac{ \alpha_{kk} 2^{\frac{B}{Q}} }{ \alpha_{ki} \tau d_i} +1} ),
\end{align}
for $j=1,\ldots,D$. In the above equation, $\,_2F_1$ represents the hypergeometric function, $\breve{a}_i=\frac{(Q+1)d_i}{Q}-\frac{1}{Q}$ and $\breve{b}_i= (Q-1)\breve{a}_i$. We recall that $Q=N_\text{t} N_\text{r}-1$.
\begin{itemize}
\item For the perfect case, the average rate we are looking for can be given by
\end{itemize}
\begin{align}
(1-L(t) \theta) d_k   R e^{- \frac{ \sigma^2 \tau}{\alpha_{kk} } }.
\end{align}
\end{theorem}
\begin{proof}
For the perfect case, the statement follows directly from Proposition \ref{pr:SuccProb}. Using this proposition, it can be seen that we need to calculate $ \mathit{MGF}_{ \! \! \mathit{RI}_{\! k}^{(m)}} \! \left(-\frac{\tau}{\alpha_{kk}} \right)$ in order to prove the statement for the imperfect case. To this end, we first recall that,
using \eqref{eq:LI},  we have $\mathit{RI}_{\! k}^{(m)} =  \sum\limits_{\substack{i \in \mathcal {L}(t), i \ne k}}   \alpha_{ki} \left\| \mathbf{h}_{ki} \right\|^2  e_{ki}  \sum\limits_{j=1}^{d_i}  \left|  \mathbf{w}_{ki}^H \, \mathbf{T}_{k,i}^{(m,j)} \right|^2 $. 
Since $\mathbf{w}_{ki}$ and $\mathbf{T}_{k,i}^{(m,j)}$ are independent and identically distributed (i.i.d.) isotropic vectors in the null space of $\mathbf{\hat{h}}_{ki}$, $\left|  \mathbf{w}_{ki}^H \, \mathbf{T}_{k,i}^{(m,j)} \right|^2$ is i.i.d. $ \text{Beta}(1,Q-1)$ distributed for all $i$, where $Q=N_\text{t}N_\text{r}-1$. Hence, 
$\sum\limits_{j=1}^{d_i}  \left|  \mathbf{w}_{ki}^H \, \mathbf{T}_{k,i}^{(m,j)} \right|^2$ is the sum of $d_i$ i.i.d. Beta variables, which can be approximated to another Beta distribution \cite{Johannesson1995Approximations}. Specifically, we have $\sum\limits_{j=1}^{d_i}  \left|  \mathbf{w}_{ki}^H \, \mathbf{T}_{k,i}^{(m,j)} \right|^2 \sim d_i \, \text{Beta}(\breve{a}_i,\breve{b}_i)$, where $\breve{a}_i=\frac{(Q+1)d_i}{Q}-\frac{1}{Q}$ and $\breve{b}_i= (Q-1)\breve{a}_i$. 
According to \cite{Xiaoming2014PerformanceAnal}, $e_{ki} \left\| \mathbf{h}_{ki} \right \|^2$ is $\text{Gamma}(Q,2^{\frac{B}{Q}})$ distributed, where $Q$ and $2^{\frac{B}{Q}}$ are the shape and rate parameters, respectively.
Let $ \delta =2^{\frac{B}{Q}}$. It follows that $\mathit{RI}_{\! k}^{(m)} =  \sum\limits_{\substack{i \in \mathcal {L}(t), i \ne k}} \rho_{ki}  X_i Y_i $, with $\rho_{ki}=\alpha_{ki} d_i$, $X_i \sim \text{Gamma}(Q,\delta)$ and $Y_i \sim \text{Beta}(\breve{a}_i,\breve{b}_i)$. \\
It is clear that $X_i Y_i$ is the product of a Gamma and Beta random variables, thus the \ac{PDF}  of $P_i=X_iY_i$ is given by \cite{Nadarajah2005ProductGammaBeta} 
\begin{align}
f_{P_i}(p_i)=
\frac{ \delta^Q \Gamma(\breve{b}_i) }{\Gamma(Q) B(\breve{a}_i,\breve{b}_i)} p_i^{Q-1} e^{-\delta p_i} \Psi(\breve{b}_i,1+Q-\breve{a}_i;\delta p_i),
\end{align}
\begin{minipage}{\textwidth}
where $\Psi$ is the Kummer function defined as
\begin{align}
\Psi(a,b;x)=\frac{1}{\Gamma(a)} \int_0^\infty e^{-xt} t^{a-1} (1+t)^{b-a-1}dt,
\end{align}
and where $\Gamma(\cdot)$ is the Gamma function and $B(\cdot,\cdot)$ is the Beta function.
\end{minipage}
Therefore, the \ac{MGF} of random variable $P_i$ can be written as
\begin{align}
\mathit{MGF}_{\! P_i}(-t) &= \int\limits_{-\infty}^{+\infty} e^{-t p_i} f_{P_i}(p_i) dp_i \nonumber &\\ &=  \kappa     \int\limits_{0}^{+ \infty}  p_i^{Q-1} e^{-t p_i -\delta p_i} \Psi(\breve{b}_i,1+Q-\breve{a}_i;\delta p_i)  dp_i   \nonumber &\\ &\stackrel{(i)}{=}   \kappa   \frac{ \Gamma(Q) \Gamma(\breve{a}_i) }{ \delta^Q \Gamma(\breve{a}_i+\breve{b}_i)  } \left(  \frac{t}{\delta}+1  \right)^{\! -Q}   \,_2F_1(\breve{b}_i,Q;\breve{a}_i+\breve{b}_i;\frac{1}{1+\frac{  \delta }{ t }}) \nonumber &\\ &\stackrel{(ii)}{=}     \left(  \frac{t}{\delta}+1  \right)^{\! -Q}   \,_2F_1(\breve{b}_i,Q;\breve{a}_i+\breve{b}_i;\frac{1}{1+\frac{  \delta }{ t }}),
\end{align}
 where $\kappa = \frac{ \delta^Q \Gamma(\breve{b}_i) }{\Gamma(Q) B(\breve{a}_i,\breve{b}_i)} \vspace{1mm}$.
 The equality (i) is obtained using \cite{Bateman1954Tables}, whereas the equality (ii) holds since the Beta function $B(\breve{a},\breve{b})= \frac{ \Gamma(\breve{a}) \Gamma(\breve{b})}{ \Gamma(\breve{a}+\breve{b})}$.
 It is clear that we can write $\mathit{RI}_{\! k}^{(m)} =  \sum\limits_{\substack{i \in \mathcal {L}(t), i \ne k}} \rho_{ki} P_i $, which is the sum of weighed (independent) random variables ($P_i$) with $\rho_{ki}$ as weights.
 The \ac{MGF} of $\mathit{RI}_{\! k}^{(m)}$ at $-t$ is then given by
 \begin{align}
 \mathit{MGF}_{\! \! \mathit{RI}_{ k}^{(m)}}(-t) \nonumber &= \prod\limits_{\substack{i \in \mathcal {L}(t), i \ne k}} \mathit{MGF}_{ \! P_i}(-t \rho_{ki}) \\ &=   \prod\limits_{\substack{i \in \mathcal {L}(t), i \ne k}}  \left(\frac{\alpha_{ki}d_i t}{\delta}+1 \right)^{\! -Q}  \,_2F_1(\breve{b}_i,Q;\breve{a}_i+\breve{b}_i;\frac{1}{\frac{\delta}{ \alpha_{ki}d_i t} +1}).
 \end{align}
This results from the fact that the  moment-generating function of a sum of independent random variables is the product of the moment-generating functions of these variables.
Hence, by taking $t=\left(-\frac{\tau}{\alpha_{kk}} \right)$ and recalling that $\delta =2^{\frac{B}{Q}}$, we eventually get 
 \begin{align}
 \mathit{MGF}_{ \! \! \mathit{RI}_{\! k}^{(m)}} \! \left(-\frac{\tau}{\alpha_{kk}} \right)= \prod\limits_{i \in \mathcal{L}(t),i\ne k}  \left( \frac{\alpha_{ki}\tau d_i  }{ \alpha_{kk}   2^{\frac{B}{Q}}   }+1 \right)^{\! -Q} \,_2F_1(\breve{b}_i,Q;\breve{a}_i+\breve{b}_i;\frac{1}{\frac{ \alpha_{kk} 2^{\frac{B}{Q}} }{ \alpha_{ki} \tau d_i} +1} ).
 \end{align}
 Notice that this \ac{MGF} expression is independent of the identity of the data stream, so we can write 
 $\sum_{m=1}^{d_k}\mathit{MGF}_{ \! \! \mathit{RI}_{\! k}^{(m)} } \! \left(-\frac{\tau}{\alpha_{kk}} \right)= d_k \, \mathit{MGF}_{ \! \! \mathit{RI}_{\! k}^{(m)} } \! \left(-\frac{\tau}{\alpha_{kk}} \right)$. Hence, the desired result follows.
\end{proof}

\section{Stability Analysis for the Symmetric Case}
\label{sec:PerAnalysis_sym}

In this section, we consider a symmetric system in which the path loss coefficients have the same value, namely $\zeta=\zeta_{ki}$, $\forall k,i$, and all the pairs have equal number of data streams, namely $d=d_k$, $\forall k$; note that we still assume different average arrival rates. Under this system, the \emph{feasibility condition} of \ac{IA}, given in \cite{Yetis2010Feasibility}, becomes $N_\text{t}+N_\text{r} \ge (L+1)d$, which we assume is satisfied here. We recall that at each timeslot, for the selected pairs, rate $R$ can be supported if the \ac{SINR} at the corresponding user is greater than or equal to a given threshold $\tau$; otherwise, the assigned rate is $0$. Let $\alpha= \frac{P\zeta}{d}$. Under this specific model, the \ac{SINR} of stream $m$ at user $k$ becomes
\begin{align}
  \gamma_k^{(m)} =
  \begin{dcases}
  \, \, \frac{ \alpha \left|  \left(\mathbf{\hat{u}}_k^{(m)} \right)^{\! H} \mathbf{H}_{kk} \mathbf{\hat{v}}_k^{(m)}  \right|^2  }{  \sigma^2 +  \sum\limits_{\substack{i \in \mathcal {L}, i \ne k}}  \alpha \left\| \mathbf{h}_{ki} \right\|^2  e_{ki}  \sum\limits_{j=1}^{d}  \left|   \mathbf{w}_{ki}^H \,  \mathbf{T}_{k,i}^{(m,j)} \right|^2  },   & \qquad{} \text{imperfect case} \\ 
  \, \, \frac{ \alpha \left|  \left(\mathbf{u}_k^{(m)} \right)^{\! H} \mathbf{H}_{kk} \mathbf{v}_k^{(m)}  \right|^2  }{\sigma^2},   &  \qquad{} \text{perfect case}   \label{eq:SINR_s} 
  \end{dcases}
\end{align}
As explained in the previous sections, if $\mathcal{L}$ is the subset of scheduled pairs, the average transmission rate per active user is given by $(1-L\theta) d R \, \mathbb{P} \left\{ \gamma_k^{(m)} \ge \tau \mid \mathcal{L} \right\}$.
Relying on Theorem \ref{th:AvgRate}, we get the following results.
\subsubsection{Imperfect Case}
the \emph{average transmission rate} for an active user $k \in \mathcal{L}$ can be given by 
\begin{align}
 (1-L\theta)d  R e^{-\frac{ \sigma^2 \tau }{\alpha}}  \left(  \left(   \frac{ d \tau }{  2^{\frac{B}{Q}} }+1  \right)^{-Q}  \,_2F_1(\breve{b},Q;\breve{a}+\breve{b};\frac{1}{ \frac{ 2^{\frac{B}{Q}} }{  d \tau } +1     }) \right)^{ L-1},  \label{eq:avr_s}  
\end{align} 
where $_2F_1$ is the hypergeometric function, $\breve{a}=\frac{(Q+1)d}{Q}-\frac{1}{Q}$ and $\breve{b}= (Q-1)\breve{a}$. It can be noticed that this average rate is independent of the identity of active user $k$ and the $L-1$ other active pairs, yet depends on the cardinality $L$ of subset $\mathcal{L}$. By denoting this rate as $r(L)$, the expression in $\eqref{eq:avr_s}$ can be re-written as 
\begin{align}
\label{eq:rsimple}
r(L) =  (1-L\theta) d R  e^{-\frac{ \sigma^2 \tau }{\alpha}}    F^{L-1},  
 	\end{align} 
in which $F=  \left(   \frac{ d \tau }{  2^{\frac{B}{Q}} }+1  \right)^{-Q}  \,_2F_1(\breve{b},Q;\breve{a}+\breve{b};\frac{1}{ \frac{ 2^{\frac{B}{Q}} }{  d \tau } +1     })$. 
Consequently, the \emph{total average transmission rate} of the system is given by
\begin{align}
 r_{\text{T}}(L) = L (1-L\theta) d R  e^{-\frac{ \sigma^2 \tau }{\alpha}}     F^{L-1}.   
\end{align}  
Studying the variation of these rate functions w.r.t. the number of active pairs $L$ is essential for the stability analysis and is thus described by the following lemma.
\begin{lemma}
Given a number of users to be scheduled, $L$,
the average transmission rate is a decreasing function with $L$, whereas the total average transmission rate is increasing from $0$ to $L_0$ and decreasing from $L_0$ to $\frac{1}{\theta}$, meaning that $r_\text{T}$ reaches its maximum at $L_0$, where $L_0 <\frac{1}{2\theta}$ and is given by
\begin{align}
 L_0 = \frac{ \frac{1}{\theta} - \frac{2}{\log F} - \sqrt{ \left(   \frac{2}{\log F}  - \frac{1}{\theta}  \right)^2 + \frac{4}{ \theta \log F}   }  }{2}.        \label{eq:l1}
\end{align}
\label{le:ratevariation}
\end{lemma}
\begin{proof}
The proof is provided in Appendix \ref{app:ratevariation}.
\end{proof} 
From $\eqref{eq:l1}$ we can notice that $L_0$ is in general a real value. But, since it represents a number of users, we need to find the best and nearest integer to $L_0$, i.e. best in terms of maximizing the total average rate function. We denote this integer by $L_\text{I}$ and we assume without lost of generality that $L_\text{I}\le N$.  We propose the following simple procedure to compute $L_\text{I}$:
\begin{enumerate}[(a)]
\item Let $L_{01}=\left \lfloor{L_0}\right \rfloor$ and $L_{02}=\left \lceil{L_0}\right \rceil$ , i.e. the largest previous and the smallest following integer of $L_0$, respectively.
\item If $ r_\text{T}(L_{02}) \ge r_\text{T}(L_{11})$, put $L_\text{I} = L_{02}$; otherwise $L_\text{I} = L_{01}$.
\end{enumerate}

\subsubsection{Perfect Case}
In this case, no residual interference exists and the corresponding SNR expression is given in $\eqref{eq:SINR_s}$. Using Theorem \ref{th:AvgRate}, the \emph{average} and \emph{total average transmission rate} expressions can be given, respectively, by 
\begin{align}
 \mu(L) = (1-L\theta) d R e^{-\frac{ \sigma^2 \tau }{\alpha}}, \label{eq:pavr} 
\end{align} 
\begin{align}
 \mu_{\text{T}}(L)=  L(1-L\theta) d R e^{-\frac{ \sigma^2 \tau }{\alpha}}.  \label{eq:ptavr} 
\end{align} 
A similar observation to that given in the first case can be made here, that is the rate functions depend only on the cardinality $L$ of $\mathcal{L}$ and not on the subset itself.
Notice that $\mu(L)$ is a decreasing function with $L$, while $\mu_{\text{T}}(L)$ is concave at $\frac{1}{2\theta}$. Since $\frac{1}{2\theta}$ represents a number of pairs, we can use the procedure  proposed for the imperfect case to find the best and nearest integer to $\frac{1}{2\theta}$. For the remainder of this paper, we denote this integer by $L_\text{P}$ .

\subsection{Stability Analysis}

After presenting results on the average rate functions, we now provide a precise characterization of the stability region of the adopted system under both the imperfect and perfect cases.

\subsubsection{Imperfect Case} 
We first define the subset $S_L$ such as $S_L=\left\{\mathbf{s} \in \mathcal{S} : \| \mathbf{s} \| _1 =L \right\}$, where we recall that $\mathbf{s} \in \mathbb{Z}^{N}$ is the vector whose coordinates take values 0 or 1 (see Section \ref{sec:systemmodel}). The corresponding subset of average rate vectors is defined as $I_L=\left\{  r(L) \mathbf{s} :  \mathbf{s} \in S_L  \right\}$. For these subsets, we define the set $\mathcal{I}$ and its complementary set $\bar{\mathcal{I}}$ as $\mathcal{I}= \left\{ I_1, I_2,...,I_{L_\text{I}} \right\}$ and $\bar{\mathcal{I}}=\left\{ I_{L_\text{I}+1},...,I_{N} \right\}$. Notice that in terms of cardinality we have $|\mathcal{I}| +  |\bar{\mathcal{I}}| = |\mathcal{S}|$. Using these definitions, we can state the following lemma, which will be useful to characterize the stability region of the system.
\begin{lemma}
Each point in the set $\bar{\mathcal{I}}$ is inside the convex hull of $\mathcal{I}$. Consequently, this hull will also contain any point in the convex hull of $\bar{\mathcal{I}}$. 	\label{le:Rl}
\end{lemma}
\begin{proof}

We first give and prove the following lemma that will help us in the proof of Lemma \ref{le:Rl}.
\begin{lemma}
 \label{pro:SL}
For any point $\mathbf{s}_{i,L+1}  \in  S_{L+1}$, there exists a point on the convex hull of $S_{L}$ that is in the same direction toward the origin as $\mathbf{s}_{i,L+1}$. Furthermore,
$\mathbf{s}_{i,L+1}$ can be written as $\frac{L+1}{L} \times$ its corresponding point on the convex hull of $S_{L}$. \label{le:Sl}
\end{lemma} 
\begin{proof}
We start the proof by first defining $\mathcal{E}_{i,L}$ as the set containing the points (vectors) that only have $L$ '1' (the other coordinate values are '0') and where the positions (indexes) of these '1' are the same as those of $L$ '1' coordinates of $\mathbf{s}_{i,L+1}$.
Note that the points in $\mathcal{E}_{i,L}$ are all different from each other.
The cardinality of $\mathcal{E}_{i,L}$, which is denoted by $\left| \mathcal{E}_{i,L} \right|$, is nothing but the result of the combination of $L+1$ elements taken $L$ at a time without repetition, and it can be computed as the following
\begin{align}
\left| \mathcal{E}_{i,L} \right|=\dbinom{L+1}{L} =\frac{(L+1)!}{L!(L+1-L)!}=L+1.
\end{align}
Thus, we have $L+1$ elements from $S_L$ that if we take them in a specific convex combination, we get a point on the same line (from the origin) as that of $\mathbf{s}_{i,L+1}$. This can be represented by
 \begin{align}
\sum\limits_{j \in \mathcal{E}_{i,L} } \delta_{j} \mathbf{s}_{j,L} \equiv \mathbf{s}_{i,L+1}, \label{eq:otsl} 
 \end{align}
where $\equiv$ is a notation used to represent the fact that these two points are on the same line from the origin, and where $\sum_{j \in \mathcal{E}_{i,L} } \delta_{j}=1$ and $\delta_{j}\ge 0$.
Let us suppose that all the coefficients $\delta_j = \frac{1}{L+1}$. This assumption satisfies the above constraints, namely $\sum_{j \in \mathcal{E}_{i,L} } \delta_{j}=1$ and $\delta_{j} \ge 0$. By replacing these coefficients in the term at the left-hand-side of $\eqref{eq:otsl}$, we get 
\begin{align}
 \sum\limits_{j \in \mathcal{E}_{i,L} } \delta_{j} \mathbf{s}_{j,L} = \frac{1}{L+1} \sum\limits_{j \in \mathcal{E}_{i,L} } \mathbf{s}_{j,L}= \frac{L}{L+1} \mathbf{s}_{i,L+1},
\end{align}
where the second equality holds since we have $L+1$ elements to sum (due to the fact that $\left| \mathcal{E}_{i,L} \right|=L+1$), each of which contains $L$ '1' at the same positions as $L$ '1' coordinates of $\mathbf{s}_{i,L+1}$, and (these elements) differ from each other in the position of one '1' (and consequently of one '0'); for instance, suppose that $N=5$, $L=2$ and $\mathbf{s}_{i,L+1}=(1,1,1,0,0)$, then the points $\mathbf{s}_{j,L}$ are given by the subset $\mathcal{E}_{i,L} = \left\{ (1,1,0,0,0);(1,0,1,0,0);(0,1,1,0,0) \right\}$. The sum corresponding to each coordinate is then equal to $L$. To complete the proof, it remains to show that $\frac{1}{L+1} \sum_{j \in \mathcal{E}_{i,L}} \mathbf{s}_{j,L}$ is on the convex hull of $S_{L}$. To this end, note that all the points in $S_{L}$ are on the same hyperplane (in $\mathbb{R}^N_+$), which is described by the equation $\sum_{k=1}^N \nu_k - L=0$; $\nu_k$ represents the $k$-th coordinate. Hence, a point on the convex hull of $S_{L}$ is also on this hyperplane. If we compute $\sum_{i=k}^N \nu_k$ for point $\frac{1}{L+1} \sum_{j \in \mathcal{E}_{i,L} } \mathbf{s}_{j,L}$, it yields $\frac{(L+1)L}{L+1} = L$ due to the definition  of $\mathcal{E}_{i,L}$, thus this point is on the defined hyperplane and consequently on the convex hull of $S_{L}$.

In order to better understand the result of this lemma, we provide a simple example for which the geometric illustration is in Figure \ref{fig:SnvsSn+1}. In this example, we take $N=2$, $S_1=\{ (1,0);(0,1) \}$ and $S_2=\{ (1,1) \}$. In addition, we define points $P_2=(1,1)$ and $P_1=(\frac{1}{2},\frac{1}{2})$. Note that $P_2 \in S_2$ and $P_1$ is on the convex hull of $S_1$. We can express $P_2$ as $ P_2= \frac{2}{1}[\frac{1}{2}(1,0)+\frac{1}{2}(0,1)]=2(\frac{1}{2},\frac{1}{2})=\frac{2}{1}P_1$. Thus, $P_2$ equals $\frac{2}{1} \times$ its corresponding point ($P_1$) on the convex hull of $S_1$.
      
 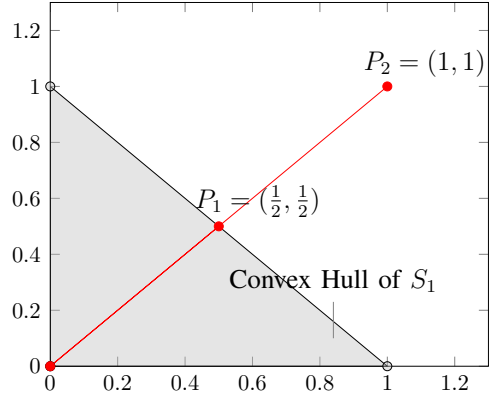
\begin{figure}[ht!]
 \centering
 \begin{tikzpicture}[scale=0.85]
  	\begin{axis}
  		[
  	    ymin=0,ymax=1.3,
  		xmin=0,xmax=1.3,
        ]  			
  	
  	\addplot[fill=lightgray,fill opacity=0.4,mark=*] coordinates {
     	(0,0)
  		(0,1)
  		(1,0)
  		(0,0)	
  	};
  	\addplot[color=red,mark=*] coordinates {
  	    (1,1)
  	  	(0,0)
  	  	(0.5,0.5)	
  	};	
  	 
  	 \node at (axis cs:1.32, 1) [anchor=south east] {\normalsize  $P_2=(1,1)$};
  	 \node at (axis cs:0.83, 0.5) [anchor=south east] {\normalsize  $P_1=(\frac{1}{2},\frac{1}{2})$};
  	 \node[coordinate,pin=above:{Convex Hull of $S_1$ }] 
  	 		at (axis cs:0.84,0.1) {};
  	\end{axis}
 \end{tikzpicture}
   \captionsetup{font=small}
   \caption{Example that illustrates the result of Lemma \ref{le:Sl}.}
  \label{fig:SnvsSn+1}
 \end{figure}  
This completes the proof of Lemma \ref{le:Sl}.
\end{proof}

Now, using the above lemma, a point $\mathbf{s}_{i,L+1}$ in $S_{L+1}$ can be expressed in function of $L+1$ specific points in $S_L$ as $\mathbf{s}_{i,L+1}= \frac{L+1}{L}\sum_{j \in \mathcal{E}_{i,L}} \delta_{j} \mathbf{s}_{j,L}$, which implies that
\begin{align} 
r(L+1) \mathbf{s}_{i,L+1}= r(L+1) \frac{L+1}{L}\sum_{j \in \mathcal{E}_{i,L}} \delta_{j} \mathbf{s}_{j,L},
\end{align}
where the definition of $\mathcal{E}_{i,L}$ can be found in the proof of Lemma \ref{le:Sl}.
By Lemma $\ref{le:ratevariation}$, we have $(L+1) r(L+1) < L r(L) $ for $L \ge L_\text{I}$. We thus get
\begin{align}
r(L+1) \frac{L+1}{L}\sum_{j \in \mathcal{E}_{i,L}} \delta_{j} \mathbf{s}_{j,L} < r(L) \frac{L}{L}\sum_{j \in \mathcal{E}_{i,L}} \delta_{j} \mathbf{s}_{j,L}=  r(L) \sum_{j \in \mathcal{E}_{i,L}} \delta_{j} \mathbf{s}_{j,L}. \label{eq:samline}
\end{align}
Note that the inequality operator in $\eqref{eq:samline}$ can be used since the two compared points are on the same line (from the origin).
Therefore, each point in $I_{L+1}$ is in the convex hull of $I_{L}$, for $L \ge L_\text{I}$, since $r(L+1) \mathbf{s}_{i,L+1} \in I_{L+1}$ and $r(L) \sum_{j \in \mathcal{E}_{i,L}} \delta_{j} \mathbf{s}_{j,L}$ is in the convex hull of $I_{L}$. Consequently, all the points in $I_{L+1}$ for $L \ge L_\text{I}$ (i.e. these points form $\bar{ \mathcal{I} }$) are in the convex hull of $I_{L_\text{I}}$, which is a subset of $\mathcal{I}$. Therefore, the desired result holds.

In the following, we illustrate the result of this lemma for the case where $N=2$ and $L_\text{I}=1$. For this example, we have $\mathcal{I}=\{(0,0);(r(1),0);(0,r(1))  \}$ and $\bar{\mathcal{I}}=\{ (r(2) ,r(2)) \}$. In addition, using Lemma $\ref{le:ratevariation}$, we can write $2 r(2)< r(1)$. 
From Figure \ref{fig:Rl}, we can easily notice that $\bar{\mathcal{I}}$ is in the convex hull of $\mathcal{I}$.
\begin{figure}[ht!]
  \centering
\begin{tikzpicture}[scale=0.85]
    	\begin{axis}
    		[ticks=none,
    	    ymin=0,ymax=1.3,
    		xmin=0,xmax=1.3,
    		xlabel={Mean Arrival Rate $a_1$},
    		xlabel near ticks,  
    		xlabel shift=10pt,  
    		ylabel={Mean Arrival Rate $a_2$},
    		ylabel near ticks,
    		ylabel shift=10pt]
    		  
    	\addplot[fill=lightgray,fill opacity=0.4,mark=*] coordinates {
       	(0,0)
    		(0,1)
    		(1,0)
    		(0,0)	
    	};
    	\addplot[color=red,mark=*] coordinates {
    	    (0.3,0.3)
    	  	(0,0)
    	  	(0.5,0.5)	
    	};		
    	\addplot[color=red,mark=x] coordinates {
    	     (1,1)
    	     (0,0)	
    	};		
    	\node at (axis cs:0.95, 0) [anchor=south west] {\normalsize  $(r(1),0)$};
    	\node at (axis cs:0, 1) [anchor=south west] {\normalsize  $(0,r(1))$};
    	\node at (axis cs:0.41, 0.31) [anchor=south east] {\normalsize  $(r(2),r(2))$};
    	\node at (axis cs:0.45, 0.5) [anchor=south west] {\normalsize  $(\frac{r(1)}{2},\frac{r(1)}{2})$};
    	 \node[coordinate,pin=above:{Convex Hull of $\mathcal{I}$ }] 
    	   	 		at (axis cs:0.84,0.1) {};
    	 
    	\end{axis}
\end{tikzpicture}
  \captionsetup{font=small}
  \caption{Example that shows the result of Lemma \ref{le:Rl}.}
  \label{fig:Rl}
\end{figure}
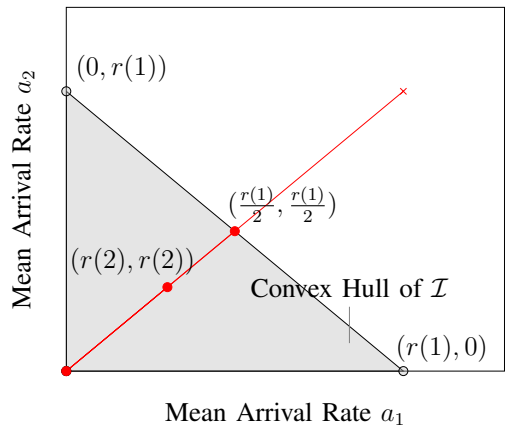   
  
This completes the proof.
\end{proof}

Now, we have all the required materials to characterize the stability region of the considered system.
We recall that the stability region is the set of all mean arrival rate vectors for which the system is strongly stable. Here, this region is given by the following theorem.
\begin{theorem}	
\label{th:stabregionSI}
The stability region of the system in the symmetric case with  limited backhaul can be characterized as
\begin{align}
\Lambda_{\text{I}}=\mathcal{CH} \left\{ I_1,I_2,...,I_{L_\text{I}} \right\} =\mathcal{CH} \left\{\mathcal{I}\right\}, 
\end{align}
where $\mathcal{CH}$ represents the closed convex hull.
\end{theorem}
\begin{proof}
First we prove that this region is achievable. Indeed,
a point $\mathbf{r}_\text{I}$ in $\Lambda_\text{I}$ can be written as the convex combination of the points in $\mathcal{I}$ as $\mathbf{r}_\text{I} = \sum_{i=1}^{|\mathcal{I}|} p_i \mathbf{r}_i$, where $\mathbf{r}_i$ represents a point in $\mathcal{I}$, $p_i \ge 0$ and $\sum_{i=1}^{|\mathcal{I}|} p_i=1$. Note that each point $\mathbf{r}_i$ represents a different scheduled subset of pairs. To achieve $\mathbf{r}_\text{I}$ it suffices to use a randomized policy that at the beginning of each timeslot selects (decision) $\mathbf{r}_i$ with  probability $p_i$. Since $\mathbf{r}_\text{I}$ is an arbitrary point in  $\Lambda_\text{I}$, we can claim that this region is achievable.
 
We then have to prove the converse, that is if there exists a centralized policy that stabilizes the system for a mean arrival rate vector $\mathbf{a}$, then $\mathbf{a} \in  \Lambda_\text{I}$. To this end, assume the system is stable for a mean arrival rate vector $\mathbf{a}$. As explained earlier, the scheduling decision (i.e. subset $\mathcal{L}$) under the centralized policy depends on the queues only, so we show this dependency by $\mathcal{L}(\mathbf{q})$. 
Let us denote by $\mathbf{r}_\text{s}$ the mean service rate vector, which can be given by 
$\lim\limits_{T \to \infty} \frac{1}{T} \sum_{t=0}^{T-1} \mathbf{B}(t)$, where $\mathbf{B}(t)$ is the vector for which the $k$-th component is given by $B_k(t)$. In addition, we denote by
$\mathbf{r}(\mathcal{L}(\mathbf{q}))$ the average rate if the queue state is $\mathbf{q}$ and the selected subset of pairs is $\mathcal{L}(\mathbf{q})$.
It is obvious that the set of all possible values of $\mathbf{r}(\mathcal{L}(\mathbf{q}))$ is nothing but $\mathcal{I} \cup \bar{\mathcal{I}}$.
Under the adopted model, the system can be described as a Markov chain, and since it is stable, it has a stationary distribution, which we denote by $\pi(\mathbf{q})$.
The mean service rate vector can then be expressed as the following
\begin{align}
\mathbf{r}_\text{s}= \sum\limits_{\mathbf{q} \in \mathbb{Z}_+^{N}} \pi(\mathbf{q}) \mathbf{r} (\mathcal{L}(\mathbf{q}))  = \sum\limits_{ \mathcal{L} \in \bm{\mathcal{L}} } \mathbf{r}(\mathcal{L}) \sum\limits_{\mathbf{q} \in \mathbb{Z}_+^{N}: \mathcal{L}(\mathbf{q})=\mathcal{L}  }  \pi(\mathbf{q})  > \mathbf{a},
\end{align}
where the operator $>$ is component-wise.
By setting $p(\mathcal{L})= \sum\limits_{\mathbf{q} \in \mathbb{Z}_+^{N}:\mathcal{L}(\mathbf{q})=\mathcal{L}  }  \pi(\mathbf{q})$ and  noticing that the set of all possible values of $\mathbf{r}(\mathcal{L})$ is the same as $\mathbf{r}(\mathcal{L}(\mathbf{q}))$, that is $\mathcal{I} \cup \bar{\mathcal{I}}$ , the mean service rate can be re-written as 
\begin{align}
\mathbf{r}_\text{s} = \sum\limits_{j=1}^{\left| \mathcal{\mathcal{I} \cup \bar{\mathcal{I}}} \right|} p_j  \mathbf{r}_j,
\end{align}
in which $j$ is used to denote decision $ \mathcal{L}$, meaning that $p_j=p(\mathcal{L})$, and $\mathbf{r}_j$ represents a point in set $\mathcal{I} \cup \bar{\mathcal{I}}$, and where $\left| \mathcal{\mathcal{I} \cup \bar{\mathcal{I}}} \right|$ represents the cardinality of this set. Hence, we can state that $\mathbf{r}_\text{s}$ is in the convex hull of $\mathcal{I} \cup \bar{\mathcal{I}}$. But, since we have demonstrated that $\bar{\mathcal{I}}$ is in the convex hull of $\mathcal{I}$ (see Lemma \ref{le:Rl}), we have $\mathbf{r}_\text{s} \in \Lambda_\text{I}$ and consequently $\mathbf{a} \in  \Lambda_\text{I}$. This completes the proof.
\end{proof}

Unlike classical results in which the stability region is given by the convex hull over all possible decisions, here the characterization is more precise and is defined by the decision subsets $S_{L}$ for all $L\le L_\text{I}$. In addition, this theorem provides an exact specification of the corner points (vertices) of the stability region, meaning that this region is characterized by the set $\mathcal{I}$ and not by the whole space $\mathcal{I} \cup \bar{\mathcal{I}}$. An additional point to note is that $\Lambda_{\text{I}}$ is a convex polytope in the $N$-dimensional space $\mathbb{R}_+^N$.

In order to choose the active pairs at each timeslot, we use the Max-Weight scheduling policy defined earlier (see $\eqref{eq:MW}$). Under the symmetric and imperfect case, this policy becomes
\begin{align}
\Delta^{\text{*}}_{\text{I}}  : \mathbf{s}(t)= \operatorname*{arg\,max}_{ \mathbf{s} \in   \mathcal{S}  }  \left\{ r(  \left\| \mathbf{s} \right\|_1) \, \mathbf{s} \cdot \mathbf{q}(t)  \right\},
\end{align}
where $\left\| \mathbf{s} \right\|_1$ gives the number of '1' coordinates in $\mathbf{s}$ (or equivalently, the number of active pairs $L$). Recall that these non-zero coordinates indicate which pairs to schedule.
For the proposed policy, the following proposition holds.
\begin{proposition}
\label{pro:optimality}
 The scheduling policy $\Delta^{\text{*}}_{\text{I}}$ is throughput optimal. In other words, $\Delta^{\text{*}}_{\text{I}}$ stabilizes the system for every arrival rate vector $\mathbf{a} \in \Lambda_{\text{I}} $.
\end{proposition}
\begin{proof}
The proof can be done in the same way as the proof of Lemma \ref{le:OptSchedPolicy}.
\end{proof}
Based on the analysis done at the end of Section \ref{sec:systemmodel}, it was shown that applying policy $\Delta^{\text{*}}$ will result in a computational complexity (CC) of $O(N2^N)$. The same holds here for policy $\Delta^{\text{*}}_{\text{I}}$. Consequently, a moderately large $N$ will lead to considerably high CC. But, we recall that this analysis corresponds to the classical implementation of the Max-Weight algorithm, that is finding the maximum over $2^N$ products of two vectors. However, in our case the implementation of this algorithm does not require all this complexity. This is due to the fact that all the active users have the same average transmission rate. This structural property allows us to propose an equivalent reduced CC implementation of $\Delta^{\text{*}}_{\text{I}}$, which we provide in Algorithm \ref{al:schedecision}.
\begin{algorithm}[ht!]
 \caption{: A Reduced Computational Complexity Implementation of $\Delta^{\text{*}}$}
 \begin{algorithmic}[1]
 		\State Initialize $L_\text{s}=0$.
 		\State Sort the queues in a descending order.
 		\For {$l=1:1:N$}
 		\State Consider $sum_l=$ sum of the first $l$ queues.
 		\If {$r(l) \, sum_l >  r({L_\text{s}}) \, sum_{L_\text{s}}$} 
 		\State put $L_\text{s}=l$
 		\EndIf
 		\EndFor
 	    \State Schedule pairs corresponding to the first $L_\text{s}$ queues.
 \end{algorithmic}
 	\label{al:schedecision}
\end{algorithm} 

The proposed implementation depends essentially on two steps: the “sorting algorithm” and the “for loop”. From the literature, the complexity for the “sorting algorithm” can be given by $O( N^2 )$. For the “for loop”, the (worst case) complexity is also $O(N^2)$ since this loop is executed $N$ times (i.e. iterations) and every iteration has another dependency to $N$. Therefore, the computational complexity of the proposed implementation is $O(N^2 + N^2)$, or equivalently $O(N^2)$, which is very small compared to $O(N 2^{N})$, especially for large $N$. 
 
\subsubsection{Perfect Case} 
A similar study to that done for the imperfect case can be adopted here. To begin with, we define $P_L=\left\{ \mu(L) \mathbf{s} :  \mathbf{s} \in S_L  \right\}$; we recall that $S_L=\left\{\mathbf{s} \in \mathcal{S} : \| \mathbf{s} \| _1 =L \right\}$. 
As seen earlier for this case, the total average rate given in $\eqref{eq:ptavr}$ reaches its maximum at $\frac{1}{2\theta}$ for which  the best and nearest integer is denoted as $L_{\text{P}}$, where we assume without lost of generality that $L_{\text{P}} \le N$. In addition, the average rate $\mu(L)$  decreases with $L$. Under these observations, the stability region can be characterized as
\begin{theorem}
For the symmetric system with unlimited backhaul, the stability region can be defined as the following
\begin{align}
\Lambda_\text{P}=\mathcal{CH} \left\{ P_1,P_2,...,P_{L_{\text{P}}} \right\}. 
\end{align}
\end{theorem}
\begin{proof}
The proof is very similar to that of Theorem \ref{th:stabregionSI}, just consider the average rate functions $\mu(L)$ and  $\mu_\text{T}(L)$ instead of $r(L)$ and $r_\text{T}(L)$; so we omit this proof to avoid redundancy.
\end{proof}
To achieve the stability region that is characterized in the above, we use the Max-Weight rule, which yields the optimal policy given by
\begin{align}
\Delta^{\text{*}}_{\text{P}} : \mathbf{s}(t)=  \operatorname*{arg\,max}_{ \mathbf{s} \in   \mathcal{S}  }  \left\{ \mu(  \left\| \mathbf{s} \right\|_1) \, \mathbf{s} \cdot \mathbf{q}(t)  \right\}.
\end{align}
As for the imperfect case, applying this policy using its classical implementation will result in a CC of $O(N2^N)$. Hence, to avoid a high complexity for large $N$, and since the structural properties of this policy and those of policy $\Delta^{\text{*}}_{\text{I}}$ are similar, the equivalent implementation proposed for the imperfect case can be applied here but after replacing $r(L)$ with $\mu(L)$. Consequently, we get a reduced complexity of $O(N^2)$.

\subsubsection{Compare the Imperfect and Perfect Cases in terms of Stability}

After having characterized the stability region for both the perfect and imperfect cases, we now investigate the gap between these two regions. This gap can be interpreted as the impact of having limited backhaul, and thus quantization, on the stability region of the system with unlimited capacity. It is straightforward that the quantization process will result in shrinking the stability region compared with that of the perfect case. To capture this shrinkage, we find the minimum fraction that the imperfect case achieves w.r.t. the stability region achieved in the perfect case. This fraction corresponds to the maximum gap.

To begin with, we first draw the attention to the fact that in addition to having $\mu(L) \ge r(L)$, we generally have $L_\text{P} > L_\text{I}$.  In order to provide some insights into how we will derive the required fraction, in Figure \ref{fig:ex1} we depict the general shapes of the two stability regions for a simple example where $L_\text{I}=1$ and $L_\text{P}=N=2$.
\begin{figure}[ht!]
  \centering
   \begin{tikzpicture}[scale=0.85]
    \begin{axis}
    		[ticks=none,
    	    ymin=0,ymax=1.7,
    		xmin=0,xmax=1.7,
    		xlabel={Mean Arrival Rate $a_1$},
    		xlabel near ticks,  
    		xlabel shift=10pt,  
    		ylabel={Mean Arrival Rate $a_2$},
    		ylabel near ticks,
    		ylabel shift=10pt]
    		  
    	\addplot[pattern=dots,opacity=0.4,mark=*] coordinates {
       	            (0,0)
    		        (0,1.3)
    		        (1.1,1.1)
    		        (1.3,0)
    		        (0,0)	
    	};
    	\addplot[fill opacity=0.4,fill=lightgray,mark=*] coordinates {
    	       	    (0,0)
    	    		(0,1.3)
    	    		(1.3,0)
    	    		(0,0)	
    	};
    	\addplot[color=red,mark=*] coordinates {
    	    (0,0)
    	    (0.65,0.65)
    	  	(1.1,1.1)	
    	};		
 		
    	\node at (axis cs:0.88, 0) [anchor=south west] {\normalsize $(r(1),0)$};
    	\node at (axis cs:0, 1.03) [anchor=south west] {\normalsize $(0,r(1))$};
    	\node at (axis cs:0.14, 0.52) [anchor=south west] {\normalsize $(\frac{r(1)}{2},\frac{r(1)}{2})$};
    	\node at (axis cs:1.26, 0) [anchor=south west] {\normalsize $(\mu(1),0)$};
    	\node at (axis cs:0, 1.25) [anchor=south west] {\normalsize $(0,\mu(1))$};
    	\node at (axis cs:1.5, 1.1) [anchor=south east] {\normalsize $(\mu(2),\mu(2))$};
    	 
   	\end{axis}
   \end{tikzpicture}
   \captionsetup{font=small}
   \caption{An illustration of the stability regions of the perfect (dotted region) and imperfect (gray region) cases for the symmetric system. Here, $L_\text{I}=1$ and $L_\text{P}=N=2$.}
   \label{fig:ex1}
\end{figure}
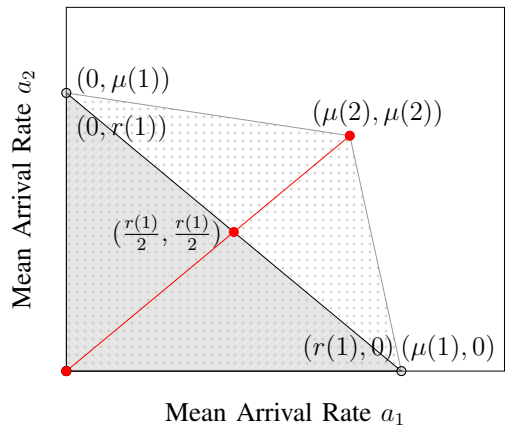   
From this figure, we can observe that we have different gaps over different  directions. To find the minimum fraction (i.e. maximum gap), we adopt the following approach. We take any point from subset $P_{L_{\text{P}}}$, and then we try to see how far is this point from the convex hull $\Lambda_{\text{I}}$ in the direction toward the origin. 
This is due to three reasons: (i) $I_{L_\text{I}}$ is a subset of the vertices  that characterize the convex hull of the imperfect case ($\Lambda_{\text{I}}$), (ii) $P_{L_{\text{P}}}$ is the subset that contains points (vertices) on the convex hull of the perfect case, and (iii) the points in $P_{L_{\text{P}}}$ are the farthest from $\Lambda_{\text{I}}$. Using this approach, we can state the following theorem. 
\begin{theorem}
\label{th:imvsp_frac}
For the symmetric system, the stability region in the imperfect case achieves at least a fraction $\frac{L_\text{I} r(L_\text{I})  }{ L_{\text{P}} \mu(L_{\text{P}}) }$ (which is $<1$) of the stability region achieved in the perfect case. Notice that this fraction is nothing but $\frac{ r_\text{T}(L_\text{I})  }{  \mu_\text{T}(L_{\text{P}}) }$.
\end{theorem}
\begin{proof}
We start the proof by first giving and proving the following result.
\begin{lemma}
\label{pro:SLn}
Each point in $S_{L+1}$ can be written as $\frac{L+1}{L+1-n} \times$ some point on the convex hull of $S_{L+1-n}$, for $1\le n \le L$.
\end{lemma}
\begin{proof}
From Lemma \ref{pro:SL}, a point in $S_{L+1}$ can be expressed in function of some subset of points, represented by $\mathcal{E}_{i,L}$, in $S_{L}$ as $\mathbf{s}_{i,L+1}=\frac{L+1}{L}\sum_{i_1 \in \mathcal{E}_{i,L}} \delta_{i_1,L} \mathbf{s}_{i_1,L}$. More specifically, we found that the coefficients $\delta_{i_1,L}$ are all equal to $\frac{1}{L+1}$, and thus $\mathbf{s}_{i,L+1}=\frac{L+1}{L}\sum_{i_1 \in \mathcal{E}_{i,L}} \frac{1}{L+1} \mathbf{s}_{i_1,L}=\frac{1}{L}\sum_{i_1 \in \mathcal{E}_{i,L}}  \mathbf{s}_{i_1,L}$. 
Similarly, each point $\mathbf{s}_{i_1,L}$  ($\in S_{L}$) can be written in function of some specific subset of points, denoted by $\mathcal{E}_{i_1,L-1}$, in $S_{L-1}$ as $\mathbf{s}_{i_1,L}=\frac{1}{L-1} \sum_{i_2 \in \mathcal{E}_{i_1,L-1}} \mathbf{s}_{i_2,L-1}$. Following this reasoning until index $L-(n-1)$, we can express $\mathbf{s}_{i,L+1}$ as 
\begin{align}
\mathbf{s}_{i,L+1}  &=  \frac{1}{L(L-1) \ldots
(L-(n-1))}\sum\limits_{i_1 \in \mathcal{E}_{i,L}} \sum\limits_{i_2 \in \mathcal{E}_{i_1,L-1}}    \ldots   \sum\limits_{i_n \in \mathcal{E}_{i_{n-1},L-(n-1)}}  \mathbf{s}_{i_{n},L-(n-1)}.   \label{eq:proofSLn1}
 \end{align}
We denote by $\mathcal{E}_{n-1,L-(n-1)}$ the subset containing the points (vectors) that only have $L-n$ '1' (the other coordinate values are '0') and where the positions (indexes) of these '1' coordinates are the same as the positions of $L-n$ '1' coordinates of $\mathbf{s}_{i,L+1}$. It is important to point out that the elements in $\mathcal{E}_{n-1,L-(n-1)}$ are all different from each other.
It can be easily noticed that each $\mathcal{E}_{i_{n-1},L-(n-1)}$ (in $\eqref{eq:proofSLn1}$) is a subset of $\mathcal{E}_{{n-1},L-(n-1)}$. The cardinality of this latter set is the result of the combination of $L+1$ elements taken $L-(n-1)$ at a time without repetition, thus we get the following 
\begin{align}
\left| \mathcal{E}_{n-1,L-(n-1)} \right| = \dbinom{L+1}{L-(n-1)} =\frac{(L+1)!}{(L-(n-1))! \, n!}.
\end{align}
Let us now examine the nested summation in the expression in $\eqref{eq:proofSLn1}$. We can remark that the summands are the elements of $\mathcal{E}_{n-1,L-(n-1)}$. Hence, the result of this nested summation is nothing but a simple sum of the vectors in $\mathcal{E}_{n-1,L-(n-1)}$, each of which multiplied by the number of times it appears in the summation. For each vector, this number is the result of the number of possible orders in which we can remove $(L+1-(L-(n-1)))$ particular '1' coordinates from $\mathbf{s}_{i,L+1}$. It follows that the  required numbers are all equal to each other and given by $(L+1-(L-(n-1)))!=n!$.
From the above and the fact that $(L(L-1) \ldots
(L-(n-1)))^{-1}= (L-n)! (L!)^{-1}$, the expression in $\eqref{eq:proofSLn1}$ can be rewritten as 
\begin{align}
\mathbf{s}_{i,L+1} \nonumber &= \frac{(L-n)!}{L!} \sum_{j \in \mathcal{E}_{n-1,L-(n-1)} } n! \, \mathbf{s}_{j,L-(n-1)} \\
 &= \frac{(L-n)! \, n!}{L!}  \frac{(L+1)!}{(L-(n-1))! \, n!} \sum_{j \in \mathcal{E}_{n-1,L-(n-1)} } \frac{(L-(n-1))! \, n!}{(L+1)!}  \mathbf{s}_{j,L-(n-1)},  \label{eq:proofSLn2}
\end{align}
where the second equality is due to multiplying and dividing by $\left|  \mathcal{E}_{n-1,L-(n-1)} \right|$.
By noticing that the factor that multiplies the summation is equal to $\frac{L+1}{L-(n-1)}$, $\eqref{eq:proofSLn2}$ can be re-expressed as
\begin{align}
\mathbf{s}_{i,L+1} = \frac{L+1}{L-(n-1)} \sum_{j \in \mathcal{E}_{n-1,L-(n-1)} } \frac{(L-(n-1))! \, n!}{(L+1)!}  \mathbf{s}_{j,L-(n-1)} \label{eq:proofSLn3}.
\end{align}
From the proof of Lemma \ref{le:Sl}, we can claim that the point formed by the convex combination $\sum_{j \in \mathcal{E}_{n-1,L-(n-1)} }  \left|  \mathcal{E}_{n-1,L-(n-1)} \right|^{-1} \mathbf{s}_{j,L-(n-1)}$ is on the convex hull of $S_{L-(n-1)}$ and in the same direction from the origin as $\mathbf{s}_{i,L+1}$; this combination is convex since we have $\left|  \mathcal{E}_{n-1,L-(n-1)} \right|^{-1} > 0$ and  $\sum_{j \in \mathcal{E}_{n-1,L-(n-1)} }  \left|  \mathcal{E}_{n-1,L-(n-1)} \right|^{-1}=1$, meaning that the coefficients of this combination are non-negative and sum to $1$.  

In order to clarify the result of this lemma, we provide a simple example in which we set $N=4$ and $L+1=4$. Under this example, we know that $S_4$ will contain one point, namely $\mathbf{s}_{i,4}= (1,1,1,1)$. For this point, we want to find its corresponding point on the convex hull of $S_2$; this implies that $n=2$. 
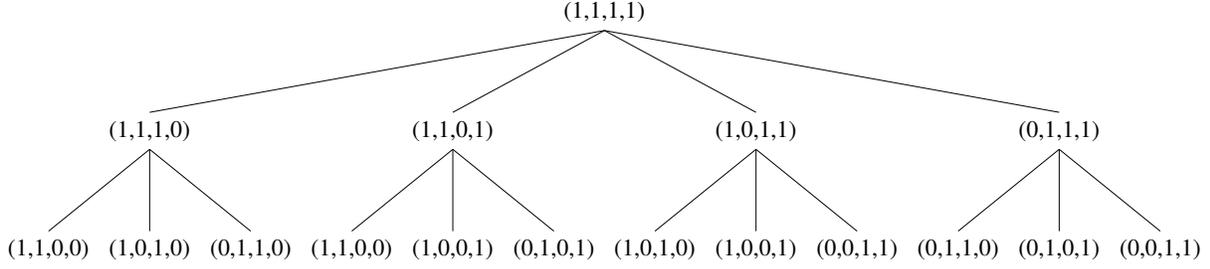
\begin{figure}
\resizebox{0.98\textwidth}{!}{%
\begin{tikzpicture}[level distance=60pt]
\Tree [.(1,1,1,1) 
             [.(1,1,1,0) [.(1,1,0,0) ] [.(1,0,1,0) ]  [.(0,1,1,0) ]]
             [.(1,1,0,1) [.(1,1,0,0) ] [.(1,0,0,1) ]  [.(0,1,0,1) ]]
             [.(1,0,1,1) [.(1,0,1,0) ] [.(1,0,0,1) ]  [.(0,0,1,1) ]] 
             [.(0,1,1,1) [.(0,1,1,0) ] [.(0,1,0,1) ]  [.(0,0,1,1) ]]      ]          
\end{tikzpicture}   }
\captionsetup{font=small}
   \caption{ A tree that shows the vectors in $S_2$ that yield $\mathbf{s}_{i,4}= (1,1,1,1)$. Here, $N=4$ and $n=2$.  }
\label{fig:tree}
\end{figure}
From the tree in Figure \ref{fig:tree}, it can be seen that
\begin{align}
\mathbf{s}_{i,4}\nonumber &=\frac{1}{3}((1,1,1,0)+(1,1,0,1)+(1,0,1,1)+(0,1,1,1)) \\ &=
\frac{1}{3}((1,1,0,0)+ (1,0,1,0)+(0,1,1,0)+(1,0,0,1)+ (0,1,0,1)+(0,0,1,1) ).
\end{align}
Remark that the $6$ different vectors in the second equality form the set
$\mathcal{E}_{n-1,L-(n-1)}=\mathcal{E}_{1,2}$, thus $\left|  \mathcal{E}_{1,2} \right|  =6 $. Using $\mathcal{E}_{1,2}$,       
the point that corresponds to $\mathbf{s}_{i,4}$ and that lies on the convex hull of $S_2$ is given by  
\begin{align}
\frac{1}{6}(1,1,0,0)+ \frac{1}{6}(1,0,1,0)+\frac{1}{6}(0,1,1,0)+\frac{1}{6}(1,0,0,1)+ \frac{1}{6}(0,1,0,1)+\frac{1}{6}(0,0,1,1).
\end{align}
We can obtain $\mathbf{s}_{i,4}$ by just multiplying this convex combination by a factor of $2$. This verifies the general formula provided in $\eqref{eq:proofSLn3}$.
\end{proof}
To find the minimum achievable fraction between the stability region of the imperfect case ($\Lambda_\text{I}$) and the stability region  of the perfect case ($\Lambda_\text{P}$), we examine the gap between each vertex that contributes in the characterization of $\Lambda_\text{P}$, where the set of these vertices is given by $\left\{ P_1,\ldots,P_{L_\text{P}} \right\}$, and the convex hull of $\Lambda_\text{I}$. 
To begin with, using the above lemma, we recall that a point in $S_{L_\text{P}}$ can be written in function of some point that lies on the convex hull of $S_{L_\text{I}}$, where these two points are in the same direction toward the origin. Furthermore, the gap between these two points can be captured using the fraction $\frac{L_\text{I}}{L_\text{P}}$.
Since any point in $P_{L_\text{P}}$ can be written as $r(L_\text{P})$ times its corresponding point in $S_{L_\text{P}}$ and a point on the convex hull of $I_{L_\text{I}}$ is $r(L_\text{I})$ times its corresponding point on the convex hull of $S_{L_\text{I}}$, we can claim that the fraction between any point in $P_{L_\text{P}}$ and its corresponding point on the convex hull of $I_{L_\text{I}}$, and thus on the convex hull of $\Lambda_\text{I}$, is given by 
\begin{align}
\frac{L_\text{I} r(L_\text{I})}{L_\text{P} \mu(L_\text{P})}=\frac{ r_\text{T}(L_\text{I})}{ \mu_\text{T}(L_\text{P})}. 
\end{align}
More generally, using the above approach, we can show that the fraction between any point (vertix) in $P_{L}$, for $ L_\text{I} \le L\le L_\text{P} $, and the convex hull of $\Lambda_\text{I}$ is equal to $\frac{L_\text{I} r(L_\text{I})}{L \mu(L)}=\frac{ r_\text{T}(L_\text{I})}{ \mu_\text{T}(L)}$.
For these fractions, since $\mu_\text{T}(L)$ increases for $L \le L_\text{P}$, the following holds
\begin{align}
\frac{ r_\text{T}(L_\text{I})}{ \mu_\text{T}(L_\text{P})} < \frac{ r_\text{T}(L_\text{I})}{ \mu_\text{T}(L_\text{P}-1)} < \ldots < \frac{ r_\text{T}(L_\text{I})}{ \mu_\text{T}(L_\text{I})}. \label{eq:inequality1}
\end{align}
On the other side, the fraction between any vertix in $P_L$, for $1 \le L \le L_\text{I}$, and its corresponding point on the convex hull of $\Lambda_\text{I}$ is given by $\frac{L r(L)}{L \mu(L)}=\frac{ r_\text{T}(L)}{ \mu_\text{T}(L)}$. This is due to the fact that the point on the convex hull of $\Lambda_\text{I}$ and that corresponds to any vertix in $P_L$, for $1 \le L \le L_\text{I}$, is nothing but a vertix in $I_L$. For the fractions in this case, it is obvious that
\begin{align}
\frac{ r_\text{T}(L_\text{I})}{ \mu_\text{T}(L_\text{I})} < \frac{ r_\text{T}(L_\text{I}-1)}{ \mu_\text{T}(L_\text{I}-1)} < \ldots < \frac{ r_\text{T}(1)}{ \mu_\text{T}(1)}. \label{eq:inequality2}
\end{align}   
Therefore, using the inequalities in \eqref{eq:inequality1} and \eqref{eq:inequality2}, the minimum fraction we are looking for is given by $\frac{ r_\text{T}(L_\text{I})}{ \mu_\text{T}(L_\text{P})}$. This completes the proof.
\end{proof}
We draw the attention to the fact that the result in the above theorem holds true even in the special case where $L_\text{P}=L_\text{I}$. Specifically, under this case, we get a fraction of $\frac{ r(L_\text{P})}{ \mu(L_\text{P})}=\frac{ r(L_\text{I})}{ \mu(L_\text{I})}$. For some insights, Figure \ref{fig:ex2} sketches the general shapes of the stability regions for $L_\text{P}=L_\text{I}=N=2$.
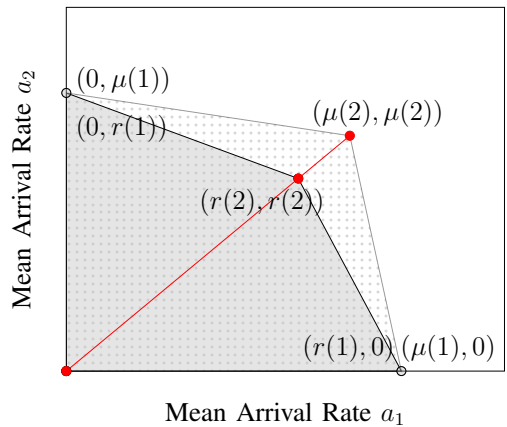
\begin{figure}[ht!]
  \centering
   \begin{tikzpicture}[scale=0.85]
    \begin{axis}
    		[ticks=none,
    	    ymin=0,ymax=1.7,
    		xmin=0,xmax=1.7,
    		xlabel={Mean Arrival Rate $a_1$},
    		xlabel near ticks,  
    		xlabel shift=10pt,  
    		ylabel={Mean Arrival Rate $a_2$},
    		ylabel near ticks,
    	    ylabel shift=10pt]  	
    		  
    	\addplot[pattern=dots,opacity=0.4,mark=*] coordinates {
       	            (0,0)
    		        (0,1.3)
    		        (1.1,1.1)
    		        (1.3,0)
    		        (0,0)	
    	};
    	\addplot[fill opacity=0.4,fill=lightgray,mark=*] coordinates {
    	       	    (0,0)
    	    		(0,1.3)
    	    		(0.9,0.9)
    	    		(1.3,0)
    	    		(0,0)	
    	};
    	\addplot[color=red,mark=*] coordinates {
    	    (0,0)
    	    (0.9,0.9)
    	  	(1.1,1.1)	
    	};		
 		
    	\node at (axis cs:0.88, 0) [anchor=south west] {\normalsize $(r(1),0)$};
    	\node at (axis cs:0, 1.03) [anchor=south west] {\normalsize $(0,r(1))$};
    	\node at (axis cs:0.48, 0.68) [anchor=south west] {\normalsize $(r(2),r(2))$};
    	\node at (axis cs:1.26, 0) [anchor=south west] {\normalsize $(\mu(1),0)$};
    	\node at (axis cs:0, 1.25) [anchor=south west] {\normalsize $(0,\mu(1))$};
    	\node at (axis cs:1.5, 1.1) [anchor=south east] {\normalsize $(\mu(2),\mu(2))$};
   	\end{axis}
   \end{tikzpicture}
   \captionsetup{font=small}
   \caption{An illustration of the stability regions of the perfect (dotted region) and imperfect (gray region) cases for the symmetric system. Here, $L_\text{P}=L_\text{I}=N=2$.}
   \label{fig:ex2}
\end{figure}

\subsection{Compare \ac{IA} to \ac{TDMA}-\ac{SVD} in terms of Stability}

In this subsection, we characterize the stability region of the case when we use \ac{SVD} technique with \ac{TDMA} instead of performing \ac{IA}. After that, we investigate which one between these two techniques outperforms the other in terms of stability. 

In the case where we apply \ac{TDMA} as a channel access method, there is only one active pair at a time, thus, at each timeslot, our system is reduced to a point-to-point \ac{MIMO} system. For this system, if we send a desired signal, denoted by $\mathbf{x}$, without any precoding scheme, the received signal is given by
\begin{align}
\mathbf{y}=\zeta \, \mathbf{H} \mathbf{x}+\mathbf{z}, \label{eq:y_svd}
\end{align} 
where $\mathbf{y}$ is the $N_{\text{r}} \times 1$ received signal vector, $\mathbf{x}$ is a $N_\text{t} \times 1$ complex vector, $\mathbf{H}$ denotes the $N_\text{r} \times N_\text{t}$ channel matrix with i.i.d. zero mean and unit variance complex Gaussian entries, and $\mathbf{z}$ is the additive white complex Gaussian noise vector with zero mean and covariance matrix $\sigma^2\mathbf{I}_{N_{\text{r}}}$. Recall that $\zeta$ is the path loss coefficient. Here, the only source of interference is the \ac{ISI} caused by the transmitter itself. To manage this problem, we use \ac{SVD} as a precoding technique. Specifically, by the singular value decomposition theorem we get 
\begin{align}
\mathbf{H}= \mathbf{U}\mathbf{D} \mathbf{V}^H, \label{eq:svdH}
\end{align} 
where $\mathbf{U}$ and $\mathbf{V}$ are $N_\text{r} \times N_\text{r}$ and $N_\text{t} \times N_\text{t}$ unitary matrices, respectively.  
$\mathbf{D}$ is a $N_\text{r} \times N_\text{t}$ diagonal matrix with non-negative square roots of the eigenvalues of matrix $\mathbf{H} \mathbf{H}^H$ in diagonal. These square roots are called the singular values of $\mathbf{H}$, and denoted by $\sqrt{\lambda_i}, i=1,2,\ldots,N_\text{r}$. We assume that we have $d$ ($\le \min(N_\text{t},N_\text{r})$) data streams to transmit. We also suppose that $\mathbf{H}$ is full rank, meaning that its rank is given by $\min(N_\text{t},N_\text{r})$; $d$ should be less than or equal to the rank of matrix $\mathbf{H}$. Let $\mathbf{y}^\prime=\mathbf{U}^H \mathbf{y}$, $\mathbf{x}^\prime=\mathbf{V}  \mathbf{x}$ and $\mathbf{z}^\prime =\mathbf{U}^H \mathbf{z}$. Using $\eqref{eq:y_svd}$ and $\eqref{eq:svdH}$ we obtain
\begin{align}
\mathbf{y}^\prime=\zeta \, \mathbf{D}\mathbf{x}+\mathbf{z}^\prime.
\end{align}
Note that $\mathbf{U}$ and $\mathbf{V}$ are unitary matrices, so $\mathbf{z}^\prime$ and $\mathbf{x}^\prime$ has the same distribution as $\mathbf{z}$ and $\mathbf{x}$, respectively.
Notice that if we send $\mathbf{x}^\prime$ instead of $\mathbf{x}$ and then, at the receiver, we multiply the corresponding received signal by $\mathbf{U}^H$, we can easily detect the transmitted signal. 
The equivalent \ac{MIMO} system can be seen as $d$ uncoupled parallel subchannels. It was proved that adaptive power allocation (i.e. water-filling algorithm) provides the highest capacity. However, to be able to come out with a fair comparison between \ac{IA} and \ac{TDMA}-\ac{SVD}, we should make the same assumption on power control, that is equal power allocation. Let us denote $k$ as the index of the active pair. The \ac{SNR} for stream $m$ can be written as
\begin{align}
\gamma_k^{(m)} = \frac{\zeta P}{N_\text{t} \, \sigma^2} \lambda_m, \qquad{} 1 \le m \le d. 
\end{align}
Let $\mathit{ma}=\max(N_\text{t},N_\text{r})$ and $\mathit{mi}=\min(N_\text{t},N_\text{r})$.
It was shown in \cite{Emre1999Capacityof} that the distribution of any one of the unordered eigenvalues is given by
\begin{align}
p(\lambda)=\frac{1}{\mathit{mi}} \sum_{n=0}^{\mathit{mi}-1} \frac{n!}{(n+\mathit{ma}-\mathit{mi})!} [L_{n}^{\mathit{ma}-\mathit{mi}}(\lambda)]^2 \lambda^{\mathit{ma}-\mathit{mi}} e^{-\lambda}, \quad \quad \lambda \ge 0,
\end{align} 
where $L_n^{\mathit{ma}-\mathit{mi}}(x)$ is the associated Laguerre polynomial of degree (order) $n$ and is given by
\begin{align}
L_n^{\mathit{ma}-\mathit{mi}}(\lambda)= \sum_{l=0}^n (-1)^l \frac{(n+\mathit{ma}-\mathit{mi})!}{(n-l)! (\mathit{ma}-\mathit{mi}+l)!  } \frac{\lambda^l}{ l!}. 
\end{align}  
Adopting the same rate model as for \ac{IA}, the average rate of the active user can be written as 
\begin{align}
(1-\theta) d R \, \mathbb{P}\left\{ \gamma_k^{(m)} \ge \tau \right\}.
\end{align}
Under this setting, it can be easily noticed that the average rate is independent of the identity of the active pair. This rate, which we denote by  $r_\text{svd}$, is derived in the following proposition.
\begin{proposition}
Under the \ac{TDMA}-\ac{SVD} technique, where one pair is active at a time, the average rate is given by 
\begin{align}
r_\text{svd} = (1-\theta) d R \sum_{n=0}^{\mathit{mi}-1}\Omega_n  \sum_{j=0}^{2n}\kappa_j \, \Gamma(j+\mathit{ma}-\mathit{mi}+1,\frac{N_\text{t} \, \sigma^2 \tau}{\zeta P}), \label{eq:r_svd}
\end{align}
where  $\Omega_n=\frac{ n!}{\mathit{mi}(n+\mathit{ma}-\mathit{mi})!}$, $\kappa_j=\sum_{i=0}^j \omega_i\omega_{j-i}$, $\omega_l=(-1)^l \frac{(n+\mathit{ma}-\mathit{mi})!}{(n-l)! (\mathit{ma}-\mathit{mi}+l)! } \frac{1}{ l!}$, with $\omega_l=0$ if $l>n$, and $\Gamma(\cdot,\cdot)$ is the upper incomplete Gamma function.
\end{proposition}
\begin{proof}
To begin with, let $\omega_l=(-1)^l \frac{(n+\mathit{ma}-\mathit{mi})!}{(n-l)! (\mathit{ma}-\mathit{mi}+l)!  } \frac{1}{ l!}$ and $\Omega_n=\frac{ n!}{\mathit{mi}(n+\mathit{ma}-\mathit{mi})!}$. Then, we have 
$L_n^{\mathit{ma}-\mathit{mi}}(\lambda)= \sum_{l=0}^n \omega_l \lambda_m^l$ and 
$p(\lambda_m)= \sum_{n=0}^{\mathit{mi}-1}\Omega_n [L_n^{\mathit{ma}-\mathit{mi}}(\lambda_m)]^2 \lambda_m^{\mathit{ma}-\mathit{mi}} e^{-\lambda_m}$. For the Laguerre polynomial we have
\begin{align}
[L_n^{\mathit{ma}-\mathit{mi}}(\lambda_m)]^2= \sum_{j=0}^{2n}\kappa_j \lambda_m^j,
\end{align}
where $\kappa_j=\sum_{i=0}^j \omega_i\omega_{j-i}$, with $\omega_s=0$ if $s>n$. On the other side, since $\gamma_k^{(m)} = \frac{\zeta P}{N_\text{t} \, \sigma^2} \lambda_m$, the corresponding success probability can be written as 
\begin{align}
\mathbb{P}\left\{ \gamma_k^{(m)} \ge \tau \right\} \nonumber &= \mathbb{P}\left\{ \lambda
_m \ge \frac{N_\text{t} \, \sigma^2 \tau}{\zeta P} \right\} \nonumber \\ &= \sum_{n=0}^{\mathit{mi}-1}\Omega_n  \sum_{j=0}^{2n}\kappa_j \int_{\frac{N_\text{t} \, \sigma^2 \tau}{\zeta P}}^{\infty} \lambda_m^{j+\mathit{ma}-\mathit{mi}} e^{-\lambda_m} d\lambda_m \nonumber\\ &= \sum_{n=0}^{\mathit{mi}-1}\Omega_n  \sum_{j=0}^{2n}\kappa_j \Gamma(j+\mathit{ma}-\mathit{mi}+1,\frac{N_\text{t} \, \sigma^2 \tau}{\zeta P}),
\end{align}
where $\Gamma(\cdot,\cdot)$ stands for the upper incomplete Gamma function. Hence, the desired result follows.
\end{proof}

Let us now focus on the stability region of \ac{TDMA}-\ac{SVD}. Under this technique, as mentioned before, only one pair is active at a timeslot, thus the subset of decision vectors is given by $S_1$.  
Using the above, the stability region for \ac{TDMA}-\ac{SVD} can be described as follows. 
\begin{proposition}
	\label{stabregSVD}
If we apply \ac{TDMA}-\ac{SVD} technique, the stability region of the corresponding system can be given by 
\begin{align}
\Lambda_{\text{svd}}=\mathcal{CH} \left\{ J_1 \right\}, 
\end{align}
where $J_1=\left\{ r_{\text{svd}} \, \mathbf{s}, \forall  \mathbf{s} \in S_1  \right\}$.
\end{proposition}
\begin{proof}
The proof of this proposition is very similar to the proof of Theorem \ref{th:stabregionSI} and is thus omitted to avoid repetition. 
\end{proof}

To schedule one of the pairs at timeslot $t$, the \ac{CS}, which we assume has a full knowledge of the queue lengths and the average transmission rate of \ac{SVD}, applies the Max-Weight rule for $\mathbf{s} \in  S_1$ (with $S_1$ is the subset of different combinations of choosing one pair), such as 
\begin{align}
\mathbf{s}(t)= \operatorname*{arg\,max}_{ \mathbf{s} \in   S_1  }  \left\{ r_\text{svd} \, \mathbf{s} \cdot \mathbf{q}(t)  \right\}.
\end{align} 
But, since $r_\text{svd}$ is independent of the identity of the active pair, this policy will always schedule the pair with the largest queue length. After finding this pair, the \ac{CS} broadcasts this information so that the corresponding user activates itself and then sends the training sequence in the uplink phase, letting its intended transmitter estimate the channel. 

Now, our aim is to compare the performances of \ac{IA} and \ac{TDMA}-\ac{SVD} techniques in terms of stability. Before proceeding with the analysis, we point out that if the stability region of \ac{IA} surpasses that of \ac{TDMA}-\ac{SVD}, this will take place only on a part of the second region. In other words, the first stability region cannot completely cover the second stability region. This observation comes from the following facts: (i) as we proved earlier, in the imperfect (resp., perfect) case the points in $I_1$ (resp., $P_1$) are vertices of the stability region of \ac{IA} given by $\Lambda_\text{I}$ (resp., $\Lambda_\text{P}$), (ii) each one of these vertices lies on a different axis in $\mathbb{R}^N_+$, with a coordinate value $r(1)$ (resp., $\mu(1)$), and (iii) the points in $J_1$, which are (the only) vertices of $\Lambda_\text{svd}$ and lie on the different axes of $\mathbb{R}^N_+$, have the same coordinate value, that is $r_\text{svd}$, which satisfies $r_\text{svd} > r(1)$ (resp., $r_\text{svd} > \mu(1)$). Notice that in our system we have $r(1)=\mu(1)$. On the other side, it is straightforward to see that for the converse case, that is when the stability region of \ac{TDMA}-\ac{SVD} surpasses that of \ac{IA}, we have a full coverage. To get some insights, the stability region of the imperfect case of \ac{IA} and that of \ac{TDMA}-\ac{SVD} for the two above-mentioned scenarios are depicted in Figure \ref{subfig:a} and Figure \ref{subfig:b}, under a simple example in which   $L_\text{I}=N=2$. Please, refer to Appendix \ref{app:examples} for more examples and illustrations.
Now, for the analysis about the comparison between the two techniques in terms of stability, we adopt the following approach. We investigate if there exists a number $L$ such that the stability region for  \ac{IA} surpasses the stability region of \ac{TDMA}-\ac{SVD}. This leads us to the following theorem.
\begin{figure}[ht!]

\begin{subfigure}[b]{0.5\textwidth}
   \begin{tikzpicture}[scale=0.85]
    \begin{axis}
    		[ticks=none,
    	    ymin=0,ymax=1.7,
    		xmin=0,xmax=1.7,
    		xlabel={Mean Arrival Rate $a_1$},
    		xlabel near ticks,  
    		xlabel shift=10pt,  
    		ylabel={Mean Arrival Rate $a_2$},
    		ylabel near ticks,
    	    ylabel shift=10pt]  	
    		  
    	\addplot[pattern=dots,opacity=0.4,mark=*] coordinates {
       	            (0,0)
    		        (0,1.5)
    		        (1.5,0)
    		        (0,0)	
    	};
    	\addplot[fill opacity=0.4,fill=lightgray,mark=*] coordinates {
    	       	    (0,0)
    	    		(0,1.3)
    	    		(0.9,0.9)
    	    		(1.3,0)
    	    		(0,0)	
    	};
    	\node at (axis cs:0.95, 0) [anchor=south west] {\normalsize $(r(1),0)$};
    	\node at (axis cs:0, 1.14) [anchor=south west] {\normalsize $(0,r(1))$};
    	\node at (axis cs:0.7, 0.94) [anchor=south west] {\normalsize $(r(2),r(2))$};
    	\node at (axis cs:1.34, 0) [anchor=south west] {\normalsize $(r_\text{svd},0)$};
    	\node at (axis cs:0, 1.4) [anchor=south west] {\normalsize $(0,r_\text{svd})$};
   	\end{axis}
   \end{tikzpicture}
 \caption{}
 \label{subfig:a}
\end{subfigure}
\begin{subfigure}[b]{0.5\textwidth}
   \begin{tikzpicture}[scale=0.85]
    \begin{axis}
    		[ticks=none,
    	    ymin=0,ymax=2.1,
    		xmin=0,xmax=2.1,
    		xlabel={Mean Arrival Rate $a_1$},
    		xlabel near ticks,  
    		xlabel shift=10pt,  
    		ylabel={Mean Arrival Rate $a_2$},
    		ylabel near ticks,
    	    ylabel shift=10pt]  	
    		  
    	\addplot[pattern=dots,opacity=0.4,mark=*] coordinates {
       	            (0,0)
    		        (0,2)
    		        (2,0)
    		        (0,0)	
    	};
    	\addplot[fill opacity=0.4,fill=lightgray,mark=*] coordinates {
    	       	    (0,0)
    	    		(0,1.3)
    	    		(0.9,0.9)
    	    		(1.3,0)
    	    		(0,0)	
    	};
    	\node at (axis cs:0.88, 0) [anchor=south west] {\normalsize $(r(1),0)$};
    	\node at (axis cs:0, 1.10) [anchor=south west] {\normalsize $(0,r(1))$};
    	\node at (axis cs:0.7, 0.94) [anchor=south west] {\normalsize $(r(2),r(2))$};
    	\node at (axis cs:1.65, 0) [anchor=south west] {\normalsize $(r_\text{svd},0)$};
    	\node at (axis cs:0, 1.85) [anchor=south west] {\normalsize  $(0,r_\text{svd})$};
   	\end{axis}
   \end{tikzpicture}
\caption{}  
\label{subfig:b}
\end{subfigure}
 \captionsetup{font=small,skip=-0.352cm}  
 \caption{Stability regions of \ac{TDMA}-\ac{SVD} (dotted region) and \ac{IA} under the imperfect case (gray region) for the symmetric system, where $L_\text{I}=N=2$. (a) \ac{IA} outperforms \ac{TDMA}-\ac{SVD} and (b)  \ac{TDMA}-\ac{SVD} outperforms \ac{IA}.}
 \label{fig:ex4}
\end{figure}
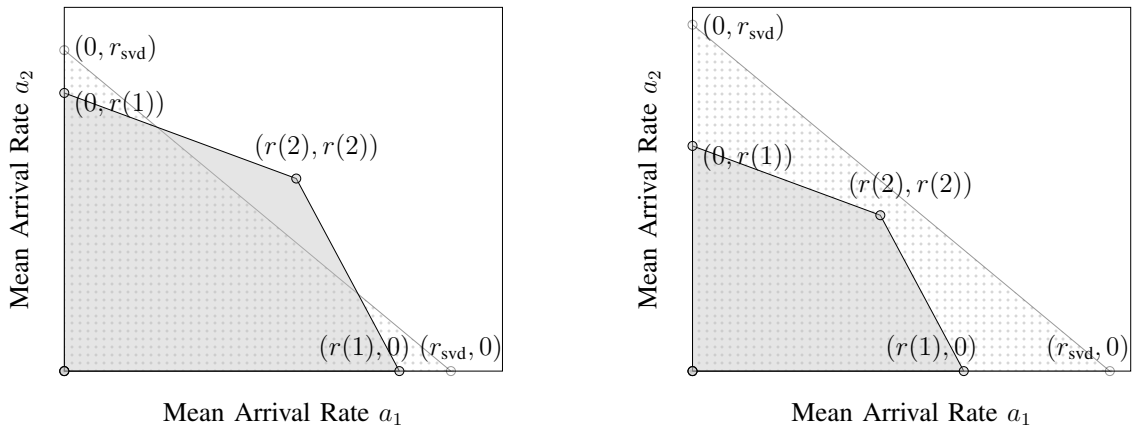   
\begin{theorem}
	\label{th:IAvsSVD}
	For the symmetric system with limited backhaul (i.e. imperfect case), \ac{IA} can outperform \ac{TDMA}-\ac{SVD} in terms of stability if there exists a number $L$ such that $Lr(L) >r_\text{svd}$, with $1\le L \le L_\text{I}$. If this condition is not satisfied, then  \ac{TDMA}-\ac{SVD} technique gives better performances than \ac{IA}. For the same system but with unlimited backhaul (i.e. perfect case), we get a similar result, that is
	\ac{IA} can yield better stability gain than \ac{TDMA}-\ac{SVD} if there is a number $L$ such that $L\mu(L) >r_\text{svd}$, with $1\le L \le L_\text{P}$; otherwise \ac{TDMA}-\ac{SVD} provides better stability performances.
\end{theorem}
\begin{proof}
The proof of this theorem is in the same spirit as the proof of Theorem \ref{th:imvsp_frac}, thus only the outline is given to avoid repetition. From Lemma \ref{le:Sl}, we can express a point in $S_{L}$ as $\frac{L}{1}$ times a point on the convex hull of $S_1$, where these two points are on the same line from the origin. Thus, the fraction between a point in $I_{L}$ and 
its corresponding point on the convex hull of $J_1$ can be given by $\frac{Lr(L)}{r_\text{svd}}$. So, the point in $I_{L}$ surpasses its corresponding point on the convex hull of $J_1$ if $\frac{Lr(L)}{r_\text{svd}} > 1$, or equivalently if $Lr(L) > r_\text{svd}$, with $1 \le L \le L_\text{I}$. Note that it suffices to test this condition for all $ L \le L_\text{I}$ since the points in $\bar{\mathcal{I}}$ are inside the convex hull of the points in $\mathcal{I}$ (see Lemma \ref{le:Rl}).
The same approach can be adopted for the perfect case, and we obtain  $L \mu(L) > r_\text{svd}$ with $L \le L_\text{P}$ as a sufficient condition to have $\Lambda_{\text{svd}}$ (partially) surpassed by $\Lambda_{\text{P}}$. 
This completes the proof.
\end{proof}

This theorem allows us to decide if the system should be deployed with \ac{TDMA}-\ac{SVD} or \ac{IA} as an interference management technique. For the imperfect (resp., perfect) case, this decision is made based on the existence (or not) of a number of pairs $L$ such that $Lr(L) >r_{\text{svd}}$ (resp., $L\mu(L) >r_{\text{svd}}$), with $1\le L \le L_\text{I}$ (resp., $1\le L \le L_\text{P}$). Specifically, if this condition is satisfied, it may be beneficial to use the \ac{IA} technique since we have a part of its stability region that surpasses the stability region of \ac{TDMA}-\ac{SVD} (given by $\Lambda_{\text{svd}}$). On the other hand, if this condition is not satisfied, then the stability region of \ac{IA} is entirely inside $\Lambda_{\text{svd}}$, and thus it is better to use \ac{TDMA}-\ac{SVD} technique.

For both the imperfect and perfect cases of \ac{IA}, if the corresponding condition (defined above) is satisfied, meaning that the stability region of \ac{IA} surpasses (partially) the stability region of \ac{TDMA}-\ac{SVD}, we can achieve a bigger stability region by deciding to switch between these two techniques instead of deciding to always use one of them; as seen before, here the decision was to always apply \ac{IA}. 
At each timeslot, we choose the interference management technique that yields the greater Max-Weight result. Specifically, for the imperfect case we compute 
\begin{align}
\max \left\{   \operatorname*{max}_{ \mathbf{s} \in   \mathcal{S}  }  \left\{ r(  \left\| \mathbf{s} \right\|_1) \, \mathbf{s} \cdot \mathbf{q}(t)  \right\},  \operatorname*{max}_{ \mathbf{s} \in   S_1  }  \left\{   r_\text{svd} \, \mathbf{s} \cdot \mathbf{q}(t)  \right\}   \right\},
\end{align}
and, similarly, for the perfect case we find 
\begin{align}
\max \left\{   \operatorname*{max}_{ \mathbf{s} \in   \mathcal{S}  }  \left\{ \mu(  \left\| \mathbf{s} \right\|_1) \, \mathbf{s} \cdot \mathbf{q}(t)  \right\},  \operatorname*{max}_{ \mathbf{s} \in   S_1  }  \left\{   r_\text{svd} \, \mathbf{s} \cdot \mathbf{q}(t)  \right\}   \right\}.
\end{align}
In the following theorem we provide a precise characterization of the resulting stability region for both cases.
\begin{theorem}
\label{th:StabRegIAandSVD}
Using an approach that consists in switching between \ac{IA} and \ac{TDMA}-\ac{SVD} by selecting at each timeslot the technique that yields the highest Max-Weight result, the resulting stability region under the imperfect case of \ac{IA} can be characterized as 
\begin{align}
\mathcal{CH} \left\{ I_1,I_2,...,I_{L_\text{I}}, J_1 \right\}, 
\end{align}
whereas for the perfect case we get 
\begin{align}
\mathcal{CH} \left\{ P_1,P_2,...,P_{L_{\text{P}}}, J_1 \right\}. 
\end{align}
\end{theorem}
\begin{proof}
First we will prove that the region in the statement of the Theorem is achievable by the
proposed policy. Indeed,
the switching process can be seen as selecting \ac{IA} and \ac{TDMA}-\ac{SVD} with probabilities $\pi_\text{ia}$ and $\pi_\text{svd}$, respectively, where $\pi_\text{ia} + \pi_\text{svd} = 1$. Hence, the resulting stability region can be given by $\pi_\text{ia} \Lambda_\text{ia}+\pi_\text{svd} \Lambda_\text{svd}$, where $\Lambda_\text{ia}$ represents $\Lambda_\text{I}$ (resp., $\Lambda_\text{P}$). It means that the resulting region is nothing but the convex hull of the stability regions of \ac{IA} and \ac{TDMA}-\ac{SVD}. More in details, knowing that the stability region of \ac{IA} under, for example, the imperfect case is $\mathcal{CH} \left\{ I_1,I_2,...,I_{L_\text{I}} \right\}$ and that of \ac{TDMA}-\ac{SVD} is $\mathcal{CH} \left\{J_1 \right\}$, the resulting stability region is given by the following
\begin{align}
\mathcal{CH} \left\{ \mathcal{CH} \left\{ I_1,I_2,...,I_{L_\text{I}} \right\}, \mathcal{CH} \left\{J_1 \right\}   \right\}= \mathcal{CH} \left\{ I_1,I_2,...,I_{L_\text{I}}, J_1 \right\}.
\end{align}

We then need to prove the converse, that is, if a centralized policy achieves stability, then the mean arrival rate lies in (the interior of) the region given by the theorem.
The proof of this part can be done in the same way as the proof of Theorem \ref{th:stabregionSI} and is thus omitted to avoid repetition.

The same analysis holds for the perfect case of \ac{IA}. This completes the proof.
\end{proof}
One last thing to mention is that here the analysis is independent of the knowledge of the arrival rate vector $\mathbf{a}$, which is unknown in general. Next, we assume that we know this rate vector based on which the interference management technique will be selected.

\subsection{Select \ac{IA} or \ac{TDMA}-\ac{SVD} based on the Arrival Rate Vector}

In this subsection, we want to select the interference management technique based on the arrival rate vector, which we suppose is known here; we recall that this vector is denoted by $\mathbf{a}$. We next provide the analysis for this selection process under the imperfect case of \ac{IA}, while noting that a similar analysis can be used under the perfect case. 

The stability regions of \ac{IA} (with the imperfect case) and \ac{TDMA}-\ac{SVD} were already characterized and denoted, respectively, by $\Lambda_\text{I}$ and $\Lambda_\text{svd}$. 
For sake of guaranteeing system stability, we assume that $\mathbf{a}$ is in the union of these two regions. As explained earlier, we recall that when we say $\Lambda_\text{svd}$ surpasses $\Lambda_\text{I}$, it implies that the stability region of \ac{TDMA}-\ac{SVD} completely covers that of \ac{IA}. On the other hand, for the converse case, the stability region of \ac{IA} partially exceeds that of \ac{TDMA}-\ac{SVD}. Two cases are to consider: $\Lambda_\text{svd}$ covers $\Lambda_\text{I}$, and  $\Lambda_\text{I}$ (partially) surpasses $\Lambda_\text{svd}$. In the first case we propose using \ac{TDMA}-\ac{SVD} since under this technique the stability performances are better than those under \ac{IA}, whereas in the second case we adopt the following reasoning based on the position of $\mathbf{a}$ compared to $\Lambda_\text{svd}$ and $\Lambda_\text{I}$ : (i) $\mathbf{a}$ is inside  $\Lambda_\text{I}$ but outside $\Lambda_\text{svd}$, it is straightforward to perform \ac{IA} technique, (ii) $\mathbf{a}$ is inside  $\Lambda_\text{svd}$ but outside $\Lambda_\text{I}$, it is clear that we should use \ac{TDMA}-\ac{SVD}, and (iii)  $\mathbf{a}$ is inside  $\Lambda_\text{I}$ and $\Lambda_\text{svd}$, we suggest applying \ac{TDMA}-\ac{SVD} because in addition to the fact that it can guarantee the system stability (as \ac{IA}), as mentioned previously, this technique does not require any backhaul usage, which is not the case for \ac{IA}. 

The above algorithm (reasoning) requires testing if point $\mathbf{a}$ is in the stability region of \ac{IA} or \ac{TDMA}-\ac{SVD}.
It is obvious that the boundary of this latter region lies on a hyperplane constructed using a set of points, each of which is on a different axis but having the same coordinate value, namely $r_\text{svd}$. Hence, the equation of this hyperplane can be written as $r_\text{svd}^{-1} \sum_{k=1}^N \nu_k =1$, or equivalently $\sum_{k=1}^N \nu_k =r_\text{svd}$, where $\nu_k$ represents the $k$-coordinate. Thus, point $\mathbf{a}$ is in $\Lambda_\text{svd}$ if
\begin{align}
\sum_{k=1}^N a_k < r_\text{svd}.
\end{align}
On the other side, in order to make this test for $\Lambda_\text{I}$, we formulate an optimization problem. In detail, we know that any point in $\Lambda_\text{I}$ can be written as the convex combination of the vertices of this convex hull; the set of these vertices is given by $\mathcal{I}$ (see Lemma \ref{le:Rl}). We let these vertices form the columns of a $N \times \left| \mathcal{I} \right|$ matrix denoted by $\mathbf{A}$, where $\left| \mathcal{I} \right|$ is the cardinality of set $\mathcal{I}$. To test if vector $\mathbf{a}$ is in the region $\Lambda_\text{I}$, we simply try to find if there exists a convex combination of the columns of $\mathbf{A}$ that can produce $\mathbf{a}$, where the coefficients of this combination are non-negative and sum to $1$.   
This is equivalent to solve the following problem
\begin{align}
\underset{\boldsymbol{\delta}}{\text{minimize}}
&\qquad  || \mathbf{A} \boldsymbol{\delta} - \mathbf{a} ||_2  \label{eq:optprob1}  \\
\text{subject to}
&\qquad \,  \mathbf{1}^T \boldsymbol{\delta} =1 \label{eq:consprob1} \\
&\qquad \,  \,  \boldsymbol{\delta} \ge 0 \label{eq:consprob2}
\end{align}  
where $\boldsymbol{\delta}$ denotes the vector of coefficients of the convex combination and $\mathbf{1}$ is the all-ones vector.   
As stated before, here we are trying to find if there is a convex combination of the columns of $\mathbf{A}$ that yields $\mathbf{a}$. Any solution $\boldsymbol{\delta}^\text{*}$ to this problem that gives the objective function $\left\| \mathbf{A} \boldsymbol{\delta}^\text{*} - \mathbf{a} \right\|_2=0$ (or equivalently, $\mathbf{A} \boldsymbol{\delta}^\text{*} - \mathbf{a}=\mathbf{0}$) is considered as feasible. This feasible solution ensures that point $\mathbf{a}$ is in  $\Lambda_\text{I}$. Note that we can define an equivalent problem to the one defined before by putting condition $\eqref{eq:consprob1}$ in the subject function. Specifically, let $\mathbf{A}_1$ denote the matrix formed by adding a row vector of ones at the end of $\mathbf{A}$, and let $\mathbf{a}_1$ the vector constructed by adding coordinate one at the end of $\mathbf{a}$. Thus, we can define the following equivalent problem
\begin{align}
\underset{\boldsymbol{\delta}}{\text{minimize}}
&\qquad  || \mathbf{A}_1 \boldsymbol{\delta} - \mathbf{a}_1 ||_2  \label{eq:optprob21}  \\
\text{subject to}
&\qquad \, \, \boldsymbol{\delta} \ge 0 \label{eq:consprob21} 
\end{align}  
We can easily see that the existence of a feasible solution, which gives $\mathbf{A_1} \boldsymbol{\delta}^* - \mathbf{a}_1 =\mathbf{0}$, ensures the satisfaction of condition $\eqref{eq:consprob1}$. This is due to the fact that the last coordinate value of $\mathbf{A_1} \boldsymbol{\delta}$, given by $\sum_{j}^{\left|   \mathcal{I}\right|} \delta_j$, is equal to the last coordinate value of $\mathbf{a}_1$ (equals to $1$). Notice that the equivalent problem defined above is nothing but the non-negative least squares problem. In general, the original problem and its equivalent one have no analytic solutions, however there exist several (low-complexity) algorithms that can be used to solve these problems numerically \cite{Boyd2004}.

A very similar analysis can be adopted for the perfect case of \ac{IA}, in which we replace $\Lambda_\text{I}$ by $\Lambda_\text{P}$, and then the columns of $\mathbf{A}$ represent the vertices of this latter region. In order to choose the interference management technique, similar reasoning and formulations to those used for the imperfect case can be considered here.

\subsection{Impact of $B$ and  $N$ on the System Stability Region}

Here the analysis is restricted for the imperfect case of \ac{IA}, where the backhaul is of finite capacity. We recall that under the adopted system the number of bits, $B$, and the maximum number of pairs, $N$, are considered as unchanged. However, since the stability analysis depends essentially on these two parameters, it is important to investigate the impact of changing these parameters on the system stability region. But before conducting such a study, we note that an increasing from $B^\prime$ to $B$ can be seen as a decreasing from $B$ to $B^\prime$; the same remark can be made for $N^\prime$ and $N$. That is to say, it suffices to study one of these two ways of changing the parameters under investigation. Here, we choose to reduce these parameters, meaning that we study the impact of reducing $B$ to $B^\prime$ and $N$ to $N^\prime$. 
We next investigate the impact of each parameter reduction on the stability region of the system. 

\subsubsection{Reduce the Number of Bits $B$}


To begin with, let $\Delta_{B}^{\text{*}}$ and $\Delta_{B^\prime}^{\text{*}}$ denote the same algorithm, that is the Max-Weight policy, for the same maximum number of pairs $N$, but the first one considers the case where the number of bits is equal to $B$ and for the second one this number is $B^\prime$. Further, let $\mathcal{L}_B$ and $\mathcal{L}_{B'}$ denote the subsets of pairs selected by $\Delta_{B}^{\text{*}}$ and $\Delta_{B'}^{\text{*}}$, respectively. Also, we denote by $\Lambda_B$ and $\Lambda_{B^\prime}$ the stability regions achieved by $\Delta_{B}^\text{*}$ and $\Delta_{B^\prime}^\text{*}$, respectively. In addition, we define $r(L,B)$ as the average rate $r(L)$ with a number of bits $B$. Equivalently, $r(L,B^\prime)$ is the average rate function $r(L)$ in which we replace $B$ by $B^\prime$.
For this model, we can state the following theorem.
\begin{theorem}
\label{th:BtoB'}
 For the same system in which the maximum number of pairs is $N$, if we decrease the number of bits from $B$ to $B^\prime$, the stability region in the second case (with $B'$), given by $\Lambda_{B^\prime}$, can be bounded as 
 \begin{align}
 \frac{r(N,B^\prime)}{r(N,B)} \Lambda_B  \subseteq \Lambda_{B^\prime} \subseteq \Lambda_B. \label{eq:BB'_frac}
 \end{align}
\end{theorem}
\begin{proof}
The proof consists in three steps. We first show that $(\mathbf{r} \cdot \mathbf{q})^{(\Delta_{B^\prime}^{\text{*}})} \ge \frac{r(L_{B},B^\prime)}{r(L_{B},B)}  (\mathbf{r} \cdot \mathbf{q})^{(\Delta_{B}^{\text{*}})} $. We then minimize the fraction $\frac{r(L_{B},B^\prime)}{r(L_{B},B)}$ under the condition that the number of active pairs, $L_B$, can be less than or equal to $N$; we get $\frac{r(N,B^\prime)}{r(N,B)}$ as a minimum fraction. 
Finally, we show that the stability region $\Lambda_{B^\prime}$ achieves at least a fraction $\frac{r(N,B^\prime)}{r(N,B)}$ of the stability region $\Lambda_B$ and we conclude that $\Lambda_{B^\prime}$ can be bounded as given in \eqref{eq:BB'_frac}.\\
\textbf{Step 1:} 
Recall that under the symmetric case all the active pairs have the same average rate, which we denote here by $r(L_B,B)$. 
Thus, we can write 
$(\mathbf{r} \cdot \mathbf{q})^{(\Delta_{B}^{\text{*}})}= r(L_B,B) \sum_{k \in \mathcal{L}_{B} } q_k$,
where $L_B=\left|  \mathcal{L}_{B} \right|.$
Similarly, we get $(\mathbf{r} \cdot \mathbf{q})^{(\Delta_{B^\prime}^{\text{*}})}= r(L_{B^\prime},B^\prime) \sum_{k \in \mathcal{L}_{B^\prime} } q_k $, with $L_{B^\prime}=\left| \mathcal{L}_{B^\prime} \right|$. Since $\Delta_{B^\prime}^{\text{*}}$ maximizes the product $\mathbf{r} \cdot \mathbf{q}$ for the case where the number of bits is $B^\prime$, it follows that 
\begin{align}
r(L_{B},B^\prime) \sum_{k \in \mathcal{L}_{B} } q_k \le r(L_{B^\prime},B^\prime) \sum_{k \in \mathcal{L}_{B^\prime} } q_k.  \label{eq:r.q}
\end{align}
Also, using the definition of $\Delta_{B}^{\text{*}}$, that is maximizing $\mathbf{r} \cdot \mathbf{q}$ for the case where $B$ is the number of bits, we have
\begin{align}
r(L_{B},B^\prime) \sum_{k \in \mathcal{L}_{B}} q_k \le r(L_{B},B) \sum_{k \in \mathcal{L}_{B}} q_k.
\end{align} 
To get $r(L_{B},B) \sum_{k \in \mathcal{L}_{B}} q_k \le  \beta^{-1} r(L_{B},B^\prime) \sum_{k \in \mathcal{L}_{B}} q_k $, for some $\beta \le 1$, it suffices to take $\beta^{-1} \ge \frac{r(L_{B},B)}{r(L_{B},B^\prime)}$, or equivalently $\beta \le \frac{r(L_{B},B^\prime)}{r(L_{B},B)}$. We consider the equality in the latter relation, i.e. $\beta = \frac{r(L_{B},B^\prime)}{r(L_{B},B)}$. Combining this result with the inequality in $\eqref{eq:r.q}$ yields 
\begin{align}
r(L_{B^\prime},B^\prime) \sum_{k \in \mathcal{L}_{B^\prime} } q_k \ge  \frac{r(L_{B},B^\prime)}{r(L_{B},B)} r(L_{B},B) \sum_{k \in \mathcal{L}_{B}} q_k. \label{eq:BB'frac}
\end{align} 
\textbf{Step 2:}
We now want to find the minimum fraction $\frac{r(L_{B},B^\prime)}{r(L_{B},B)}$ w.r.t. $L_B$, such as
 \begin{align}
 \underset{L_B}{\text{minimize}}
 &\qquad  \frac{r(L_{B},B^\prime)}{r(L_{B},B)}  \label{eq:frac3BB'} \\
 \text{subject to}
 &\qquad \, L_B \le N
 \end{align}  
To solve this problem, we show that the objective function to minimize in  $\eqref{eq:frac3BB'}$ is a decreasing function w.r.t. $L_B$. Indeed, using $\eqref{eq:rsimple}$, we have
 \begin{align}
 \frac{r(L_{B},B^\prime)}{r(L_{B},B)}= \frac{(1-L_B\theta)d R e^{-\frac{ \sigma^2 \tau }{\alpha}}   \left( F(B^\prime) \right)^{L_B-1} } {(1-L_B\theta)d R e^{-\frac{ \sigma^2 \tau }{\alpha}}  \left(  F(B) \right)^{L_B-1} }= \left(\frac{F \left(B^\prime \right)}{F(B)} \right)^{L_B-1}, 
 \end{align}
in which function $F$ was already defined for equation $\eqref{eq:rsimple}$. It is clear that $F\left(B^\prime \right) < F\left(B\right)$ because $B^\prime < B$, which implies that $\left(\frac{F \left(B^\prime \right)}{F(B)} \right)^{L_B-1}$ decreases with $L_B$. Since $L_B\le N$, the optimization problem reaches its minimum at $L_B=N$. Therefore, the minimum fraction we are looking for can be given by $\frac{r(N,B^\prime)}{r(N,B)}$. \\
\textbf{Step 3:} 
Using the minimum fraction derived before, we now want to examine the stability region achieved by $\Delta_{B^\prime}^{\text{*}}$.
To this end, we define the quadratic \emph{Lyapunov function} as
\begin{align}
\mathit{Ly}(\mathbf{q}(t)) \triangleq \frac{1}{2} \left( \mathbf{q}(t) \cdot \mathbf{q}(t) \right)= \frac{1}{2} \sum_{k=1}^N q_k(t)^2.
\end{align}
From the evolution equation for the queue lengths (see $\eqref{eq:q_evo}$) we have
\begin{align}
\mathit{Ly}(\mathbf{q}(t+1))-\mathit{Ly}(\mathbf{q}(t)) \nonumber &=\frac{1}{2} \sum_{k=1}^{N} \left[  q_k(t+1)^2-q_k(t)^2 \right]  \\  \nonumber &= \frac{1}{2} \sum_{k=1}^{N} \left[  \max \left\{q_k(t)-B_k(t),0 \right\}^2 +A_k(t)^2-q_k(t)^2  \right] \\ & \le  \sum_{k=1}^{N}  \frac{ \left[   A_k(t)^2+B_k(t)^2  \right] }{2} + \sum_{k=1}^{N} q_k(t)  \left[ A_k(t)- B_k(t) \right],  \label{eq:Lya}
\end{align}
where in the final inequality we have used the fact that for any $q \ge 0$, $A \ge 0$, $B \ge 0$, we have
\begin{align}
\nonumber ( \max \left\{q-B,0 \right\}+A)^2 \le q^2 + A^2 +B^2 +2 q(A-B).
\end{align}
Now define $\mathit{Dr}(\mathbf{q}(t))$ as the \emph{conditional Lyapunov drift} for timeslot $t$
\begin{align}
\mathit{Dr}(\mathbf{q}(t)) \triangleq \mathbb{E} \left\{ \mathit{Ly}(\mathbf{q}(t+1))-\mathit{Ly}(\mathbf{q}(t)) \mid  \mathbf{q}(t)  \right\}. 
\end{align}
From $\eqref{eq:Lya}$, we have that $Dr(\mathbf{q}(t)) $ for a general scheduling policy satisfies
\begin{align}
Dr(\mathbf{q}(t)) \le  \mathbb{E} \left\{  \sum_{k=1}^{N}  \frac{  A_k(t)^2+B_k(t)^2   }{2}  \mid \mathbf{q}(t) \right\} + \sum_{k=1}^{N} q_k(t) a_k -  \mathbb{E} \left\{  \sum_{k=1}^{N}  q_k(t) B_k(t) \mid \mathbf{q}(t) \right\}, \label{eq:Dri}
\end{align}
where we have used the fact that arrivals are i.i.d. over slots and hence independent of current queue backlogs, so that $\mathbb{E} \left\{ A_k(t) \mid \mathbf{q}(t) \right\}= \mathbb{E} \left\{ A_k(t)  \right\} = a_k$. Now define $E$ as a finite positive constant that bounds the first term on the right-hand-side of the above drift inequality, so that for all $t$, all possible $q_k(t)$,
and all possible control decisions that can be taken, we have
\begin{align}
\mathbb{E} \left\{  \sum_{k=1}^{N}  \frac{  A_k(t)^2+B_k(t)^2   }{2}  \mid \mathbf{q}(t) \right\} \le E.
\end{align}
Note that $E$ exists since $A_k(t)< A_{\text{max}}$ and $B_k(t)< B_{\text{max}}$.
Using the expression in $\eqref{eq:Dri}$ yields
\begin{align}
Dr(\mathbf{q}(t)) & \le E + \sum_{k=1}^{N} q_k(t) a_k - \mathbb{E} \left\{  \sum_{k=1}^{N} q_k(t)   B_k(t) \mid \mathbf{q}(t) \right\}. \label{eq:Dri1}
\end{align}
The conditional expectation at the right-hand-side of the above inequality is with respect to the randomly observed channel states and the (possibly random) scheduling policy.
%
Thus, the drift under $\Delta_{B^\prime}^{\text{*}}$ can be expressed as
%
\begin{align}
Dr^{(\Delta_{B^\prime}^{\text{*}})}(\mathbf{q}(t)) \le   E - \sum_{k=1}^{N} q_k(t)  \left[ \mathbb{E} \left\{  B^{(\Delta_{B^\prime}^{\text{*}})}_k(t) \mid \mathbf{q}(t) \right\} - a_k \right], \label{eq:DrB'}
\end{align}
Note that here we have $\mathbb{E} \left\{ B^{(\Delta_{B^\prime}^{\text{*}})}_k(t) \mid \mathbf{q}(t), \mathcal{L}_{B^\prime} \right\}= r(L_{B^\prime},B^\prime)$, thus  $\mathbb{E} \left\{  B^{(\Delta_{B^\prime}^{\text{*}})}_k(t) \mid \mathbf{q}(t) \right\}= \mathbb{E} \left\{  r(L_{B^\prime},B^\prime) \mid \mathbf{q}(t) \right\}$, where the expectation at the left-hand-side of this latter equality is over the randomly observed channel state and the randomness of policy $\Delta_{B^\prime}^{\text{*}}$, whereas the expectation at the right-hand-side of this equality is (only) over the randomness of $\Delta_{B^\prime}^{\text{*}}$. Similarly, we have  $\mathbb{E} \left\{  B^{(\Delta_{B}^{\text{*}})}_k(t) \mid \mathbf{q}(t) \right\}= \mathbb{E} \left\{  r(L_{B},B) \mid \mathbf{q}(t) \right\}$. Hence, using $\eqref{eq:BB'frac}$ and the fact that the minimum fraction is $\beta=\frac{r(N,B^\prime)}{r(N,B)}$, we can claim that
\begin{align}
\sum_{k=1}^N q_k(t) \mathbb{E} \left\{  B^{(\Delta_{B^\prime}^{\text{*}})}_k(t) \mid \mathbf{q}(t) \right\}  \ge \sum_{k=1}^N q_k(t) \beta \, \mathbb{E} \left\{  B^{(\Delta_{B}^{\text{*}})}_k(t) \mid \mathbf{q}(t)  \right\}.
\end{align}
Plugging this directly into $\eqref{eq:DrB'}$ yields
\begin{align}
Dr^{(\Delta_{B^\prime}^{\text{*}})}(\mathbf{q}(t)) \le   E - \sum_{k=1}^N q_k(t) \left[ \beta \, \mathbb{E} \left\{  B^{(\Delta_{B}^{\text{*}})}_k(t) \mid \mathbf{q}(t) \right\} - a_k  \right]. \label{eq:DrB'1}
\end{align}
The above expression can be re-expressed as 
\begin{align}
Dr^{(\Delta_{B^\prime}^{\text{*}})}(\mathbf{q}(t)) \le   E - \beta \sum_{k=1}^N q_k(t) \left[ \mathbb{E} \left\{  B^{(\Delta_{B}^{\text{*}})}_k(t) \mid \mathbf{q}(t) \right\} - a^\prime_k  \right], \label{eq:DrB'2}
\end{align}
in which $a_k=\beta a_k^\prime$. Because $\Delta_{B}^{\text{*}}$ maximizes the weighted sum $\sum_{k=1}^N q_k(t) \mathbb{E} \left\{  B_k(t) \mid \mathbf{q}(t) \right\}$ over all alternative decisions, we have 
\begin{align}
\sum_{k=1}^N q_k(t) \mathbb{E} \left\{  B^{(\Delta_{B}^{\text{*}})}_k(t) \mid \mathbf{q}(t) \right\}  \ge \sum_{k=1}^N q_k(t) \mathbb{E} \left\{  B^{(\Delta_{B})}_k(t) \mid \mathbf{q}(t)  \right\}.
\end{align}
where $\Delta_{B}$ represents any alternative (possibly randomized) scheduling decision that can be made on timeslot $t$. Plugging the above directly into \eqref{eq:DrB'2} yields 
\begin{align}
Dr^{(\Delta_{B^\prime}^{\text{*}})}(\mathbf{q}(t)) \le   E - \beta \sum_{k=1}^N q_k(t) \left[ \mathbb{E} \left\{  B^{(\Delta_{B})}_k(t) \mid \mathbf{q}(t) \right\} - a^\prime_k  \right].
\end{align}
%
Now suppose the arrival rate vector $\mathbf{a}^\prime$ is interior to the stability region. For these arrivals, there always exists an $\epsilon_{\text{max}}(\mathbf{a}^\prime)$ such that $\mathbb{E} \left\{  B^{(\Delta_B)}_k(t) \right\} \ge a^\prime_k+\epsilon_{\text{max}}(\mathbf{a^\prime})$, $\forall k \in \left\{1,\ldots,N \right\}$. Taking an expectation of $Dr^{(\Delta_{B^\prime}^{\text{*}})}$ over the randomness of the queue lengths and summing over $t \in \left\{ 0,1,\ldots,T-1 \right\}$ for some integer $T >0$ we get
\begin{align}
\mathbb{E}  \left\{ \mathit{Ly}(\mathbf{q}(T))  \right\} - \mathbb{E}  \left\{ \mathit{Ly}(\mathbf{q}(0))  \right\} \le E T - \epsilon_{\text{max}}(\mathbf{a}^\prime) \sum_{t=0}^{T-1} \sum_{k=1}^{N} \mathbb{E} \left\{ q_k(t) \right\}.
\end{align}
Rearranging terms, dividing by $\epsilon_{\text{max}}(\mathbf{a}^\prime) T$, and taking a $\limsup$ we eventually obtain
\begin{align}
\limsup\limits_{T \to \infty} \frac{1}{T} \sum_{t=0}^{T-1} \sum_{k=1}^{N} \mathbb{E} \left\{ q_k(t) \right\} \le \frac{E}{\epsilon_{\text{max}}(\mathbf{a}^\prime)}.
\end{align}
It follows that $\Delta_{B^\prime}^{\text{*}}$ stabilizes any arrival rate vector $\mathbf{a}=\beta \mathbf{a}^\prime$. Therefore, since $\mathbf{a}^\prime$ can be any point (vector) in the stability region of
$\Delta_{B}^{\text{*}}$, we can claim that $\Delta_{B^\prime}^{\text{*}}$ stabilizes any arrival rate vector interior to fraction $\beta$ of the stability region of  $\Delta_{B}^{\text{*}}$, meaning that $\Delta_{B^\prime}^{\text{*}}$ achieves up to $\Lambda_{B^\prime}=\frac{r(N,B^\prime)}{r(N,B)} \Lambda_B$. Note that this achievable region corresponds to the worst case, that is when the fraction is $\frac{r(N,B^\prime)}{r(N,B)}$. Hence, since the fraction is greater than or equal to $\frac{r(N,B^\prime)}{r(N,B)}$, we get
\begin{align}
\frac{r(N,B^\prime)}{r(N,B)} \Lambda_B  \subseteq \Lambda_{B^\prime} \subseteq \Lambda_B,
\end{align}
meaning that $\Lambda_{B^\prime}$ achieves  \emph{at least} a fraction $\frac{r(N,B^\prime)}{r(N,B)}$ of $\Lambda_{B}$. This completes the proof.
\end{proof}

\subsubsection{Reduce the Maximum Number of Pairs $N$}
Here, depending on whether $N$ and/or $N'$ are larger than $L_\text{I}$ or not (with $ N' \le N$), three cases are to investigate. We recall that $L_\text{I}$ is the number of pairs that maximizes the total average rate $r_\text{T}(L)$. It is worth noting that, here, we relax the condition on $L_\text{I}$, that is we consider the general case where $L_\text{I}$ can be greater than $N$, so that we can provide a complete study that covers all the possible cases. 
Before continuing the analysis, we provide some important remarks that are essential for a better understanding. When $L_\text{I} \ge N$, the stability region is characterized as the convex hull of the subsets $I_1, I_2,\dots, I_N$, or equivalently
\begin{align}
\mathcal{CH} \left\{ I_1, I_2, \dots, I_{N} \right\}.
\end{align} 
On the other hand, when $L_\text{I} \le N$, the stability region is given by 
\begin{align}
\mathcal{CH} \left\{ I_1, I_2, \dots, I_{L_\text{I}} \right\}.
\end{align} 
Note that these two claims result from Theorem \ref{th:stabregionSI}.
Using these remarks, we now provide the three different cases and study their impact on the stability region of our system.\\
{\footnotesize$\bullet$} We start by the case where $N' \le N \le L_\text{I}$. 
With a maximum number of pairs $N$, the stability region is characterized as the following
\begin{align}
\mathcal{CH} \left\{ I_1, I_2, \dots , I_{N^\prime-1},  I_{N^\prime} ,  I_{N^\prime+1}, \dots, I_{N} \right\},
\end{align} 
whereas with $N'$ this region is represented by
\begin{align}
\mathcal{CH} \left\{ I_1, I_2, \dots, I_{N^\prime} \right\}.
\end{align} 
It can be seen that the first stability region (with $N$) includes the second one (with $N^\prime$). Another important observation is that the gap between these two regions depends on the chosen direction. That is to say that the difference between the two convex hulls is not the same (i.e. asymmetric gap) and it depends on what direction and location we pick. In order to capture the maximum difference between these two regions, we proceed as follows: we choose any vertex from $I_{N}$ and then we find the fraction between the distance of this vertex from the origin and the distance from the origin to a point on the convex hull of $I_{N^\prime}$ and on the same line (from the origin) to the selected vertex. Note that the same approach was adopted for the comparison between the perfect and imperfect cases. Using the above and the result in Theorem \ref{th:imvsp_frac}, the minimum fraction between the two regions is 
\begin{align}
\frac{N^\prime \, r(N^\prime,B) }{ N \, r(N,B) }= \frac{r_\text{T}(N^\prime,B)}{r_\text{T}(N,B)}.
\end{align} 
{\footnotesize $\bullet$} The second case corresponds to consider $N^\prime \le L_\text{I} \le  N$. With $N$ as a maximum number of pairs, the system stability region is given by  
\begin{align}
\mathcal{CH} \left\{ I_1, I_2, \dots, I_{L_\text{I}} \right\},
\end{align} 
while with $N^\prime$ this region is characterized as 
\begin{align}
\mathcal{CH} \left\{ I_1, I_2, \ldots, I_{N^\prime} \right\}.
\end{align} 
Here, the common subsets of vertices between the two stability regions are $I_1, I_2, \dots, I_{N^\prime}$. Using a similar approach as for the first case, the fraction between the two convex hulls is equal to
\begin{align}
\frac{ N^\prime \, r(N^\prime ,B) }{L_\text{I} \, r(L_\text{I},B) }=\frac{ r_\text{T}(N^\prime ,B) }{ r_\text{T}(L_\text{I},B) }.
\end{align} 
{\footnotesize $\bullet$} The last case appears when $L_\text{I} \le N^\prime \le N$. With $N$ or $N^\prime$ as a maximum number of pairs, the stability region is the same and given as follows  
\begin{align}
\mathcal{CH} \left\{ I_1, I_2, \dots, I_{L_\text{I}} \right\}.
\end{align} 
It means that the gap between the convex hulls is zero in this case, hence the fraction we are looking for is equal to $1$.

\section{Algorithmic Design and Performance Analysis for the General Case}
\label{sec:generalcase}

We now consider a more general model where, unlike the symmetric case, the path loss coefficients are not necessarily equal to each other. However, for the sake of simplicity, and without loss of generality, we keep the same assumption on the number of streams, that is all the pairs have equal number of data streams, namely $d$. Also, as for the symmetric case, we suppose that the assigned rate is $R$ if the \ac{SINR} is greater than or equal to $\tau$; otherwise, the assigned rate is $0$.
We recall that $\mathcal{L}$  stands for the subset of active pairs, with $ \left| \mathcal{L} \right|=L$. Let $\alpha_{ki}=\frac{P \zeta_{ki}}{d}$ and $\alpha_{kk}=\frac{P \zeta_{kk}}{d}$.
Under these assumptions, and using $\eqref{eq:SINR}$, the \ac{SINR} for stream $m$ at active user $k$ can be expressed as
  \begin{align}
  \gamma_k^{(m)} =
  \begin{dcases}
  \, \, \frac{ \alpha_{kk} \left|  \left(\mathbf{\hat{u}}_k^{(m)} \right)^{\! H} \mathbf{H}_{kk} \mathbf{\hat{v}}_k^{(m)}  \right|^2  }{  \sigma^2 +  \sum\limits_{\substack{i \in \mathcal {L}, i \ne k}}   \alpha_{ki} \left\| \mathbf{h}_{ki} \right\|^2  e_{ki}  \sum\limits_{j=1}^{d}  \left| \mathbf{w}_{ki}^H \, \mathbf{T}_{k,i}^{(m,j)} \right|^2  },   & \qquad{} \text{imperfect case}  \\
   \, \, \frac{ \alpha_{kk} \left|  \left(\mathbf{u}_k^{(m)} \right)^{\! H} \mathbf{H}_{kk} \mathbf{v}_k^{(m)}  \right|^2  }{\sigma^2},   &  \qquad{} \text{perfect case} 
  \label{eq:SINR_G} 
  \end{dcases}
  \end{align}

\subsubsection{Imperfect Case}
Using Theorem \ref{th:AvgRate} and the fact that $\frac{\alpha_{ki}  }{ \alpha_{kk} }=  \frac{\zeta_{ki}  }{ \zeta_{kk} }$, the average rate of user $k$ ($\in \mathcal{L}$) can be given by the following expression
\begin{align}
r_k  =   (1-L\theta)d R   e^{-\frac{ \sigma^2 \tau }{\alpha_{kk}}}  \prod\limits_{i \in \mathcal{L},i\ne k}  \left( \frac{\zeta_{ki}\tau d    }{ \zeta_{kk}   2^{\frac{B}{Q}}   }+1 \right)^{\! -Q} \,_2F_1(\breve{b},Q;\breve{a}+\breve{b};\frac{1}{\frac{ \zeta_{kk} 2^{\frac{B}{Q}} }{ \zeta_{ki} \tau d} +1} ), \label{eq:r_k_G} 
\end{align}
where $\breve{a}=\frac{(Q+1)d}{Q}-\frac{1}{Q}$ and $\breve{b}= (Q-1)\breve{a}$. Let the average rate vector be $\mathbf{r}$, which contains $r_k$ in position $k$ if pair $k$ is active and $0$ otherwise. As mentioned previously, $\mathbf{s}$ and $\mathcal{L}$ are two different representations for the (same) set of active pairs, so we will use $\mathbf{r}(\mathbf{s})$ to represent the fact that $\mathbf{r}$ results from decision vector $\mathbf{s}$. Notice that, in contrast to the symmetric case, here the average rate expression depends on the identity of the active pairs. This lack of symmetry will make us incapable of finding the set of vertices of the corresponding stability region. However, we can still provide a (general) characterization of this stability region by considering all the possible decisions of scheduling the pairs, as follows
\begin{align}
\Lambda_\text{GI}=\mathcal{CH} \left\{ 	\mathit{GI}_1,\mathit{GI}_2,...,\mathit{GI}_N \right\},
\end{align}
where $\mathit{GI}_L=\left\{  \mathbf{r}(  \mathbf{s}) :  \mathbf{s} \in S_L  \right\}$.
To achieve this stability region we can apply the Max-Weight rule, which is an optimal scheduling policy, such as
\begin{align}
\Delta^{\text{*}}_{\text{GI}}  : \mathbf{s}(t)= \operatorname*{arg\,max}_{ \mathbf{s} \in   \mathcal{S}  }  \left\{ \mathbf{r}(  \mathbf{s})  \cdot \mathbf{q}(t)  \right\}.
\end{align}

\subsubsection{Perfect Case}
For this case, the average rate of active user $k$ is 
\begin{align}
\mu_k  =   (1-L\theta)d R   e^{-\frac{ \sigma^2 \tau }{\alpha_{kk}}}. \label{eq:mu_k_G} 
\end{align}
We denote by $\boldsymbol{\mu}$ the average rate vector that contains $\mu_k$ at position $k$ if pair $k$ is scheduled and $0$ otherwise. Also, let $\boldsymbol{\mu}( \mathbf{s})$ be the rate vector under decision vector $\mathbf{s}$. The stability region can be represented as
\begin{align}
\Lambda_\text{GP}=\mathcal{CH} \left\{ 	\mathit{GP}_1,\mathit{GP}_2,...,\mathit{GP}_N \right\},
\end{align}
where $\mathit{GP}_L=\left\{  \boldsymbol{\mu}(  \mathbf{s}) :  \mathbf{s} \in S_L  \right\}$.
The (optimal) policy that schedules the pairs and achieves this above region can be given by
\begin{align}
\Delta^{\text{*}}_{\text{GP}}  : \mathbf{s}(t)= \operatorname*{arg\,max}_{ \mathbf{s} \in   \mathcal{S}  }  \left\{ \bm{\mu}(  \mathbf{s})  \cdot \mathbf{q}(t)  \right\}.
\end{align}

\subsection{$\beta_\text{A}$-Approximate Policy and its Corresponding Achievable Stability Region}

As detailed earlier, under the considered system, the classical implementation of the Max-Weight policy, in both the perfect and imperfect cases, has a computational complexity of $O(N2^N)$. Whereas for the symmetric case some structural properties allowed us to find a low computational complexity implementation of this policy, here no such properties exist. To deal with this problem, we try to find an alternative policy that has a reduced computational complexity so that we can apply it instead of the corresponding optimal policy. Here we are interested in finding this alternative policy only under the imperfect case, which can be considered as the hardest case to analyze compared with the perfect one.
The alternative policy in this case is denoted by $\Delta^{\text{A}}$ and termed as $\beta_\text{A}$-\emph{approximate policy}, where this latter expression is justified by the fact that this policy approximates $\Delta_{\text{GI}}^{\text{*}}$ to a fraction of $\beta_\text{A}$. More specifically, for every queue length vector $\mathbf{q}$, the following holds
\begin{align}
 \nonumber \mathbb{E}\left\{ \mathbf{B}^{(\Delta_\text{GI}^{\text{*}})}(t) \cdot \mathbf{q}(t) \mid \mathbf{q}(t)=\mathbf{q} \right\} \le \beta_\text{A}^{-1}
 \mathbb{E}\left\{ \mathbf{B}^{(\Delta^{\text{A}})}(t) \cdot \mathbf{q}(t) \mid \mathbf{q}(t)=\mathbf{q} \right\},
\end{align}
or equivalently this can be represented as
\begin{align}
\nonumber (\mathbf{r} \cdot \mathbf{q})^{(\Delta_\text{GI}^{\text{*}})} \le \beta_\text{A}^{-1} (\mathbf{r} \cdot \mathbf{q})^{(\Delta^{\text{A}})}.  
\end{align}
%
%
%
%

For the rest of the paper, for notational conciseness, we will use the term “approximate policy” instead of “$\beta_\text{A}$-approximate policy” unless stated otherwise. 
The key step in this investigation is to determine a specific approximation of the average rate expression $r_k$, more specifically an approximation that possesses a set of structural features that let us define the approximate policy. 
Indeed, we will derive such approximation and see that it is very accurate if the fraction $\frac{ \zeta_{kk} 2^{\frac{B}{Q}} }{ \zeta_{ki} \tau d}$ (or equivalently, $\frac{ \alpha_{kk} 2^{\frac{B}{Q}} }{ \alpha_{ki} \tau d}$) is sufficiently high ($> 10$), $\forall i \ne k$. For fixed $\tau$ and $d$, this latter condition corresponds to a scenario where the number of quantization bits is high and/or the cross links have small path loss coefficients in comparison with the direct links (i.e. low interference scenario).
The approximation we are talking about is given by the following proposition.
\begin{proposition}
Given a subset of active pairs, $\mathcal{L}$, the rate of active user $k$ ($\in \mathcal{L}$) can be approximated as the following
\begin{align}
r_{k} \approx  (1-L\theta)d R e^{-\frac{ \sigma^2 \tau }{\alpha_{kk}}}  \prod\limits_{i \in \mathcal{L},i \ne k} (1-g_{ki}), \label{eq:r_approx}
\end{align}
where $g_{ki}= \left( \frac{ \zeta_{kk} 2^{\frac{B}{Q}} }{ \zeta_{ki} \tau d} +1 \right)^{-1}= \left( \frac{ \alpha_{kk} 2^{\frac{B}{Q}} }{ \alpha_{ki} \tau d} +1 \right)^{-1}$.
\end{proposition} 
\begin{proof}
To begin with, we note that the expression of $r_k$ can be re-expressed as 
\begin{align}
r_k  =   (1-L\theta)dR  e^{-\frac{ \sigma^2 \tau }{\alpha_{kk}}}  \prod\limits_{ i \in \mathcal{L}, i\ne k}  (1-g_{ki})^Q   \,_2F_1(\breve{b},Q;\breve{a}+\breve{b};g_{ki}),  \label{eq:r_k_proof}
\end{align}
which follows since $1-g_{ki}=\left( \frac{\zeta_{ki} \tau d  }{\zeta_{kk} 2^{\frac{B}{Q}} } +1 \right)^{-1}$.\\
We focus on the term $(1-g_{ki})^Q   \,_2F_1(\breve{b},Q;\breve{a}+\breve{b};g_{ki})$.
Using linear transformations (of variable) properties for the hypergeometric function, we have the relation 
\begin{align}
(1-g_{ki})^{Q} \,_2F_1(\breve{b},Q;\breve{a}+\breve{b};g_{ki})= \,_2F_1(\breve{a},Q;\breve{a}+\breve{b};\frac{g_{ki}}{g_{ki}-1}). \label{eq:hypergeom}
\end{align} 
For sufficiently small $g_{ki}$ values, we can (numerically) verify that the following accurate approximation holds
\begin{align}
\,_2F_1(\breve{a},Q;\breve{a}+\breve{b};\frac{g_{ki}}{g_{ki}-1} ) \approx \,_2F_1(\breve{a},Q;\breve{a}+\breve{b};-g_{ki}).
\end{align}	
We recall that $\breve{a}=\frac{(Q+1)d}{Q}-\frac{1}{Q}$, $\breve{b}=(Q-1)\breve{a}$ and $Q=N_\text{t}N_\text{r}-1$. Thus, for sufficiently large $Q$, and since $d \le \min(N_\text{t},N_\text{r})$ (this ensures that $Q$ is sufficiently larger than $d$), we can easily see that $\breve{a}\approx d$, $\breve{b}\approx Qd-d$ and $\breve{a}+\breve{b}\approx Qd$. Now, using the Maclaurin expansion to the second order we can write 
\begin{align}
\,_2F_1(\breve{a},Q;\breve{a}+\breve{b};-g_{ki}) \approx 1-\frac{\breve{a}Q}{\breve{a}+\breve{b}}g_{ki}+\frac{1}{2} \frac{\breve{a}Q}{\breve{a}+\breve{b}}  \frac{(\breve{a}+1) (Q+1)}{\breve{a}+\breve{b}+1}g_{ki}^2 + o(g_{ki}^2) \approx 1-g_{ki}. \label{eq:MacExpansion}
\end{align}
In this approximation we used the facts that $\frac{1}{2}g_{ki}^2 \ll g_{ki}$, $\frac{\breve{a}Q}{\breve{a}+\breve{b}}=\frac{dQ}{Qd}=1$ and $\frac{ (\breve{a}+1)(Q+1)}{\breve{a}+\breve{b}+1}=\frac{ (d+1)(Q+1)}{Qd+1} \approx 1$. In addition, $o(g_{ki}^2)$ can be removed since it is negligible compared with $1-g_{ki}$. 
This latter property follows from the fact that the Maclaurin expansion to higher orders (greater than two) will add terms in $g_{ki}^3,g_{ki}^4,\dots$, which are, as $g_{ki}^2$, very small with respect to $1$ and to the term in $g_{ki}$; this is due to the condition $g_{ki} <0.1$.
Hence, by replacing the above approximation in the expression of $r_k$ given in $\eqref{eq:r_k_proof}$, we obtain
\begin{align}
r_{k} \approx  (1-L\theta)d R e^{-\frac{ \sigma^2 \tau }{\alpha_{kk}}}  \prod\limits_{i \in \mathcal{L},i \ne k} (1-g_{ki}).
\end{align} 
This concludes the proof.

\end{proof}

To proceed further with the analysis, we recall that we denote the approximate policy by $\Delta^{\text{A}}$. Let $\bar{g}_{k}$ be the average value of all the $g_{ki}$, with $i\ne k$, for the same number of active pairs $L$. 
In detail, for a fixed cardinality $L$,
we take all the possible subsets (i.e. scheduling decisions) in which user $k$ is active. For each of these subsets, there are $L-1$ values of $g_{ki}$. Hence, $\bar{g}_{k}$ is the average of these $g_{ki}$ values over all the considered decisions. Using this average value $\bar{g}_k$ and the approximate expression of $r_k$ in $\eqref{eq:r_approx}$, we define $\phi_{k}(L)$ as
\begin{align}
\phi_{k}(L)= (1-L\theta)d R e^{-\frac{ \sigma^2 \tau }{\alpha_{kk}}}  (1-\bar{g}_k)^{L-1}.
\end{align}
Also, we denote $\boldsymbol{\phi}$ as the vector containing $\phi_{k}(L)$ at position $k$ if  pair $k$ is scheduled (with $L-1$ other pairs); otherwise, we set $0$ at this position.
Under this setting, we define the approximate policy $\Delta^{\text{A}}$ as follows
\begin{align}
 \nonumber   \Delta^{\text{A}}  : \mathbf{s}(t) = \operatorname*{arg\,max}_{ \mathbf{s} \in   \mathcal{S}  }  \left\{   \boldsymbol{\phi}   (  \mathbf{s} )   \cdot \mathbf{q}(t)  \right\},
\end{align}
where $\boldsymbol{\phi}   (  \mathbf{s} )$ results from decision vector $\mathbf{s}$. It is noteworthy to mention that although we use $\phi_{k}(L)$ (of active pair $k$) to make the scheduling decision under $\Delta^{\text{A}}$, the actual average rate of user $k$ is still $r_k$. Also, remark that $\Delta^{\text{A}}$ follows the Max-Weight rule, thus, as was shown earlier, implementing $\Delta^{\text{A}}$ as a classical maximization problem over all the possible decisions $\mathbf{s}$ needs a CC of $O(N2^N)$. However, in contrast to $\Delta_\text{GI}^\text{*}$, policy $\Delta^{\text{A}}$ has a structural property that will
allow us to propose an equivalent reduced CC implementation instead of the classical one. This is due to the fact that, for example, $\phi_k(L)$ is independent of the $L-1$ other active users, but only depends on pair $k$ and the cardinality $L$.
The proposed implementation of policy $\Delta^{\text{A}}$ is given by Algorithm \ref{al:schedecis_G}.
 \begin{algorithm}[ht!]
 	\caption{: A Reduced Computational Complexity Implementation of $\Delta^{\text{A}}$}
 \begin{algorithmic}[1]
 \State Initialize $L_\text{g}=0$ and $\mathit{ws}_{L_\text{g}}=0$.
 \For {$l=1:1:N$}
 \State Sort the users in a descending order with respect to the product $\mathit{pro}_k=\phi_{k}(l)\,q_k$.
 \State Let $\mathit{ws}_l=$ sum of the first $l$ biggest $\mathit{pro}_k$ values; save  $\mathit{sq}_l$ that represents which $l$ users yield $\mathit{ws}_l$.
 \If {$\mathit{ws}_l >  \mathit{ws}_{L_\text{g}}$} 
 \State Put $\mathit{ws}_{L_\text{g}} = \mathit{ws}_{l}$ and $L_\text{g}=l$.
 \EndIf
 \EndFor
 \State Schedule the pairs given by $\mathit{sq}_{L_\text{g}}$.
 	\end{algorithmic}
 	\label{al:schedecis_G}
 \end{algorithm}
 
To compare with the classical implementation, we now focus on the computational complexity of the proposed implementation, which depends essentially on a “for loop” of $N$ iterations, each of which contains: (i) a “sorting algorithm”, which needs in the worst case $O(N^2)$, (ii) a sum of $l$ terms in iteration $l$, and (iii) other steps of small CC compared with those mentioned before. Thus, by neglecting the CC of the steps in (iii) and noticing that the summing steps (in (ii)) over all the iterations need $O(\frac{N(N+1)}{2})=O(N^2)$, the CC of the proposed implementation is roughly $O(N^2 N+N^2)=O(N^3)$, which is very small compared with $O(N2^N)$ for large $N$.

In general, the approximate policy comes with the disadvantage of reducing the achievable stability region compared with the optimal policy. Indeed, as we will see later on, policy $\Delta^{\text{A}}$ only achieves a fraction of the stability region achieved by policy $\Delta^{\text{*}}_\text{GI}$. We recall that $g_{ki}=( \frac{ \zeta_{kk} 2^{\frac{B}{Q}} }{ \zeta_{ki} \tau d} +1 )^{-1}$.
Let us  define $\mathcal{L}_\text{A}$ as the subset chosen by $\Delta^{\text{A}}$, and we let the cardinality of this subset be $L_\text{A}$. For $\Delta_\text{GI}^{\text{*}}$ we keep the original notation, that is $\mathcal{L}$ is the scheduled subset, with $L=\left| \mathcal{L} \right|$. Moreover, let $\bm{\mathcal{L}}$ stand for the set of all possible decision subsets, so $\mathcal{L}_\text{A}$ and $\mathcal{L}$ are subsets from $\bm{\mathcal{L}}$. Remark that $\bm{\mathcal{L}}$ is just an equivalent representation of set $\mathcal{S}$. In the following, we state a proposition that is essential for the characterization of the achievable fraction. 
\begin{proposition}
The approximation of $r_k$ in $\eqref{eq:r_approx}$ can in its turn be approximated as
\label{pr:approximation2}
\begin{align}
r_k \approx \left( 1-L\theta\right) dR e^{-\frac{ \sigma^2 \tau }{\alpha_{kk}}} \left[ (1-\bar{g}_k)^{L-1}-(1-\bar{g}_{k})^{L-2}\sum\limits_{ i \in \mathcal{L},i \ne k } (g_{ki}-\bar{g}_k) \right].
\end{align} 
\end{proposition}
\begin{proof}
We focus on the product $\prod\limits_{i \in \mathcal{L},i \ne k} (1-g_{ki})$. Its Taylor series to first order about the point $(\bar{g}_{k},\dots,\bar{g}_{k} )$ can be written as
\begin{align}
(1-\bar{g}_{k})^{L-1}- (1-\bar{g}_{k})^{L-2}  \sum\limits_{i \in \mathcal{L},i \ne k }  \left( g_{ki}-\bar{g}_{k} \right).
\end{align}
To prove this latter result, we start by a simple example and then we provide the result for the general case.
We consider a simple case where the product function is $f(g_{k1},g_{k2})=(1-g_{k1})(1-g_{k2})$, i.e. it corresponds to $L=3$ and $k=3$, then we can write its Taylor series to first order about $(\bar{g}_{k},\bar{g}_{k} )$ as
\begin{align}
f(g_{k1},g_{k2})\nonumber & \approx
f(\bar{g}_{k}, \bar{g}_{k})+ (g_{k1}-\bar{g}_{k}) \frac{\partial f}{\partial g_{k1} } |_{(\bar{g}_{k}, \bar{g}_{k})} + (g_{k2}-\bar{g}_{k}) \frac{\partial f}{\partial g_{k2} } |_{(\bar{g}_{k}, \bar{g}_{k})}  \\ &=(1-\bar{g}_{k}) (1-\bar{g}_{k})-(g_{k1}-\bar{g}_{k})(1-\bar{g}_{k})- (g_{k2}-\bar{g}_{k})(1-\bar{g}_{k})  \nonumber \\ &= (1-\bar{g}_k)^{3-1}-(1-\bar{g}_{k})^{3-2}(g_{k1}-\bar{g}_{k}+g_{k2}-\bar{g}_{k}).
\end{align} 
Note that the higher order elements of the above expansion are removed since they are very small compared with the other elements.
The obtained result can be easily generalized as 
\begin{align}
\prod\limits_{i \in \mathcal{L},i \ne k} (1-g_{ki}) \approx (1-\bar{g}_k)^{L-1}-(1-\bar{g}_{k})^{L-2}\sum\limits_{ i \in \mathcal{L},i \ne k}(g_{ki}-\bar{g}_{k}). \label{eq:f_approx_exp}
\end{align}
By replacing the function $\prod\limits_{i \in \mathcal{L},i \ne k} (1-g_{ki})$ with its expansion provided in $\eqref{eq:f_approx_exp}$ and recalling that the approximated rate expression is $(1-L\theta) dR e^{-\frac{ \sigma^2 \tau }{\alpha_{kk}}}\prod\limits_{i \in \mathcal{L},i \ne k} (1-g_{ki}) $, we eventually get 
\begin{align}
r_k \approx \left(1-L\theta\right) dR e^{-\frac{ \sigma^2 \tau }{\alpha_{kk}}} \left[ (1-\bar{g}_k)^{L-1}-(1-\bar{g}_{k})^{L-2}\sum\limits_{ i \in \mathcal{L},i \ne k } (g_{ki}-\bar{g}_k) \right].
\end{align}
Therefore, the desired result holds.
\end{proof}
Based on the above, we now provide the main result of this subsection, that is the stability region achieved by the approximate policy $\Delta^{\text{A}}$.
\begin{theorem}
The approximate policy $\Delta^{\text{A}}$ achieves at least a fraction $\beta_\text{A}$ ($\le 1$) of the stability region achieved by the optimal policy $\Delta^{\text{*}}_\text{GI}\,$, where $\beta_\text{A}$ is given by
\begin{align}
\beta_\text{A}=\frac{1+\, \underset{\mathcal{L}_\text{A}\in \bm{\mathcal{L}} }{\min}    \left\{  \underset{k \in \mathcal{L}_\text{A}}{\min} \left\{ -(1-\bar{g}_k)^{-1}   \sum\limits_{i \in \mathcal{L}_\text{A}, i \ne k} (g_{ki}-\bar{g}_k) \right\}    \right\}  }{1+\, \underset{\mathcal{L} \in \bm{\mathcal{L}} }{\max}  \left\{ \underset{k \in\mathcal{L}}{\max }  \left\{ - (1-\bar{g}_k)^{-1}   \sum\limits_{i \in \mathcal{L}, i \ne k} (g_{ki}-\bar{g}_k) \right\} \right\}  }.
\end{align} 
\label{th:beta_A}
\end{theorem}

\begin{proof}
Using Proposition \ref{pr:approximation2}, under policy $\Delta^\text{A}$ the dot product $\mathbf{r} \cdot \mathbf{q}$ can be expressed as 
\begin{align}
(1-L_\text{A} \theta) d R  \left[  \sum\limits_{k \in \mathcal{L}_\text{A}}  e^{-\frac{ \sigma^2 \tau }{\alpha_{kk}}}     (1-\bar{g}_k)^{L_\text{A}-1} q_k- \sum\limits_{k \in \mathcal{L}_\text{A}} e^{-\frac{ \sigma^2 \tau }{\alpha_{kk}}}   (1-\bar{g}_k)^{L_\text{A}-2} q_k \sum\limits_{i \in \mathcal{L}_\text{A},i \ne k}(g_{ki}-\bar{g}_k) \right]  ,
\end{align}
whereas under $\Delta^{\text{*}}_\text{GI}$ this dot product is given by
\begin{align}
(1-L \theta) d R \left[  \sum\limits_{k \in \mathcal{L} }  e^{-\frac{ \sigma^2 \tau }{\alpha_{kk}}}      (1-\bar{g}_k)^{L-1} q_k- \sum\limits_{k \in \mathcal{L}}  e^{-\frac{ \sigma^2 \tau }{\alpha_{kk}}}      (1-\bar{g}_k)^{L-2} q_k \sum\limits_{i \in \mathcal{L}, i \ne k}(g_{ki}-\bar{g}_k)  \right].
\end{align}
Since the approximate policy $\Delta^\text{A}$ schedules the subset $\mathcal{L}_\text{A}$ that maximizes the dot product $\boldsymbol{\phi}  \cdot \mathbf{q}$, and recalling that $\phi_k(l)=(1-l \theta) d R e^{-\frac{ \sigma^2 \tau }{\alpha_{kk}}}   (1-\bar{g}_k)^{l-1}$, it yields
\begin{align}
(1-L_\text{A} \theta) d R \sum\limits_{k \in \mathcal{L}_\text{A} }  e^{-\frac{ \sigma^2 \tau }{\alpha_{kk}}}   (1-\bar{g}_k)^{L_\text{A} -1} q_k  \ge   (1-L \theta)  d R \sum\limits_{k \in \mathcal{L} }   e^{-\frac{ \sigma^2 \tau }{\alpha_{kk}}}   (1-\bar{g}_k)^{L-1} q_k. \label{eq:ineq1}
\end{align}
Similarly, using the definition of the optimal policy $\Delta^{\text{*}}_\text{GI}$ under which the dot product $\mathbf{r}  \cdot \mathbf{q}$ is maximized, the following inequality holds
\begin{align}
\nonumber  &  (1-L \theta)  d R \left[   \sum\limits_{k \in \mathcal{L} }  e^{-\frac{ \sigma^2 \tau }{\alpha_{kk}}}      (1-\bar{g}_k)^{L-1} q_k- \sum\limits_{k \in \mathcal{L} }  e^{-\frac{ \sigma^2 \tau }{\alpha_{kk}}}      (1-\bar{g}_k)^{L-2} q_k \sum\limits_{i \in \mathcal{L}, i \ne k}(g_{ki}-\bar{g}_k)  \right] \, \ge \\    &   (1-L_\text{A} \theta)  d R \left[   \sum\limits_{k \in \mathcal{L}_\text{A}}  e^{-\frac{ \sigma^2 \tau }{\alpha_{kk}}}     (1-\bar{g}_k)^{L_\text{A}-1} q_k- \sum\limits_{k \in \mathcal{L}_\text{A}} e^{-\frac{ \sigma^2 \tau }{\alpha_{kk}}}   (1-\bar{g}_k)^{L_\text{A}-2} q_k \sum\limits_{i \in \mathcal{L}_\text{A}, i \ne k}(g_{ki}-\bar{g}_k)   \right] .	 \label{eq:ineq2}
\end{align}
These last two inequalities, in $\eqref{eq:ineq1}$ and $\eqref{eq:ineq2}$, lead us to the simple observation
\begin{align}
\nonumber   & - (1-L\theta) d R \left[ \sum\limits_{k \in \mathcal{L}}  e^{-\frac{ \sigma^2 \tau }{\alpha_{kk}}} (1-\bar{g}_k)^{L-2} q_k \sum\limits_{i \in \mathcal{L}, i \ne k} (g_{ki}-\bar{g}_k) \right]  \,  \ge \\ &
- (1-L_\text{A} \theta) d R \left[  \sum\limits_{k \in \mathcal{L}_\text{A} }  e^{-\frac{ \sigma^2 \tau }{\alpha_{kk}}}  (1-\bar{g}_k)^{L_\text{A}-2} q_k  \sum\limits_{i \in \mathcal{L}_\text{A} , i \ne k} (g_{ki}-\bar{g}_k) \right] .  
\end{align} 
For the rest of this proof, we define 
\begin{align}
\nonumber o_1= (1-L_\text{A} \theta) d R \sum\limits_{k \in \mathcal{L}_\text{A} } e^{-\frac{ \sigma^2 \tau }{\alpha_{kk}}}   (1-\bar{g}_k)^{L_\text{A}-1} q_k,
\end{align}
\begin{align}
\nonumber p_1= - (1-L_\text{A} \theta) d R \left[\sum\limits_{k \in \mathcal{L}_\text{A}}  e^{-\frac{ \sigma^2 \tau }{\alpha_{kk}}}   (1-\bar{g}_k)^{L_\text{A}-2} q_k  \sum\limits_{i \in \mathcal{L}_\text{A}, i \ne k}  (g_{ki}-\bar{g}_k) \right],  
\end{align}
\begin{align}
 \nonumber o_2= (1-L \theta) d R \sum\limits_{k \in \mathcal{L} } e^{-\frac{ \sigma^2 \tau }{\alpha_{kk}}}  (1-\bar{g}_k)^{L-1} q_k,
\end{align}
\begin{align}
\nonumber p_2= - (1-L \theta) d R \left[ \sum\limits_{k \in \mathcal{L}}  e^{-\frac{ \sigma^2 \tau }{\alpha_{kk}}}  (1-\bar{g}_k)^{L-2} q_k \sum\limits_{i \in \mathcal{L}, i \ne k} (g_{ki}-\bar{g}_k) \right]. 
\end{align}
We can easily notice that $p_1$ and $p_2$ can be rewritten, respectively, as
\begin{align}
p_1= -(1-L_\text{A} \theta) d R \left[ \sum\limits_{k \in \mathcal{L}_\text{A}}   e^{-\frac{ \sigma^2 \tau }{\alpha_{kk}}}   (1-\bar{g}_k)^{-1} (1-\bar{g}_k)^{L_\text{A}-1} q_k \sum\limits_{i \in \mathcal{L}_\text{A}, i \ne k}  (g_{ki}-\bar{g}_k) \right], 
\end{align}
\begin{align}
p_2= -(1-L \theta) d R \left[ \sum\limits_{k \in \mathcal{L} }   e^{-\frac{ \sigma^2 \tau }{\alpha_{kk}}}   (1-\bar{g}_k)^{-1} (1-\bar{g}_k)^{L-1} q_k \sum\limits_{i \in \mathcal{L}, i \ne k} (g_{ki}-\bar{g}_k) \right] . 
\end{align}

We next point out two simple but important properties that will help us complete the proof.\\ 
{\footnotesize$\bullet$} For any policy $\Delta_2$ that approximates any policy $\Delta_1$ to a fraction $\beta$ ($\le 1$), we have $(\mathbf{r} \cdot \mathbf{q})^{(\Delta_1)} \le \beta^{-1}(\mathbf{r} \cdot \mathbf{q})^{(\Delta_2)}$. If w.r.t. the approximate policy ($\Delta_2$) there exists a scheduling policy $\Delta_{22}$ such that $(\mathbf{r} \cdot \mathbf{q})^{(\Delta_{22})} \le (\mathbf{r} \cdot \mathbf{q})^{(\Delta_2)}$, then we can derive a fraction based on $(\mathbf{r} \cdot \mathbf{q})^{(\Delta_{22})}$ instead of $(\mathbf{r} \cdot \mathbf{q})^{(\Delta_2)}$. We can easily notice that this fraction is lower than or equal to $\beta$, therefore, w.r.t. the stability region achieved by $\Delta_1$, $\Delta_2$ reaches a fraction larger than that achieved by $\Delta_{22}$.\\
{\footnotesize$\bullet$} On the other side, if w.r.t. the approximated policy ($\Delta_1$) there exists a scheduling policy $\Delta_{11}$ such that $(\mathbf{r} \cdot \mathbf{q})^{(\Delta_1)} \le (\mathbf{r} \cdot \mathbf{q})^{(\Delta_{11})} $, then we can derive an achievable fraction based on $(\mathbf{r} \cdot \mathbf{q})^{(\Delta_{11})}$, and this fraction will be lower than or equal to $\beta$. The key idea here is that sometimes it is easier to find the fraction using $\Delta_{22}$ (resp., $\Delta_{11}$) instead of $\Delta_2$ (resp., $\Delta_1$ ), but this will be to the detriment of finding an achievable fraction that is, in general, lower than the exact solution.

To proceed further, we consider the extreme case that corresponds to define
\begin{align}
\nonumber p_{1\text{e}}= \underset{\mathcal{L}_\text{A} \in \bm{\mathcal{L}} }{\min}    \left\{    \underset{k \in \mathcal{L}_\text{A}}{\min} \left\{-    (1-\bar{g}_k)^{-1}   \sum\limits_{i \in \mathcal{L}_\text{A}, i \ne k} (g_{ki}-\bar{g}_k)  \right\}   \right\}  (1-L_\text{A}\theta) d R  \sum\limits_{k \in \mathcal{L}_\text{A}} e^{-\frac{ \sigma^2 \tau }{\alpha_{kk}}}  (1-\bar{g}_k)^{L_\text{A}-1}  q_k,
\end{align}
\begin{align}
\nonumber p_{2\text{e}}=  \underset{\mathcal{L} \in \bm{\mathcal{L}} }{\max}  \left\{ \underset{k \in \mathcal{L}}{\max } \left\{ - (1-\bar{g}_k)^{-1}   \sum\limits_{i \in \mathcal{L}, i \ne k} (g_{ki}-\bar{g}_k) \right\} \right\}   (1-L\theta) d R \sum\limits_{k \in \mathcal{L}} e^{-\frac{ \sigma^2 \tau }{\alpha_{kk}}}  (1-\bar{g}_k)^{L-1}  q_k.
\end{align}
It is obvious that $p_1 \ge p_{1\text{e}}$ and $p_2 \le p_{2\text{e}}$.
Let us define $m_1$ and $m_2$ as 
\begin{align}
\nonumber m_1=  \underset{\mathcal{L}_\text{A} \in \bm{\mathcal{L}} }{\min}    \left\{ \underset{k \in \mathcal{L}_\text{A} }{\min} \left\{ -   (1-\bar{g}_k)^{-1}   \sum\limits_{i \in \mathcal{L}_\text{A} , i \ne k} (g_{ki}-\bar{g}_k) \right\}  \right\},  
\end{align}
\begin{align}
\nonumber m_2= \underset{\mathcal{L} \in \bm{\mathcal{L} } } {\max}  \left\{ \underset{k \in \mathcal{L} }  {\max }  \left\{ - (1-\bar{g}_k)^{-1}   \sum\limits_{i \in \mathcal{L}, i \ne k} (g_{ki}-\bar{g}_k) \right\} \right\}. 
\end{align}
Then, it is easy to see that $p_{1\text{e}}=m_1 o_1$ and $p_{2\text{e}}=m_2 o_2$. This yields the following
\begin{align}
(\mathbf{r} \cdot \mathbf{q} \, )^{(\Delta^{\text{A}})} = o_1 + p_1 \ge o_1+p_{1e} = o_1+m_1o_1, \label{eq:delta_A}
\end{align}
\begin{align}
(\mathbf{r} \cdot \mathbf{q} \, )^{(\Delta^{\text{*}}_\text{GI})}  = o_2+p_2 \le o_2+ p_{2e} =o_2+m_2o_2. \label{eq:delta_GI}
\end{align} 
As mentioned earlier, $\Delta^\text{A}$ approximates $\Delta_\text{GI}^{\text{*}}$ to a fraction $\beta$ if the following inequality holds
\begin{align}
(\mathbf{r} \cdot \mathbf{q} \, )^{(\Delta^{\text{*}}_\text{GI})} \le \beta^{-1} (\mathbf{r} \cdot \mathbf{q} \, )^{(\Delta^{\text{A}})}.
\end{align}
In our case, it is difficult to derive $\beta$, however we can compute a fraction $\beta_\text{A} \le \beta$. In detail, based on the two properties  about the achievable fraction given in the above paragraph, and combining \eqref{eq:delta_A} with \eqref{eq:delta_GI}, the problem turns out to find $\beta_\text{A}$ such that
\begin{align}
o_2(1+m_2) \le \beta_\text{A}^{-1} o_1(1+m_1).  \label{eq:2_r2_r1}
\end{align}  
Using the fact that $o_2 \le o_1$, which was shown at the beginning of this proof, it suffices to have $\beta_\text{A}^{-1} \ge \frac{1+m_2}{1+m_1}$, or equivalently $\beta_\text{A} \le \frac{1+m_1}{1+m_2}$, to satisfy the inequality in $\eqref{eq:2_r2_r1}$. Let us take $\beta_\text{A}=\frac{1+m_1}{1+m_2}$. 
Now, to complete the proof, we use a similar approach to that used in Step 3 of the proof for Theorem \ref{th:BtoB'}. Specifically, the drift under $\Delta^{\text{A}}$ can be expressed as 
\begin{align}
Dr^{(\Delta^{\text{A}})}(\mathbf{q}(t)) \le   E - \sum_{k=1}^{N} q_k(t)  \left[ \mathbb{E} \left\{  B^{(\Delta^{\text{A}})}_k(t) \mid \mathbf{q}(t) \right\} - a_k \right], \label{eq:DrB'_g} 
\end{align}
for some finite constant $E$. Hence, using $\eqref{eq:2_r2_r1}$ and the fact that $\beta_\text{A} \le \beta$, we can write
\begin{align}
(\mathbf{r} \cdot \mathbf{q} \, )^{(\Delta^{\text{*}}_\text{GI})} \le \beta_\text{A}^{-1} (\mathbf{r} \cdot \mathbf{q} \, )^{(\Delta^{\text{A}})},
\end{align}
or equivalently
\begin{align}
\sum_{k=1}^N q_k(t) \mathbb{E} \left\{  B^{(\Delta^{\text{A}})}_k(t) \mid \mathbf{q}(t) \right\}  \ge \sum_{k=1}^N q_k(t) \beta_\text{A} \, \mathbb{E} \left\{  B^{(\Delta^{\text{*}}_\text{GI})}_k(t) \mid \mathbf{q}(t)  \right\}.
\end{align}
Plugging this directly into $\eqref{eq:DrB'_g}$ yields
\begin{align}
Dr^{(\Delta^{\text{A}})}(\mathbf{q}(t))  \le   E - \beta_\text{A} \sum_{k=1}^N q_k(t) \left[ \mathbb{E} \left\{  B^{(\Delta^{\text{*}}_\text{GI})}_k(t) \mid \mathbf{q}(t) \right\} - a^\prime_k  \right], 
\end{align}
in which $a_k=\beta_\text{A}  a_k^\prime$.
After some manipulations, which are very similar to those used in the proof of Theorem \ref{th:BtoB'}, we eventually obtain
\begin{align}
\limsup\limits_{T \to \infty} \frac{1}{T} \sum_{t=0}^{T-1} \sum_{k=1}^{N} \mathbb{E} \left\{ q_k(t) \right\} \le \frac{E}{\epsilon_{\text{max}}(\mathbf{a}^\prime)},
\end{align}
for some $\epsilon_{\text{max}}(\mathbf{a}^\prime)$ and $T$.
It follows that $\Delta^{\text{A}}$ stabilizes any arrival rate vector $\mathbf{a}=\beta_\text{A} \mathbf{a}^\prime$. Hence, since $\mathbf{a}^\prime$ can be any point in the stability region of
$\Delta^{\text{*}}_\text{GI}$, we can state that $\Delta^{\text{A}}$ stabilizes any arrival rate vector interior to fraction $\beta_\text{A}$ of the stability region of  $\Delta^{\text{*}}_\text{GI}$.
A last point to note is that the term “at least” in the theorem is justified by the fact that $\beta_\text{A}$ is lower than or equal to the exact solution ($\beta$). Therefore, the desired statement follows. 
\end{proof}

\subsection{Compare the Imperfect Case with the Perfect Case in Terms of Stability}
At the very beginning of this section, we showed that policy 
$\Delta_{\text{GP}}^\text{*}$ achieves the system stability region in the perfect case. 
Let us denote by $\mathcal{L}_\text{P}$ the subset of scheduled pairs using $\Delta_{\text{GP}}^\text{*}$ and by $L_\text{P}$ the cardinality of this subset. On the other side, for the imperfect case, we adopt the same notation as before, i.e. the subset of scheduled users and its cardinality are represented by $\mathcal{L}$ and $L$, respectively. 
Here, an essential parameter to investigate is the fraction the stability region the imperfect case achieves compared with the stability region of the perfect case. This fraction is captured in the following theorem.
\begin{theorem}
The stability region of the imperfect case reaches at least a fraction $\beta_\text{P}$ of the stability region achieved in the perfect case, where 
\begin{align}
\beta_\text{P} = \underset{\mathcal{L}_\text{P} \in \bm{\mathcal{L}} }{\min} \left\{ \underset{k \in \mathcal{L}_\text{P} }{\min} \, \left\{ \prod\limits_{i \in \mathcal{L}_\text{P} ,i\ne k} (1-g_{ki})   \right\} \right\} .
\end{align}
We recall that $g_{ki}=\left( \frac{ \zeta_{kk} 2^{\frac{B}{Q}} }{ \zeta_{ki} \tau d} +1 \right)^{-1}$.
\end{theorem}
\begin{proof}
Under policy $\Delta^\text{*}_\text{GI}$, and using the approximate expression of $r_k$ given in $\eqref{eq:r_approx}$,  the dot product $(\mathbf{r} \cdot \mathbf{q})^{ \Delta^\text{*}_\text{GI} }$ can be written as 
\begin{align}
(1-L\theta) \left[ \sum\limits_{k \in \mathcal{L}} d R e^{-\frac{ \sigma^2 \tau }{\alpha_{kk}}} \, q_k \prod\limits_{i \in \mathcal{L},i\ne k}(1-g_{ki}) \right]. 
\end{align}
On the other hand, using the definition of $\Delta_\text{GP}^\text{*}$, the product $(\bm{\mu} \cdot \mathbf{q})^{(\Delta_\text{GP}^\text{*})}$ can be expressed as 
\begin{align}
 (1-L_\text{P}\theta) \sum\limits_{k \in \mathcal{L}_\text{P}} d R  e^{-\frac{ \sigma^2 \tau }{\alpha_{kk}}} q_k.
\end{align}
One can easily remark that this latter expression has the following equivalent representation, which results from multiplying and dividing by the same term $\prod\limits_{i \in \mathcal{L}_\text{P},i\ne k}(1-g_{ki})$,
\begin{align}
(1-L_\text{P}\theta) \sum\limits_{k \in \mathcal{L}_\text{P}} d R e^{-\frac{ \sigma^2 \tau }{\alpha_{kk}}} \,  q_k \left[\frac{    \prod\limits_{i \in  \mathcal{L}_\text{P} ,i\ne k}(1-g_{ki}) }   {  \prod\limits_{i \in  \mathcal{L}_\text{P} ,i\ne k}(1-g_{ki})    }  \right].
\end{align}
The extreme case of  $(\bm{\mu} \cdot \mathbf{q})^{(\Delta_\text{GP}^\text{*})}$ corresponds to 
\begin{align}
m_3 (1-L_\text{P} \theta) \left[ \sum\limits_{k \in \mathcal{L}_\text{P} } d R  e^{-\frac{ \sigma^2 \tau }{\alpha_{kk}}} \, q_k \prod\limits_{i \in \mathcal{L}_\text{P} ,i\ne k} (1-g_{ki})    \right], \label{eq:muqextreme}
\end{align}
where $m_3^{-1} = \underset{\mathcal{L}_\text{P} \in \mathcal{L}_\text{P} }{\min} \left\{ \underset{k \in \mathcal{L}_\text{P} }{\min} \, \left\{ \prod\limits_{i \in \mathcal{L}_\text{P} ,i\ne k} (1-g_{ki})      \right\}  \right\} $.
Since, by definition, policy $\Delta_{\text{GI}}^{\text{*}}$ produces $\mathcal{L}$ and maximizes the product $(\mathbf{r} \cdot \mathbf{q})$, it yields
\begin{align}
 (1-L\theta) \left[ \sum\limits_{k \in \mathcal{L} } d R e^{-\frac{ \sigma^2 \tau }{\alpha_{kk}}}   \,  q_k \prod\limits_{i \in  \mathcal{L},i\ne k}(1-g_{ki})  \right]  \ge 
 (1-L_\text{P}\theta) \left[ \sum\limits_{k \in \mathcal{L}_\text{P} } d R e^{-\frac{ \sigma^2 \tau }{\alpha_{kk}}}  \,  q_k \prod\limits_{i \in  \mathcal{L}_\text{P} ,i\ne k}(1-g_{ki})  \right]. \label{eq:ineq_per}
\end{align}
As explained earlier, the stability region achieved by $\Delta^{\text{*}}_\text{GI}$ approximates the one achieved by $\Delta_\text{GP}^{\text{*}}$ to a fraction $\beta$ if 
\begin{align}
(\bm{\mu} \cdot \mathbf{q})^{(\Delta_{\text{GP}}^\text{*} )}  \le \beta^{-1} (\mathbf{r} \cdot \mathbf{q})^{(\Delta^{\text{*}}_\text{GI})}.
\end{align} 
It is hard to find $\beta$ based on $(\bm{\mu} \cdot \mathbf{q})^{(\Delta_\text{GP}^\text{*})}$, however, using a similar observation to that provided at the end of the proof of Theorem \ref{th:beta_A}, we can compute a fraction $\beta_\text{P} \le \beta$ based on an upper bound on this product. In detail, using \eqref{eq:muqextreme}, which represents this upper bound, our problem turns out to find $\beta_\text{P}$ such that
\begin{align}
\nonumber & m_3 (1-L_\text{P}\theta) \left[\sum\limits_{k \in \mathcal{L}_\text{P} }   d R e^{-\frac{ \sigma^2 \tau }{\alpha_{kk}}} \,  q_k  \prod\limits_{i \in  \mathcal{L}_\text{P},i\ne k}(1-g_{ki}) \right] \, \le \\  & \beta_\text{P}^{-1} (1-L\theta) \left[ \sum\limits_{k \in \mathcal{L} } d R e^{-\frac{ \sigma^2 \tau }{\alpha_{kk}}} \,   q_k  \prod\limits_{i \in \mathcal{L} ,i\ne k}(1-g_{ki}) \right].
\end{align} 
Using the relation in \eqref{eq:ineq_per}, it suffices to have $\beta_\text{P}^{-1} \ge m_3$, or equivalently $\beta_\text{P} \le m_3^{-1}$, to satisfy the above inequality. By taking $\beta_\text{P} = m_3^{-1}$, the desired result holds. 
\end{proof} 

An important factor on which fraction $\beta_\text{P}$ depends is the number of quantization bits $B$, so it is essential to compute the number of bits that can guarantee this fraction. Finding the explicit relation that gives the number of bits in function of $\beta_\text{P}$ is a difficult task, however we can obtain the required result numerically. In detail, using the expression of $\beta_\text{P}$ given in the above theorem, we start from a small value of $B$ for which we calculate the corresponding fraction, then we keep increasing $B$ until the desired value of $\beta_\text{P}$ is obtained. Although computing the exact relation of $B$ in function of $\beta_\text{P}$ is hard to achieve, we can still find a relation that gives a rough idea of the required number of bits. Specifically, we know that $1-g_{ki}= (1+ 2^{-\frac{B}{Q}}  c_{ki} )^{-1}$, where $c_{ki}=\frac{\zeta_{ki} \tau d }{ \zeta_{kk}}$,
then, after selecting the set $\mathcal{L}_\text{P}$ (of cardinality $L_\text{P}$) and $k$ ($\in \mathcal{L}_\text{P}$) that yield $\beta_\text{P}$, we find $c= \underset{i \in \mathcal{L}_\text{P},i \ne k }{\min} c_{ki}$. Thus, we get 
\begin{align}
\boxed{\beta_\text{P} \le \left( 1+ 2^{-\frac{B}{Q}} c \right)^{-(L_\text{P}-1)}} 
\end{align}
or equivalently we obtain
\begin{align}
\boxed{B \ge Q   \log_2 \left( c \left( \beta_\text{P}^{-(L_\text{P}-1)^{-1}}-1 \right)^{-1} \right) } \label{eq:LB_B}
\end{align}
Therefore, it suffices to use a number of quantization bits equals to the lower bound in the above inequality to guarantee the fraction $\beta_\text{P}$.
Note that the exact number of bits, given by the numerical method, is less than the calculated lower bound.

\section{Numerical Results}
\label{sec:numerical}

In this section we present our numerical results. We consider a system where the number of antennas $N_\text{t}=N_\text{r}=7$, $P=10$, $\sigma=1$, $d = 2$, $\theta=0.01$. We take $N=6$, which satisfies the condition $N_\text{t}+N_\text{r} \ge (N+1)d$. In addition, we assume that all the users have Poisson incoming traffic with the same average arrival rates as $a_k=a$. A 
coding scheme with a rate of $1$ bits per channel use if the \ac{SINR} of a scheduled user exceeds $\tau$ is assumed.  We set the slot duration to be $T_s=1000$ channel uses. 
Thus, we have $R = 1000$ bits per slot.
Even though in practice all the path loss coefficients are different, we consider in this section a very special case that simplifies the simulations and can still provide insights on the comparison between \ac{IA} and \ac{TDMA}-\ac{SVD}. In detail, we assume that all the direct links have a path loss coefficient of $1$ and all the cross links have a path loss coefficient of $\zeta_\text{c}$ (with $\zeta_\text{c} \le 1$). This setting allows us to examine, with respect to parameter $\zeta_\text{c}$, the impact of the cross links (or equivalently, the impact of interference) on the stability performances of \ac{IA}, and it let us detect when this latter technique outperforms \ac{TDMA}-\ac{SVD} in terms of stability and vice versa. 
To show the stability performance of the considered system, we plot the total average queue length given by $\frac{1}{M_\text{s}} \sum_{t=0}^{M_\text{s}-1}  \sum_{k=1 }^{ N} q_k(t)$ for different values of $a$, where each simulation lasts $M_\text{s}$ timeslots. We set $M_\text{s}=10^5$. Note that the point where the total average queue length function increases very steeply is the point at which the system becomes unstable.
%
%
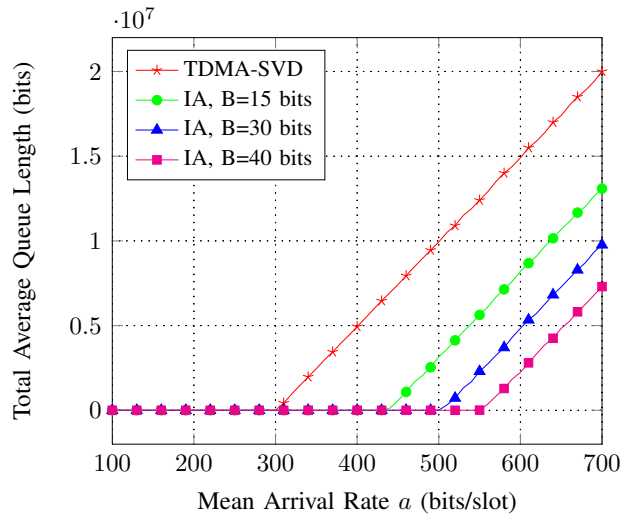
\begin{figure}[!ht]
\centering
\begin{tikzpicture}[scale=0.95]
	\begin{axis}[
 		grid = major,
 		legend cell align=left,
 		mark repeat={3},
	    xmin=100,xmax=700,
 		legend style ={legend pos=north west},
 		xlabel={Mean Arrival Rate $a$ (bits/slot)},
 		ylabel={Total Average Queue Length (bits)},
 		cycle list name = laneasIACaching1]
 	
 		\addplot [color=red, mark=star,mark size=2.2] table[col sep=comma] 
 	    	{\string"stab_svd.csv"};
 	    	\addlegendentry{TDMA-SVD}; 
 	    
 	    \addplot [color=green, mark=*,mark size=1.9] table [col sep=comma]
 	     	     {\string"stab_15b_0.2zeta.csv"};
 	     	     \addlegendentry{IA, B=15 bits};
 	     	    
 	    \addplot [color=blue, mark=triangle*,mark size=2.5] table [col sep=comma]
 	     	     {\string"stab_30b_0.2zeta.csv"};
 	     	     \addlegendentry{IA, B=30 bits};  	
 	     	     
         \addplot [color=magenta, mark=square*,mark size=1.7] table [col sep=comma]
 	    	     {\string"stab_40b_0.2zeta.csv"};
 	     	     \addlegendentry{IA, B=40 bits};  
 	    	 	     	     
	\end{axis}
\end{tikzpicture}
\captionsetup{font=small}
\caption{Total average queue length vs. mean arrival rate $a$. Here $\zeta_\text{c}=0.2$ and $\tau=1$. }
\label{fig:Totavqueue}
\end{figure}
\begin{figure}[!ht]
\centering
\begin{tikzpicture}[scale=0.95]
	\begin{axis}[
 		grid = major,
 		legend cell align=left,
 		mark repeat={1},
	    xmin=200,xmax=400,	
 		legend style ={legend pos=north west},
 		xlabel={Mean Arrival Rate $a$ (bits/slot)},
 		ylabel={Total Average Queue Length (bits)},
 		cycle list name = laneasIACaching1]

 		\addplot [color=red, mark=star,mark size=2.2] table[col sep=comma] 
  		{\string"stab_svd.csv"};
  		\addlegendentry{TDMA-SVD}; 
  			 
 	    \addplot [color=green, mark=*,mark size=1.9] table [col sep=comma]
 	    {\string"stab_15b_0.5zeta.csv"};
 	    \addlegendentry{IA, B=15 bits};
 	    
 	    \addplot [color=blue, mark=triangle*,mark size=2.5] table [col sep=comma]
 	    {\string"stab_30b_0.5zeta.csv"};
 	    \addlegendentry{IA, B=30 bits};
 	    
        \addplot [color=magenta, mark=square*,mark size=1.7] table [col sep=comma]
 	    {\string"stab_40b_0.5zeta.csv"};
 	    \addlegendentry{IA, B=40 bits}; 	    
 	    
	\end{axis}
\end{tikzpicture}
\captionsetup{font=small}
\caption{ Total average queue length vs. mean arrival rate $a$. Here $\zeta_\text{c}=0.5$ and $\tau=1$.}
\label{fig:Totavqueue0.5}
\end{figure}
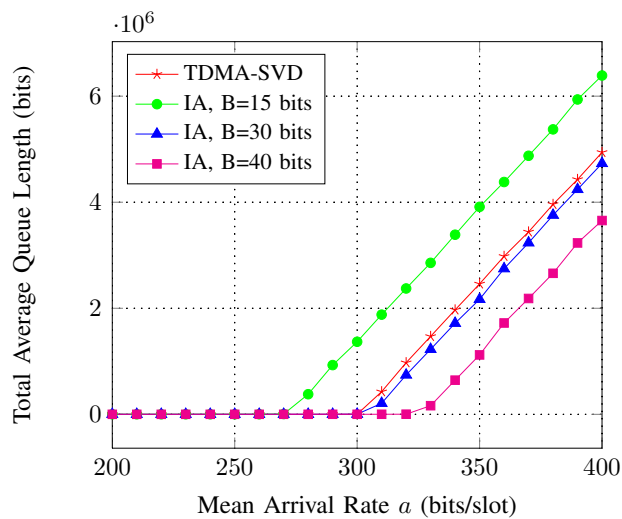
\begin{figure}[!ht]
\centering
\begin{tikzpicture}[scale=0.95]
	\begin{axis}[
 		grid = major,
 		legend cell align=left,
 		mark repeat={2},
	    xmin=10,xmax=40,	
 		legend style ={legend pos=north west},
 		xlabel={Number of bits $B^\prime$},
 		ylabel={Achievable fraction $\frac{r(N,B^\prime)}{r(N,B)}$},
 		cycle list name = laneasIACaching1]

 		\addplot [color=red, mark=*,mark size=1.9] table[col sep=comma] 
 		{\string"frac_BB1_0.5tau.csv"};
 		\addlegendentry{$\tau=0.5$};
 		
	    \addplot [color=green, mark=triangle*,mark size=2.5] table[col sep=comma] 
 		{\string"frac_BB1_1tau.csv"};
 		\addlegendentry{$\tau=1$}; 	
 			
	    \addplot [color=blue, mark=square*,mark size=1.7] table[col sep=comma] 
 		{\string"frac_BB1_1.5tau.csv"};
 		\addlegendentry{$\tau=1.5$}; 			 
 	    \addplot [color=magenta, mark=diamond*,mark size=2.5] table [col sep=comma]
 	    {\string"frac_BB1_2tau.csv"};
 	    \addlegendentry{$\tau=2$};
 	   	
	\end{axis}
\end{tikzpicture}
\captionsetup{font=small}
\caption{Achievable fraction $\frac{r(N,B^\prime)}{r(N,B)}$ vs. number of bits $B^\prime$. Here $\zeta_\text{c}=1$ and $B=40$ bits.}
\label{fig:fracBB'}
\end{figure}
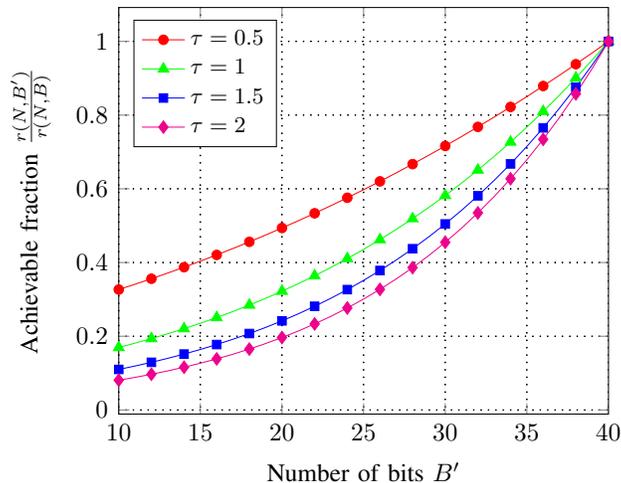  
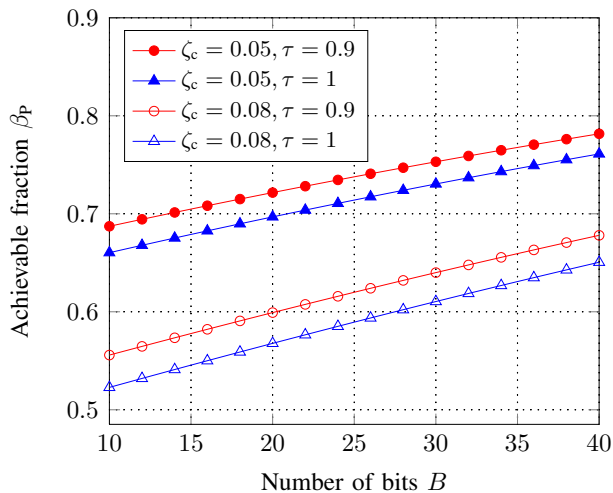
\begin{figure}[!ht]
\centering
\begin{tikzpicture}[scale=0.95]
	\begin{axis}[
 		grid = major,
 		legend cell align=left,
 		mark repeat={2},
 		ymax=0.9,
	    xmin=10,xmax=40,	
 		legend style ={legend pos=north west},
 		xlabel={Number of bits $B$},
 		ylabel={Achievable fraction $\beta_\text{P}$},
 		cycle list name = laneasIACaching1]

        \addplot [color=red, mark=*,mark size=1.9] table[col sep=comma] 
 		{\string"frac_BetaP_0.9tau0.05zeta.csv"};
 		\addlegendentry{$\zeta_\text{c}=0.05,\tau=0.9$}; 	
 	
 		\addplot [color=blue, mark=triangle*,mark size=2.5] table[col sep=comma] 
 		{\string"frac_BetaP_1tau0.05zeta.csv"};
 		\addlegendentry{$\zeta_\text{c}=0.05,\tau=1$};

        \addplot [color=red, mark=o,mark size=1.9] table[col sep=comma] 
 		{\string"frac_BetaP_0.9tau0.08zeta.csv"};
 		\addlegendentry{$\zeta_\text{c}=0.08,\tau=0.9$}; 	
 		
	    \addplot [color=blue, mark=triangle,mark size=2.5] table[col sep=comma] 
 		{\string"frac_BetaP_1tau0.08zeta.csv"};
 		\addlegendentry{$\zeta_\text{c}=0.08,\tau=1$};

	\end{axis}
\end{tikzpicture}
\captionsetup{font=small}
\caption{Achievable fraction $\beta_\text{P}$ vs. number of bits $B$.}
\label{fig:frac_BetaP}
\end{figure}  

Fig. \ref{fig:Totavqueue} shows that \ac{IA} gives better performances when we increase the number of quantization bits. This is due to the fact that the more the quantization is precise, the more we achieve higher rates which implies better stability performances.
From Fig. \ref{fig:Totavqueue} and \ref{fig:Totavqueue0.5}, we can see that \ac{TDMA}-\ac{SVD} outperforms \ac{IA} when the interference impact is high (for instance with $\zeta_\text{c}=0.5$ and $B=15$ bits), whereas we obtain the converse for less interfering system (for instance with $\zeta_\text{c}=0.2$). This is due to the fact that in high interference scenarios, \ac{IA} needs a better \ac{CSI} knowledge in order to maintain a good alignment of interference, and this can be provided by using a large number of bits in the quantization process.
It is worth noting that there exist other parameters that may affect the comparison between \ac{TDMA}-\ac{SVD} and \ac{IA}, such as the number of antennas, the threshold $\tau$, the number of data streams, etc. Figure \ref{fig:fracBB'} depicts the variation of the fraction $\frac{r(N,B^\prime)}{r(N,B)}$ with the number of bits $B^\prime$, for different values of $\tau$ and for a fixed reference number of bits $B=40$ bits; here, we set $\zeta_\text{c}=1$ since $\frac{r(N,B^\prime)}{r(N,B)}$ is defined for the symmetric system. It is clear from this figure that increasing the number of quantization bits and/or decreasing the threshold $\tau$ result in higher achievable fractions. Also, one can notice that changing (increasing or decreasing) the number of quantization bits has a higher impact on the achievable fraction for greater values of $\tau$, meaning that the more the threshold is high, the more the fraction $\frac{r(N,B^\prime)}{r(N,B)}$ is sensitive to the variation of the number of bits. In Figure \ref{fig:frac_BetaP}, we illustrate the variation of fraction $\beta_\text{P}$ with the number of bits $B$, for different values of $\tau$ and $\zeta_\text{c}$. The plots in this figure confirm the expectation that the stability region in the imperfect case gets bigger, meaning that the fraction this stability region achieves with respect to the stability region in the perfect case is greater, if the system achieves higher transmission rates. Note that these (higher) rates result from greater $B$, lower $\zeta_\text{c}$ and/or lower $\tau$.                     
\section{Conclusion}
\label{sec:conclusions}

In this paper, we characterized the stability region for \ac{IA} in a \ac{MIMO} interference network under \ac{TDD} mode with limited backhaul capacity and taking into account the probing cost. Specifically, this characterization was provided for the symmetric system
and for the general system. Also, for each one of these scenarios, we characterized the stability region under the prefect case (i.e. unlimited backhaul), and we captured the gap between this region and the one achieved under the imperfect case (i.e. limited backhaul). In addition, under the different considered cases and scenarios, an optimal centralized scheduling policy that achieves the system stability region was provided. We noticed that this scheduling policy can be implemented with a reduced complexity for the symmetric system, whereas under the general system the high computational complexity of this policy leads us to propose an approximate policy that has a reduced complexity but that achieves only a fraction of the system stability region. A characterization of this achievable fraction was given. Furthermore, under the symmetric system,
we characterized the system stability region when using \ac{TDMA}-\ac{SVD} instead of \ac{IA}, we compared the stability regions of these two techniques and, using the result of this comparison, we provided a condition under which one of these two techniques outperforms the other in terms of stability. 
Finally, we showed that, under some conditions, we can achieve better stability results (i.e. bigger stability region) by deciding to switch between these two techniques.
 
Important extensions can be addressing the stability analysis when we adopt decentralized or even mixed (centralized + decentralized) methods for feedback and scheduling.

\section*{Acknowledgments}
This research has been partially funded by Huwaei and "Fondation Supélec".
The work of M. Debbah has been funded by the ERC Grant 305123 MORE (Advanced Mathematical Tools for Complex Network Engineering).

\appendices

\section{Proof of Lemma \ref{le:ratevariation}}
\label{app:ratevariation}

We start the proof by first showing that $r(L)$ decreases with $L$. 
The first derivative of this rate function is given by 
\begin{align} 
\frac{dr}{dL} = d R e^{-\frac{ \sigma^2 \tau }{\alpha}} \left(-\theta   + (1-L\theta)\log F   \right)  F^{L-1}.    
\end{align} 
Since we have $L<\frac{1}{\theta}$ and $\log F <0$, the first derivative is negative and so $r$ decreases with $L$. 

To study the variation of $r_\text{T}(L)$ (w.r.t. $L$) we need to first compute its first derivatives, which will help us determine the optimal number of pairs.
The first derivative can be written as 
\begin{align} 
\frac{dr_{\text{T}}}{dL} = d R e^{-\frac{ \sigma^2 \tau }{\alpha}}   \left( -L^2\theta \log F + L(-2\theta+\log F)+1  \right) F^{L-1} .   
\end{align} 
Setting $\frac{dr_{\text{T}}}{dL}=0$ yields 
\begin{align} 
-L^2\theta \log F + L(-2\theta+\log F)+1 =0,
\end{align} 
or equivalently 
\begin{align} 
L^2 + L \left( \frac{2}{\log F}-\frac{1}{\theta} \right) -\frac{1}{\theta \log F} =0,
\end{align} 
We can easily show that the only zeros of $\frac{dr_{\text{T}}}{dL}$ are at 
\begin{align}
L_0 = \frac{ \frac{1}{\theta} - \frac{2}{\log F} - \sqrt{ \left(   \frac{2}{\log F}  - \frac{1}{\theta}  \right)^2 + \frac{4}{ \theta \log F}   }  }{2}  \label{eq:L0},
\end{align}
\begin{align}
L_1 = \frac{ \frac{1}{\theta} - \frac{2}{\log F} + \sqrt{ \left(   \frac{2}{\log F}  - \frac{1}{\theta}  \right)^2 + \frac{4}{ \theta \log F}   }  }{2}  \label{eq:L1}.
\end{align}
Note that $\log F<0$ and $\left(   \frac{2}{\log F}  - \frac{1}{\theta}  \right)^2 + \frac{4}{ \theta \log F}=\frac{1}{\theta^2}+\frac{4}{(\log F)^2}$.
Let us now examine the feasibility of $L_0$ and $L_1$. Indeed, under our setting a number $L$ is feasible if it satisfies $ 0<L < \frac{1}{\theta}$. For $L_0$ we have
\begin{align} 
L_0 = \frac{ \frac{1}{\theta} - \frac{2}{\log F} - \sqrt{  \frac{1}{\theta^2}+\frac{4}{(\log F)^2}  }  }{2} < \frac{ \frac{1}{\theta} - \frac{2}{\log F} -  \frac{2}{\left| \log F  \right|} } {2}= \frac{1}{2 \theta},
\end{align} 
where the inequality results from the fact that $\frac{2}{\left| \log F  \right|}  < \sqrt{  \frac{1}{\theta^2}+\frac{4}{(\log F)^2}  } $.
We can also observe that  
\begin{align} 
L_0 = \frac{ \frac{1}{\theta} - \frac{2}{\log F} - \sqrt{  \frac{1}{\theta^2}+\frac{4}{(\log F)^2}  }  }{2}  > \frac{ \frac{1}{\theta} - \frac{2}{\log F} -  \frac{1}{\theta}- \frac{2}{\left| \log F  \right|} } {2}=0.
\end{align} 
Thus, $L_0$ is a feasible solution since $ 0<L_0 < \frac{1}{\theta}$. On the other hand, for $L_1$ we can notice that
\begin{align} 
L_1 = \frac{ \frac{1}{\theta} - \frac{2}{\log F} + \sqrt{  \frac{1}{\theta^2}+\frac{4}{(\log F)^2}  }  }{2} > \frac{ \frac{1}{\theta} +  \sqrt{\frac{1}{\theta^2}} } {2}= \frac{1}{\theta}.
\end{align} 
Hence, $L_1$ is not a feasible solution because $L_1 > \frac{1}{\theta} $. To complete the proof it suffices to show that $r_\text{T}(L)$ reaches its maximum at $L_0$.  To this end, we note that $r_\text{T}(0)=0$, $r_\text{T}(\frac{1}{\theta})=0$ and $\frac{dr_{\text{T}}}{dL}|_{L= \frac{1}{2\theta} } <0$, and we recall that $0<L_0<\frac{1}{2\theta}<\frac{1}{\theta}$. In addition, one can easily notice that $r_\text{T}$ and its first derivative ($\frac{dr_{\text{T}}}{dL}$) are continuous over $\left[ 0, \frac{1}{\theta} \right]$. Based on these observations, the variation of $r_\text{T}$ over $\left[ 0, \frac{1}{\theta} \right]$  can be described as follows: $r_\text{T}$ is increasing from $0$ to $L_0$ and decreasing from $L_0$ to $\frac{1}{\theta}$. This concludes the proof.

\section{Examples and Illustrations}
\label{app:examples}
\begin{figure}[H]
\centering
\includegraphics[width=0.7\linewidth, height=7.4cm]{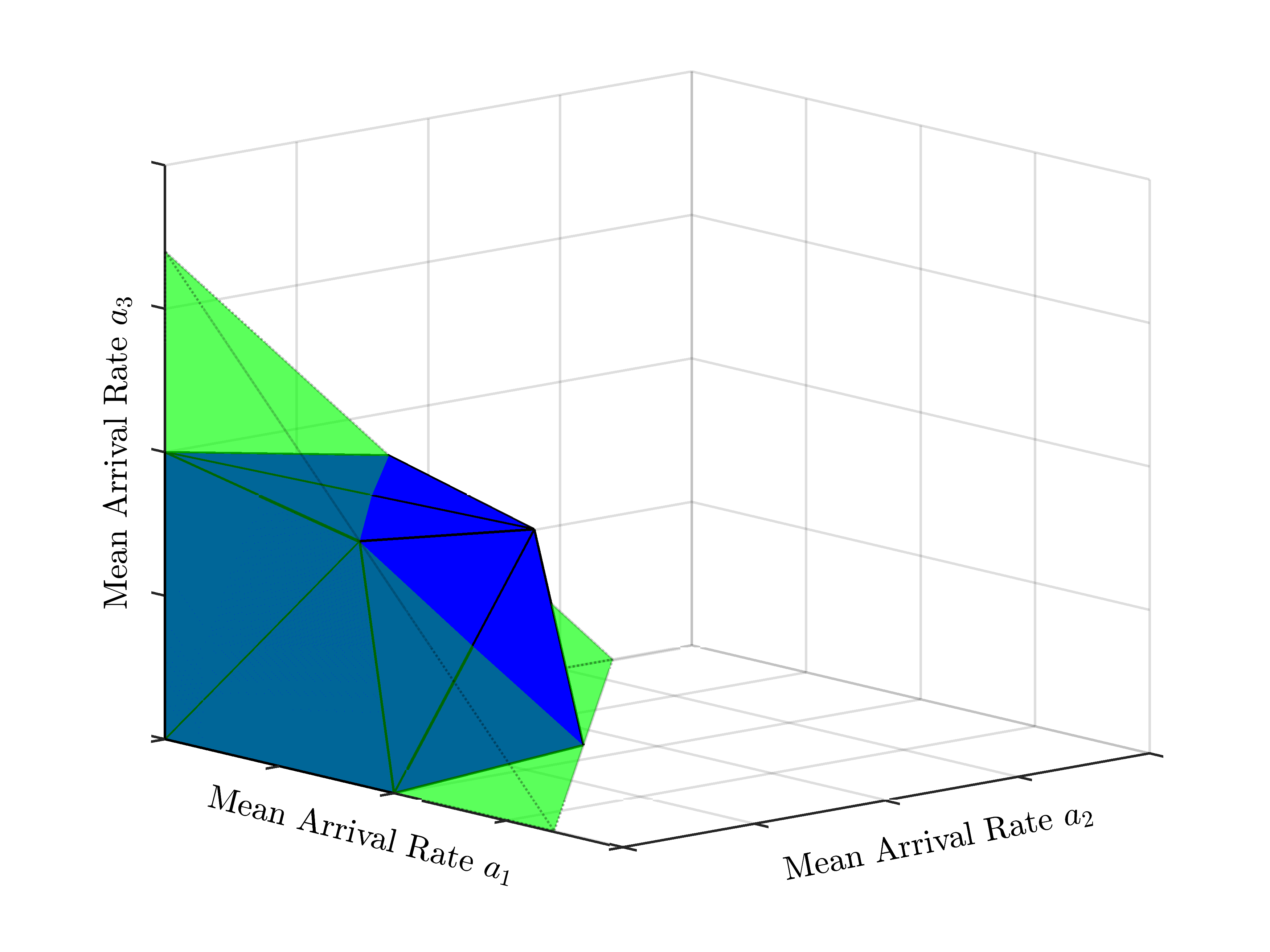}
\captionsetup{font=small}
\caption{Stability regions of \ac{TDMA}-\ac{SVD} (green region) and \ac{IA} under the imperfect case (blue region) for the symmetric system, where $L_\text{I}=N=3$. This illustration represents the case where \ac{IA} outperforms \ac{TDMA}-\ac{SVD}, in which the blue region surpasses (partially) the green region.
Note that a similar illustration can be given to compare between \ac{IA} under the perfect case and \ac{TDMA}-\ac{SVD}.
}
\label{fig:2}
\end{figure} 
\vspace{-0.8cm}
\begin{figure}[H]
\centering
\includegraphics[width=0.7\linewidth, height=7.4cm]{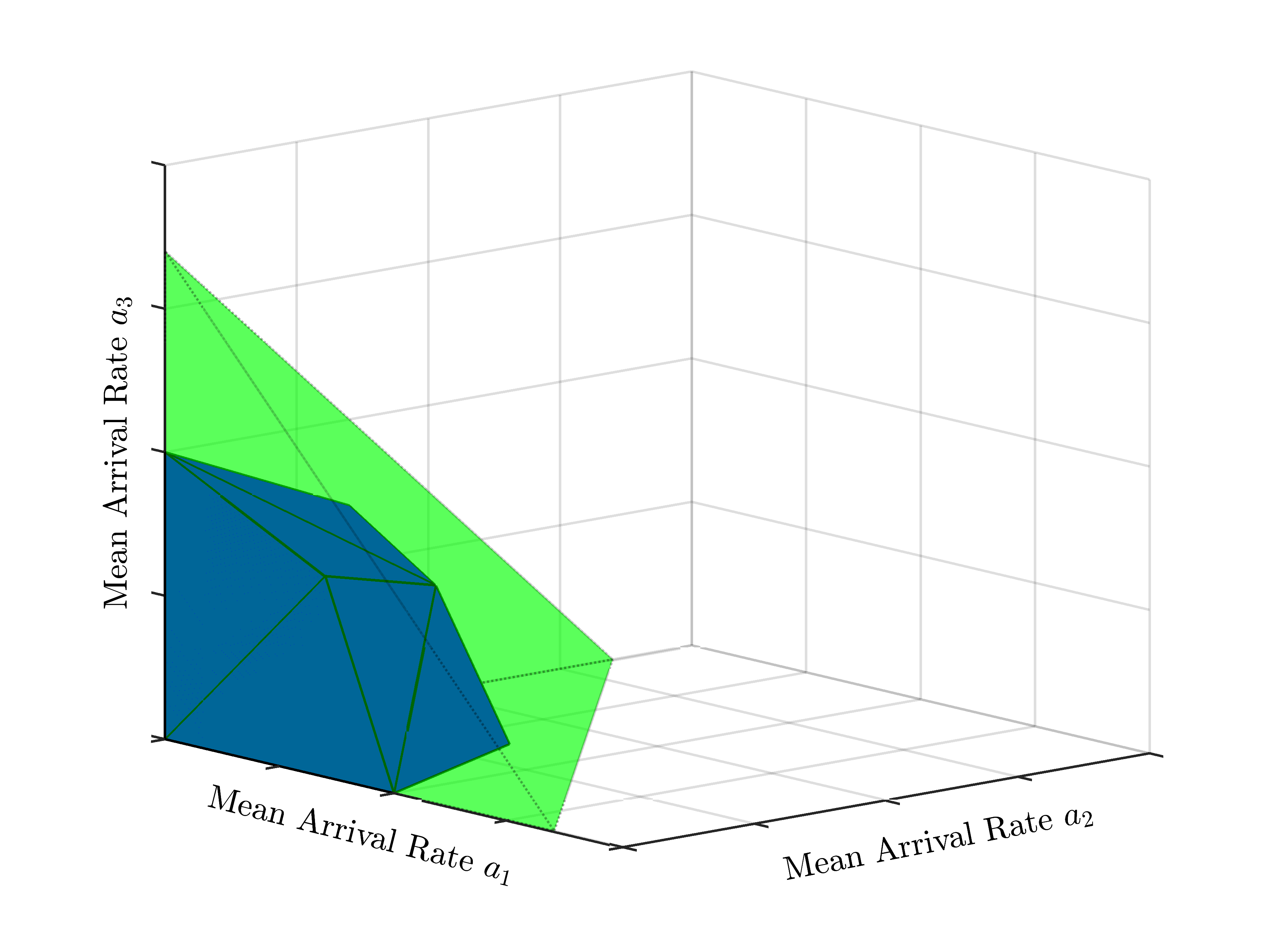}
\captionsetup{font=small}
\caption{Stability regions of \ac{TDMA}-\ac{SVD} (green region) and \ac{IA} under the imperfect case (blue region) for the symmetric system, where $L_\text{I}=N=3$. This illustration represents the case where \ac{TDMA}-\ac{SVD} outperforms \ac{IA}, in which the green region covers the blue region.
Note that a similar illustration can be given to compare between \ac{IA} under the perfect case and \ac{TDMA}-\ac{SVD}.
}
\label{fig:2}
\end{figure}

\bibliographystyle{IEEEtran}
\bibliography{references}

\begin{thebibliography}{10}
\providecommand{\url}[1]{#1}
\csname url@samestyle\endcsname
\providecommand{\newblock}{\relax}
\providecommand{\bibinfo}[2]{#2}
\providecommand{\BIBentrySTDinterwordspacing}{\spaceskip=0pt\relax}
\providecommand{\BIBentryALTinterwordstretchfactor}{4}
\providecommand{\BIBentryALTinterwordspacing}{\spaceskip=\fontdimen2\font plus
\BIBentryALTinterwordstretchfactor\fontdimen3\font minus
  \fontdimen4\font\relax}
\providecommand{\BIBforeignlanguage}[2]{{%
\expandafter\ifx\csname l@#1\endcsname\relax
\typeout{** WARNING: IEEEtran.bst: No hyphenation pattern has been}%
\typeout{** loaded for the language `#1'. Using the pattern for}%
\typeout{** the default language instead.}%
\else
\language=\csname l@#1\endcsname
\fi
#2}}
\providecommand{\BIBdecl}{\relax}
\BIBdecl

\bibitem{Deghel2015StabilityIA}
M.~Deghel, M.~Assaad, and M.~Debbah, ``Queueing stability and {CSI} probing of
  a {TDD} wireless network with interference alignment,'' in \emph{Information
  Theory (ISIT), 2015 IEEE International Symposium on}, June 2015, pp.
  794--798.

\bibitem{Cadambe2008Interference}
V.~R. Cadambe and S.~A. Jafar, ``Interference alignment and degrees of freedom
  of the $k$-user interference channel,'' \emph{IEEE Trans. Inform. Theory},
  vol.~54, no.~8, pp. 3425--3441, August 2008.

\bibitem{Thukral2009LimIA}
\BIBentryALTinterwordspacing
J.~Thukral and H.~B{\"{o}}lcskei, ``Interference alignment with limited
  feedback,'' \emph{CoRR}, vol. abs/0905.0374, 2009. [Online]. Available:
  \url{http://arxiv.org/abs/0905.0374}
\BIBentrySTDinterwordspacing

\bibitem{Krishnamachari2010LimIA}
R.~Krishnamachari and M.~Varanasi, ``Interference alignment under limited
  feedback for {MIMO} interference channels,'' in \emph{Information Theory
  Proceedings (ISIT), 2010 IEEE International Symposium on}, June 2010, pp.
  619--623.

\bibitem{Rezaee2012LimIA}
M.~Rezaee and M.~Guillaud, ``Limited feedback for interference alignment in the
  {K}-user {MIMO} interference channel,'' in \emph{Information Theory Workshop
  (ITW), 2012 IEEE}, Sept 2012, pp. 667--671.

\bibitem{Jindal2006MIMO}
N.~Jindal, ``{M}{I}{M}{O} broadcast channels with finite-rate feedback,''
  \emph{Information Theory, IEEE Transactions on}, vol.~52, no.~11, pp.
  5045--5060, Nov 2006.

\bibitem{Santipach2009Capacity}
W.~Santipach and M.~Honig, ``Capacity of a multiple-antenna fading channel with
  a quantized precoding matrix,'' \emph{Information Theory, IEEE Transactions
  on}, vol.~55, no.~3, pp. 1218--1234, March 2009.

\bibitem{Ayach2012Interference}
O.~E. Ayach and R.~W. Heath, ``Interference alignment with analog channel state
  feedback,'' \emph{IEEE Trans. Wireless Commun.}, vol.~11, no.~2, pp.
  626--636, February 2012.

\bibitem{Tresch2009Cellular}
R.~Tresch and M.~Guillaud, ``Cellular interference alignment with imperfect
  channel knowledge,'' in \emph{Communications Workshops, 2009. ICC Workshops
  2009. IEEE International Conference on}, June 2009, pp. 1--5.

\bibitem{Park2013PerCloud}
\BIBentryALTinterwordspacing
S.~Park, O.~Simeone, O.~Sahin, and S.~Shamai, ``Performance evaluation of
  multiterminal backhaul compression for cloud radio access networks,''
  \emph{CoRR}, vol. abs/1311.6492, 2013. [Online]. Available:
  \url{http://arxiv.org/abs/1311.6492}
\BIBentrySTDinterwordspacing

\bibitem{Rezaee2013CSIT}
M.~Rezaee, M.~Guillaud, and F.~Lindqvist, ``{C}{S}{I}{T} sharing over finite
  capacity backhaul for spatial interference alignment,'' in \emph{IEEE
  International Symposium on Information Theory Proceedings (ISIT'13)}.\hskip
  1em plus 0.5em minus 0.4em\relax IEEE, July 2013, pp. 569--573.

\bibitem{Boche2006TheInterplay}
H.~Boche and M.~Wiczanowski, ``The interplay of link layer and physical layer
  under {MIMO} enhancement: benefits and challenges,'' \emph{Wireless
  Communications, IEEE}, vol.~13, no.~4, pp. 48--55, Aug 2006.

\bibitem{McKeown1999Achieving100}
N.~McKeown, A.~Mekkittikul, V.~Anantharam, and J.~Walrand, ``Achieving 100
  percent throughput in an input-queued switch,'' \emph{Communications, IEEE
  Transactions on}, vol.~47, no.~8, pp. 1260--1267, Aug 1999.

\bibitem{Kumar1996Duality}
P.~Kumar and S.~Meyn, ``Duality and linear programs for stability and
  performance analysis of queuing networks and scheduling policies,''
  \emph{Automatic Control, IEEE Transactions on}, vol.~41, no.~1, pp. 4--17,
  Jan 1996.

\bibitem{Leonardi2001Bounds}
E.~Leonardi, M.~Mellia, F.~Neri, and M.~Ajmone~Marsan, ``Bounds on average
  delays and queue size averages and variances in input-queued cell-based
  switches,'' in \emph{INFOCOM 2001. Twentieth Annual Joint Conference of the
  IEEE Computer and Communications Societies. Proceedings. IEEE}, vol.~2, 2001,
  pp. 1095--1103 vol.2.

\bibitem{Szpankowski1993StabilityConditions}
W.~Szpankowski, ``Stability conditions for some distributed systems: Buffered
  random access systems,'' \emph{Buffered Random Access Systems, Adv. Appl.
  Probab}, vol.~26, pp. 498--515, 1993.

\bibitem{Tassiulas1992Stability}
L.~Tassiulas and A.~Ephremides, ``Stability properties of constrained queueing
  systems and scheduling policies for maximum throughput in multihop radio
  networks,'' \emph{Automatic Control, IEEE Transactions on}, vol.~37, no.~12,
  pp. 1936--1948, Dec 1992.

\bibitem{Neely2003Power}
M.~Neely, E.~Modiano, and C.~Rohrs, ``Power allocation and routing in multibeam
  satellites with time-varying channels,'' \emph{Networking, IEEE/ACM
  Transactions on}, vol.~11, no.~1, pp. 138--152, Feb 2003.

\bibitem{Boche2007Optimization}
H.~Boche and M.~Wiczanowski, ``Optimization-theoretic analysis of
  stability-optimal transmission policy for multiple-antenna multiple-access
  channel,'' \emph{Signal Processing, IEEE Transactions on}, vol.~55, no.~6,
  pp. 2688--2702, June 2007.

\bibitem{Swannack2004Lowcomplexity}
C.~Swannack, E.~Uysal-biyikoglu, and G.~Wornell, ``Low complexity multiuser
  scheduling for maximizing throughput,'' in \emph{{in the MIMO} Broadcast
  Channel,” in Proc. Allerton Conf. Commun., Contr., and Computing}, 2004.

\bibitem{Kobayashi2006AnIterativeWater}
M.~Kobayashi and G.~Caire, ``An iterative water-filling algorithm for maximum
  weighted sum-rate of gaussian mimo-bc,'' \emph{Selected Areas in
  Communications, IEEE Journal on}, vol.~24, no.~8, pp. 1640--1646, Aug 2006.

\bibitem{Cheng2006OptimalDown}
C.~Wang and R.~Murch, ``Optimal downlink multi-user mimo cross-layer scheduling
  using hol packet waiting time,'' \emph{Wireless Communications, IEEE
  Transactions on}, vol.~5, no.~10, pp. 2856--2862, Oct 2006.

\bibitem{Kobayashi2007TransmitDivers}
M.~Kobayashi, G.~Caire, and D.~Gesbert, ``Transmit diversity versus
  opportunistic beamforming in data packet mobile downlink transmission,''
  \emph{Communications, IEEE Transactions on}, vol.~55, no.~1, pp. 151--157,
  Jan 2007.

\bibitem{Kobayashi2007JointBeamforming}
M.~Kobayashi and G.~Caire, ``Joint beamforming and scheduling for a
  multi-antenna downlink with imperfect transmitter channel knowledge,''
  \emph{Selected Areas in Communications, IEEE Journal on}, vol.~25, no.~7, pp.
  1468--1477, September 2007.

\bibitem{Lau2012StabilityDelay}
K.~Huang and V.~Lau, ``Stability and delay of zero-forcing {SDMA} with limited
  feedback,'' \emph{Information Theory, IEEE Transactions on}, vol.~58, no.~10,
  pp. 6499--6514, Oct 2012.

\bibitem{Chaporkar2009ScedwithLimited}
\BIBentryALTinterwordspacing
P.~Chaporkar, A.~Proutiere, H.~Asnani, and A.~Karandikar, ``Scheduling with
  limited information in wireless systems,'' in \emph{Proceedings of the Tenth
  ACM International Symposium on Mobile Ad Hoc Networking and Computing}, ser.
  MobiHoc '09.\hskip 1em plus 0.5em minus 0.4em\relax New York, NY, USA: ACM,
  2009, pp. 75--84. [Online]. Available:
  \url{http://doi.acm.org/10.1145/1530748.1530759}
\BIBentrySTDinterwordspacing

\bibitem{Destounis2015Traffic-Aware}
A.~Destounis, M.~Assaad, M.~Debbah, and B.~Sayadi, ``Traffic-aware training and
  scheduling for {MISO} wireless downlink systems,'' \emph{Information Theory,
  IEEE Transactions on}, vol.~61, no.~5, pp. 2574--2599, May 2015.

\bibitem{Stefania2009LTEtheUMTS}
M.~B. Stefania~Sesia, Issam~Toufik, ``{LTE}, the {UMTS} long term evolution:
  From theory to practice.''\hskip 1em plus 0.5em minus 0.4em\relax Wiley,
  2009.

\bibitem{Emre1999Capacityof}
E.~T. Ar and I.~E. Telatar, ``Capacity of multi-antenna gaussian channels,''
  \emph{European Transactions on Telecommunications}, vol.~10, pp. 585--595,
  1999.

\bibitem{ElAyach2012OverheadofIA}
O.~El~Ayach, A.~Lozano, and R.~Heath, ``On the overhead of interference
  alignment: Training, feedback, and cooperation,'' \emph{Wireless
  Communications, IEEE Transactions on}, vol.~11, no.~11, pp. 4192--4203,
  November 2012.

\bibitem{Yetis2010Feasibility}
C.~M. Yetis, T.~Gou, S.~A. Jafar, and A.~H. Kayran, ``On feasibility of
  interference alignment in {M}{I}{M}{O} interference networks,'' \emph{IEEE
  Trans. Signal Processing}, vol.~58, no.~9, pp. 4771--4782, September 2010.

\bibitem{Georgiadis06resourceallocation}
L.~Georgiadis, M.~J, and R.~Tassiulas, ``Resource allocation and cross-layer
  control in wireless networks,'' in \emph{Foundations and Trends in
  Networking}, 2006, pp. 1--149.

\bibitem{Xiaoming2014PerformanceAnaly}
X.~Chen and C.~Yuen, ``Performance analysis and optimization for interference
  alignment over mimo interference channels with limited feedback,''
  \emph{Signal Processing, IEEE Transactions on}, vol.~62, no.~7, pp.
  1785--1795, April 2014.

\bibitem{Johannesson1995Approximations}
B.~J{\'o}hannesson and N.~Giri, ``On approximations involving the beta
  distribution,'' \emph{Communications in Statistics-Simulation and
  Computation}, vol.~24, no.~2, pp. 489--503, 1995.

\bibitem{Xiaoming2014PerformanceAnal}
X.~Chen and C.~Yuen, ``Performance analysis and optimization for interference
  alignment over {MIMO} interference channels with limited feedback,''
  \emph{Signal Processing, IEEE Transactions on}, vol.~62, no.~7, pp.
  1785--1795, April 2014.

\bibitem{Nadarajah2005ProductGammaBeta}
\BIBentryALTinterwordspacing
S.~Nadarajah and S.~Kotz, ``\BIBforeignlanguage{English}{On the product and
  ratio of gamma and beta random variables},''
  \emph{\BIBforeignlanguage{English}{Allgemeines Statistisches Archiv}},
  vol.~89, no.~4, pp. 435--449, 2005. [Online]. Available:
  \url{http://dx.doi.org/10.1007/s10182-005-0214-9}
\BIBentrySTDinterwordspacing

\bibitem{Bateman1954Tables}
H.~Bateman, A.~Erd{\'e}lyi, H.~van Haeringen, and L.~Kok, \emph{Tables of
  integral transforms}.\hskip 1em plus 0.5em minus 0.4em\relax McGraw-Hill New
  York, 1954, vol.~1.

\bibitem{Boyd2004}
S.~Boyd and L.~Vandenberghe, \emph{Convex Optimization}.\hskip 1em plus 0.5em
  minus 0.4em\relax New York, NY, USA: Cambridge University Press, 2004.

\end{thebibliography}

\end{document}